\documentclass[12pt,draftclsnofoot,onecolumn]{IEEEtran}
\usepackage{booktabs}
\usepackage{graphicx}
\usepackage[cmex10]{amsmath}
\usepackage{cases}
\usepackage[tight,footnotesize]{subfigure}
\usepackage{amsthm} 
\usepackage{amsmath}
\usepackage{cite}
\usepackage{amssymb}
\usepackage{algorithm}
\usepackage{algorithmic}
\usepackage{xcolor}
\usepackage{stfloats}
\usepackage{multirow}
\usepackage{CJK}
\usepackage{subeqnarray}
\usepackage{longtable}
\usepackage{multicol}
\usepackage{extpfeil}
\usepackage{verbatim}
\newtheorem{lemma}{\bf{Lemma}}
\newtheorem{theorem}{\bf{Theorem}}

\ifCLASSINFOpdf
\else
\fi

\begin{document}

\title{\LARGE Secure Massive RIS aided Multicast with Uncertain CSI: Energy-Efficiency Maximization via Accelerated First-Order Algorithms}

\author{Zongze Li, Shuai Wang, Miaowen Wen, and Yik-Chung Wu
	
\thanks{Zongze Li and Yik-Chung Wu are with the Department of Electrical and
		Electronic Engineering, The University of Hong Kong, Hong Kong (e-mail:
		zzli@eee.hku.hk; ycwu@eee.hku.hk). Shuai Wang is with the Department of Electrical and Electronic Engineering, Southern University of Science and Technology, Shenzhen 518055, China (e-mail: wangs3@sustech.edu.cn). Miaowen Wen is with the School of Electronics and Information Engineering, South China University of Technology, Guangzhou 510640, China (e-mail: eemwwen@scut.edu.cn).}
}

\maketitle

\vspace{-0.8 cm} 
\begin{abstract}
Reconfigurable intelligent surface (RIS) has the potential to significantly enhance the network secure transmission performance by reconfiguring the wireless propagation environment.
However, due to the passive nature of eavesdroppers and the cascaded channel brought by the RIS, the eavesdroppers' channel state information is imperfectly obtained at the base station. Under the channel uncertainty, the optimal phase-shift, power allocation, and transmission rate design for secure transmission is currently unknown due to the difficulty of handling the probabilistic constraint with coupled variables. To fill this gap, this paper formulates a problem of energy-efficient secure transmission design while incorporating the probabilistic constraint. By transforming the probabilistic constraint and decoupling the variables, the secure energy efficiency maximization problem can be solved via alternatively executing concave-convex procedure and semidefinite relaxation technique. 
To scale the solution to massive antennas and reflecting elements scenario, an accelerated first-order algorithm with low complexity is further proposed.
Simulation results show that the proposed accelerated first-order algorithm achieves identical performance to the conventional method but saves at least two orders of magnitude in computation time. Moreover, the resultant RIS aided secure transmission significantly improves the energy efficiency compared to baseline schemes of random phase-shift, fixed phase-shift, and RIS ignoring CSI uncertainty. 
\end{abstract}



\IEEEpeerreviewmaketitle

\section{Introduction}
With huge demand on transmission rate in the era of big data, energy consumption becomes a serious concern for future wireless networks, and energy efficiency (EE) is a key consideration in wireless system design. The recently introduced reconfigurable intelligent surface (RIS) emerges as a promising technology for improving the EE of wireless systems via reconfiguring the signal propagation environment~\cite{J_IRS_EE19Huang,J_alendos20reconfigurable,J_Alexa21DMA}. Due to the passive nature and programmability of the RIS, the power consumption and added thermal noise during reflection are extremely low. Accordingly, RIS only costs a small amount of energy but could significantly improve the quality-of-service of users who suffer from unfavourable propagation conditions. 
As a consequence, it has the potential for significantly improving EE and enabling green communications~\cite{J_Wu17Review5G}.

On the other hand, RIS can also provide a new level of physical layer security. Specifically, when the legitimate receivers and the eavesdroppers are in the same directions to the base station (BS), the channel responses of the legitimate receivers will be highly correlated with those of the eavesdroppers.
This makes traditional beamforming, which directs energy toward legitimate receivers, also benefits eavesdroppers.
Hence, it is difficult to guarantee the security with the use of beamforming only at the transceivers. 
Fortunately, the employment of the RIS provides more degrees of freedom for additional transmission links to the legitimate receivers while nulling the directions towards the eavesdroppers, thus reducing the information leakage~\cite{J_alexandropoulos20safeguarding}.

Pioneering works on RIS aided transmission with security consideration assume perfect knowledge of the eavesdroppers' channels and the eavesdroppers are treated as unscheduled active users in the network~\cite{J_Chen19IRS_Secure}.
Under this assumption, it is shown that secure transmission with RIS achieves a higher secrecy rate than the transmission with random phase-shift or fixed phase-shift matrix~\cite{J_Chen19IRS_Secure,J_Chu21IRSecure,J_Chu21IRS_perfectCSI}.
Although these results are encouraging, the assumption on perfect knowledge of the eavesdroppers' channels is too strong in practice, especially when there are multiple non-colluding eavesdroppers in the system. Even the eavesdroppers are unscheduled active users in the network, due to the cascaded channel brought by RIS, the channel state information (CSI) estimation and acquisition is more challenging than that in a conventional communication system~\cite{J_He20CascadeIFR}. Although channel estimation schemes tailored to the RIS system such as matrix quantization or Hadamard-matrix truncation have been recently proposed~\cite{J_YouZRui20_IRS}, assuming perfect CSI of eavesdroppers is still far from realistic.  
Due to the uncertainty in the eavesdroppers' CSI, the outage of secure transmission should be considered.

To ensure the outage probability is within tolerable level, the secure transmission design in this paper incorporates a probabilistic constraint, which unfortunately is challenging for further analysis. Moreover, since secure EE (defined as the ratio of the secrecy rate to the total power consumption)~\cite{J_Zheng18SecureEE} is used as the objective function, the transmission design belongs to the more challenging problem of fractional programs.
To handle the above challenges, this paper first transforms the intractable probabilistic constraint into a deterministic one by leveraging the exponential distribution property of the received signal power~\cite{B_S_Prin02}.
Then, the resultant problem can be decomposed into two subproblems and iteratively solved via block coordinate descent approach, where the interior-point method with concave-convex procedure (CCP)~\cite{J_Lipp16CCP} and semidefinite relaxation (SDR) with Gaussian randomization procedure~\cite{J_LuoSDR} are respectively used to handle each subproblem.

Although the above algorithm provides a workable solution, it does not scale well with the network size. 
Considering the massive antennas at the BS or large-scale reflecting elements in the RIS, both the interior-point method and SDR technique would be too computationally complex~\cite{Ben-TalA01}.  
To make the large-scale RIS aided secure transmission possible, accelerated first-order algorithms are further proposed.
In particular, to replace the interior-point method, an accelerated projected-gradient method and the path-following procedure (PFP)~\cite{J_Anstre01Anbf} are employed to obtain an iterative algorithm, which converges to at least a local optimal solution. On the other hand, to get around SDR technique, an accelerated Riemannian manifold algorithm is employed~\cite{B_AbsilP09} with convergence to a stationary point guaranteed.
It is proved that the overall first-order algorithm is guaranteed to converge, and the complexity order only scales linearly with the number of antennas at the BS/elements in the RIS. Furthermore, simulation results demonstrate that the proposed accelerated first-order algorithm reduces computation time by at least two orders of magnitude while achieving the same performance as the conventional method that alternatively executes the interior-point method and SDR technique. Finally, simulation results also show that the resultant transmission scheme achieves a significantly higher secure EE than the random phase-shift, fixed phase-shift, and RIS ignoring CSI uncertainty.

The rest of this paper is organized as follows. System model and the secure EE maximization problem are formulated in Section II. In Sections III and IV, a conventional method and an accelerated first-order method are respectively proposed for solving the optimization problem. Simulation results are presented in Section V. Finally, conclusion is drawn in Section VI.

$\mathit{Notation:}$ Column vectors and matrices are denoted by lowercase and uppercase boldface letters, respectively. Conjugate transpose, transpose, trace, the modulus of a scalar, and the $(i,j)^{th}$ element of matrix $\mathbf{X}$ are denoted by $(\cdot)^H$, $(\cdot)^T$, $\mathrm{Tr}(\cdot)$, $|\cdot|$ and $\mathbf{X}_{i,j}$, respectively. 
The mathematical expectation is denoted by $\mathbb{E}\{\cdot\}$.  $\mathrm{diag}\{x_1,\ldots,x_N\}$ denotes a diagonal matrix whose diagonal components are $x_1,\ldots, x_N$. The notations $[x]^+$ and $\mathrm{Pr}(\cdot)$ stand for $\max\{x,0\}$ and probability, respectively. The real part of a complex variable and the Hadamard product between two matrices are denoted by $\Re[\cdot]$ and $\circ$, respectively.
$\mathcal{CN}\left(0,a\right)$ denotes the circularly symmetric complex normal distribution with zero mean and variance $a$, and $\mathrm{Exp}(b)$ denotes the exponential distribution with mean $b$.

\section{System Model and Problem Formulation}
\begin{figure}[tb]
	\centering
	\includegraphics[scale=0.3]{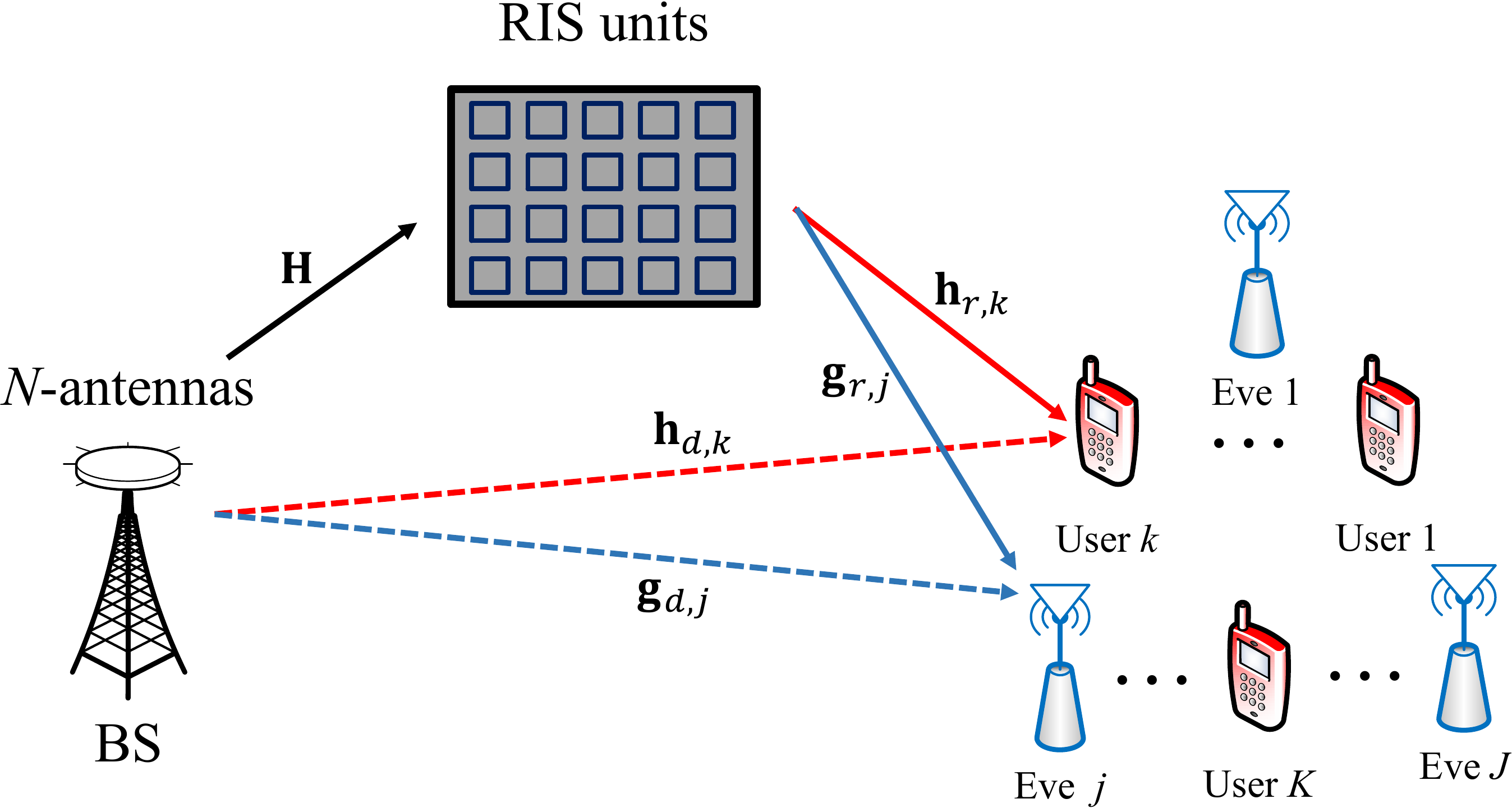}
	\caption{{RIS-aided secure multicast network with multiple users and  eavesdroppers.}}\label{fig:SystemModel}
\end{figure}
We consider a downlink secure multicast system with an $N$-antennas BS, one RIS with $M$ reflecting elements (controlled by communication-oriented software), $K$ single-antenna legitimate users, and $J$ passive single-antenna eavesdroppers (Eves). An example is shown in Fig.~\ref{fig:SystemModel}. The BS intends to multicast common data symbols to $K$ users.
All wireless channels experience quasi-static flat-fading and perturbed by additive white Gaussian noise~\cite{J_Archiyou20IRS_EE}. 
Let the channels from the BS to the RIS, from the RIS to user $k$, from the RIS to Eve $j$, from the BS to user $k$ and from the BS to Eve $j$ be respectively denoted by $\mathbf{H}\in \mathbb{C}^{M\times N}$, $\mathbf{h}_{r,k}\in \mathbb{C}^{M\times 1}$, $\mathbf{g}_{r,j}\in \mathbb{C}^{M\times 1}$, $\mathbf{h}_{d,k}\in\mathbb{C}^{N\times 1}$, and $\mathbf{g}_{d,j}\in\mathbb{C}^{N\times 1}$ with $k \in\{1,\ldots, K\}$ and $j \in\{1,\ldots, J\}$. The phase-shift coefficients of RIS are modeled as $ \mathbf{\Theta}=\mathrm{diag}\{\mathbf{e}\}\in\mathbb{C}^{M\times M}$ with $\mathbf{e}=\{e_m|e_m\in e^{i\theta_m}, \theta_m\in [0,2\pi), \forall m\}$, where $i^2 = -1$, if the phase is modeled as continuous variable. If an RIS with finite resolution is considered, each RIS element can only take one of the $L$ possible phase shifts. As a result, $\mathbf{e}$ is modeled as $\mathbf{e}=\{\left.e_m\right|e_m\in \{0, e^{i\frac{2\pi}{L}}, \ldots,e^{i\frac{2\pi(L-1)}{L}}\},\forall m\}$.

Let the signal transmitted from the BS to all users be $\mathbf{x}=\mathbf{w}s$, where $\mathbf{w}\in \mathbb{C}^{N\times 1}$ is the beamforming vector, and $s$ is the complex information symbol with $\mathbb{E}\{|s|^2\}=1$. 
Then, the received signals at user $k$ and Eve $j$ are respectively given by
\begin{equation}\label{eq:y_B}
y_k=(\sqrt{\alpha_{1}\alpha_{r,k}}\mathbf{h}_{r,k}^H\mathbf{\Theta}\mathbf{H}+\sqrt{\alpha_{d,k}}\mathbf{h}^H_{d,k})\mathbf{w}s+n_k, ~\forall k,
\end{equation}
\begin{equation}\label{eq:y_E}
y_j=(\sqrt{\alpha_{1}\alpha_{r,j}}\mathbf{{g}}_{r,j}^{H}\mathbf{\Theta}\mathbf{H}+\sqrt{\alpha_{d,j}}\mathbf{{g}}_{d,j}^{H})\mathbf{w}s+n_j, ~\forall j,
\end{equation}
where $\alpha_{1}$, $ \alpha_{r,k}$, $\alpha_{r,j}$, $\alpha_{d,k}$, and $\alpha_{d,j}$ are the path-loss coefficients for the links BS-RIS, RIS-user $k$, RIS-Eve $j$, BS-user $k$, and BS-Eve $j$, respectively. The receiver noises at user $k$ and Eve $j$ in~\eqref{eq:y_B} and~\eqref{eq:y_E} are given by 
$n_k\sim \mathcal{CN}(0,\sigma^2_k)$ and $n_j\sim \mathcal{CN}(0,\sigma^2_j)$, respectively.

Based on~\eqref{eq:y_B} and~\eqref{eq:y_E}, the achievable rates of the user $k$ and Eve $j$ are respectively given by $C_k=\log_2(1+|(\sqrt{\alpha_{1}\alpha_{r,k}}\mathbf{h}_{r,k}^H\mathbf{\Theta}\mathbf{H}+\sqrt{\alpha_{d,k}}\mathbf{h}^H_{d,k})\mathbf{w}|^2/\sigma^2_k)$ and $C_j=\log_2(1+|(\sqrt{\alpha_{1}\alpha_{r,j}}\mathbf{{g}}_{r,j}^{H}\mathbf{\Theta}\mathbf{H}+\sqrt{\alpha_{d,j}}\mathbf{{g}}_{d,j}^{H})\mathbf{w}|^2/\sigma^2_j)$.
For the $j^{th}$ Eve's channel, since the BS does not acquire its perfect value, the knowledge of $C_j$ is uncertain~\cite{J_ZZ20SecureProb}. Consequently, a secrecy outage event occurs at the BS when $C_j$ exceeds the redundancy rate of the user $k$, denoted by $D_{k,j}$, and the secrecy outage probability (SOP) of the user $k$ due to Eve $j$ is given by
\begin{equation}\label{eq:SOP_Cons}
\begin{split}
\mathrm{SOP}: \quad p^{k,j}_{so}
&=\mathrm{Pr}\left\{D_{k,j}< \log_2\left(1+|(\sqrt{\alpha_{1}\alpha_{r,j}}\mathbf{{g}}_{r,j}^{H}\mathbf{\Theta}\mathbf{H}+\sqrt{\alpha_{d,j}}\mathbf{{g}}_{d,j}^{H})\mathbf{w}|^2/\sigma^2_j\right)\right\}.
\end{split}
\end{equation}
Furthermore, assuming the non-collaborative eavesdropping model, in which Eves do not exchange their observations or outputs, the instantaneous secrecy rate at user $k$ 
is expressed as $\min\limits_{1\leq j\leq J}\left[\log_2\left(1+|(\sqrt{\alpha_{1}\alpha_{r,k}}\mathbf{h}_{r,k}^H\mathbf{\Theta}\mathbf{H}+\sqrt{\alpha_{d,k}}\mathbf{h}^H_{d,k})\mathbf{w}|^2/\sigma^2_k\right)-D_{k,j}\right]^+$,
which is the minimum over the secrecy rates achieved by the BS under wiretapping of all Eves. Notice that the achievable secrecy rate for multicast network is determined by the worst link~\cite{J_Liu14Gao_outage}. As a result, the achievable secrecy rate is given by
\begin{equation}\label{eq:secRate_sum}
\min\limits_{1\leq k \leq K,1\leq j\leq J}~\left[\log_2\left(1+|(\sqrt{\alpha_{1}\alpha_{r,k}}\mathbf{h}_{r,k}^H\mathbf{\Theta}\mathbf{H}+\sqrt{\alpha_{d,k}}\mathbf{h}^H_{d,k})\mathbf{w}|^2/\sigma^2_k\right)-D_{k,j}\right]^+,
\end{equation}
which is the minimum achievable secrecy rate of all users.

The energy consumption of the RIS-assisted downlink system constitutes three major parts: 1) the transmit power; 2) the hardware static power; and 3) the RIS power consumption. Mathematically, the total power consumption is given by~\cite{J_IRS_EE19Huang}
\begin{equation}\label{eq:total_power} \frac{1}{\eta}\mathrm{Tr}\left(\mathbf{w}\mathbf{w}^H\right)+  P_{a}+K P_c + MP_s,
\end{equation} 
where $\eta$ is the power amplifier efficiency, $\mathrm{Tr}(\mathbf{w}\mathbf{w}^H)$ is the transmit power due to beamforming at the BS, $P_a$, $P_c$ and $P_s$ are the hardware-dissipated power at the BS, the circuit power at each user, and the hardware-dissipated power at each reflecting element, respectively.

Our objective is to design an optimal secure transmission scheme to maximize the EE subject to the SOP constraint and transmit power budget. Since secure EE of the downlink network is defined as a ratio of the achievable secrecy rate to the total power consumption~\cite{J_Zheng18SecureEE}, by using~\eqref{eq:secRate_sum} and~\eqref{eq:total_power}, the secure EE  maximization problem is thus formulated as
\begin{subequations}\label{eq:Secrecy_Rate_Max}
\begin{align}
  \mathcal{P}0: \quad & \max_{\mathbf{w}, \mathbf{\Theta}, \{D_{k,j}\}}
     \frac{\min\limits_{1\leq k \leq K,1\leq j\leq J}\left[\log_2\left(1+\frac{|(\sqrt{\alpha_{1}\alpha_{r,k}}\mathbf{h}_{r,k}^H\mathbf{\Theta}\mathbf{H}+\sqrt{\alpha_{d,k}}\mathbf{h}^H_{d,k})\mathbf{w}|^2}{\sigma^2_k}\right)-D_{k,j}\right]^+}{\frac{1}{\eta}\mathrm{Tr}\left(\mathbf{w}\mathbf{w}^H\right)+ P_a+KP_c + MP_s}, \label{eq:max_SEE} \\
    \mathrm{s.t.}\quad
     &p^{k,j}_\mathrm{so}\leq\varepsilon_k, ~\forall k,j \label{eq:pso_cons_opt} \\
     & |\mathbf{\Theta}_{m,m}| = 1, ~\forall m,\label{cons:phase_IRS} \\
     & \mathrm{Tr}\left(\mathbf{w}\mathbf{w}^H\right)\leq P_\mathrm{max},
\end{align}
\end{subequations}
where $\varepsilon_k\in (0,1)$ is a predefined upper bound representing the maximum tolerable SOP for user $k$, and $P_\mathrm{max}$ is the maximum transmit power at the BS. The unit modulus constraint~\eqref{cons:phase_IRS} ensures that each reflecting element in the RIS does not change the amplitude of the signal, indicating 100\% reflection efficiency. 
Besides, the feasible set of $\mathbf{\Theta}$ can be either continuous phase-shift coefficients or discrete phase-shift coefficients. Since the solution to discrete phase-shift coefficients can be obtained from the continuous phase-shift case via quantization~\cite{J_Huang18IRSDis}, the algorithm derivations focus on the continuous phase-shift case of $\mathcal{P}0$.

Problem $\mathcal{P}0$ provides a general formulation to measure the RIS aided secure transmission performance.
Notice that if we set ${1}/{\eta}=0$, the denominator of the objective function becomes a constant. Therefore, $\mathcal{P}0$ can be employed to investigate not only the secure EE maximization but also the spectral efficiency maximization.

\section{Secure EE Maximization}
\subsection{Handling the Probabilistic Constraint  in $\mathcal{P}0$}
From~\eqref{eq:SOP_Cons}, it can be seen that the SOP is a probability of Rayleigh fading induced outage. Due to the passive and unauthorized nature of Eves, the BS only knows the statistical CSI of the channel from the BS to Eve $j$~\cite{J_DongHuiM20Secure}. 
By leveraging the exponential distribution property of the received signal power in the Rayleigh fading environment~\cite{B_S_Prin02}, a closed-form expression of the SOP can be derived and is given in the following theorem, which is proved in Appendix~\ref{appe:Derive_SOP}.
\begin{theorem}\label{the:CF_SOP_prob}
Supposing $\mathbf{g}_{r,j}\sim \mathcal{CN}(\mathbf{0},\mu_{r,j}^2\mathbf{I}_M)$ and $\mathbf{g}_{d,j}\sim \mathcal{CN}(\mathbf{0},\mu_{d,j}^2\mathbf{I}_N)$, a closed-form expression of the SOP is given by
\begin{equation}\label{eq:them1_SOP}
 p^{k,j}_\mathrm{so}=\exp\left(-\frac{2^{D_{k,j}}-1}{\kappa_{1,j}(\mathbf{\Theta}\mathbf{H}\mathbf{w})^H\mathbf{\Theta}\mathbf{H}\mathbf{w}+\kappa_{2,j}\mathbf{w}^H\mathbf{w}}\right),~~\forall k, j,
\end{equation}
where $\kappa_{1,j}=\alpha_{1}\alpha_{r,j}\mu_{r,j}^2/\sigma^2_j$ and $\kappa_{2,j}=\alpha_{d,j}\mu_{d,j}^2/\sigma^2_j$. 
\end{theorem}
\noindent Based on Theorem~\ref{the:CF_SOP_prob}, the SOP constraint~\eqref{eq:pso_cons_opt} can be expressed as
\begin{equation}\label{ineq:D_j_SOP}
D_{k,j}\geq \log_2\left(1+{(\kappa_{1,j}(\mathbf{\Theta}\mathbf{H}\mathbf{w})^H\mathbf{\Theta}\mathbf{H}\mathbf{w}+\kappa_{2,j}\mathbf{w}^H\mathbf{w})}\ln\varepsilon_k^{-1}\right), ~~\forall k,j.
\end{equation}

By virtue of~\eqref{ineq:D_j_SOP}, $\mathcal{P}0$ can be equivalently transformed into 
\begin{subequations}
	\begin{align}
	\mathcal{P}1: \quad &  \max_{\mathbf{w}, \mathbf{\Theta}, \{D_{k,j}\}}
     \frac{\min\limits_{1\leq k \leq K,1\leq j\leq J}\left[\log_2\left(1+\frac{|(\sqrt{\alpha_{1}\alpha_{r,k}}\mathbf{h}_{r,k}^H\mathbf{\Theta}\mathbf{H}+\sqrt{\alpha_{d,k}}\mathbf{h}^H_{d,k})\mathbf{w}|^2}{\sigma^2_k}\right)-D_{k,j}\right]^+}{\frac{1}{\eta}\mathrm{Tr}\left(\mathbf{w}\mathbf{w}^H\right)+ P_a+KP_c + MP_s}, \label{obj:EE_Se_determine}\\
	\mathrm{s.t.}\quad
	& D_{k,j}\geq \log_2\left(1+{(\kappa_{1,j}(\mathbf{\Theta}\mathbf{H}\mathbf{w})^H\mathbf{\Theta}\mathbf{H}\mathbf{w}+\kappa_{2,j}\mathbf{w}^H\mathbf{w})}\ln\varepsilon_k^{-1} \right), ~~\forall k,j, \label{eq: clf_SOP_trace}\\
	& |\mathbf{\Theta}_{m,m}| = 1, \quad \forall m,\\
	& \mathrm{Tr}\left(\mathbf{w}\mathbf{w}^H\right)\leq P_\mathrm{max}.
	\end{align}
\end{subequations}
Notice that in~\eqref{obj:EE_Se_determine}, decreasing $D_{k,j}$ would lead to a larger value of the objective function. Hence, the optimal value of $D_{k,j}$ is obtained when~\eqref{eq: clf_SOP_trace} reaches equality and is expressed as 
\begin{equation}
D^*_{k,j}= \log_2\left(1+({\kappa_{1,j}(\mathbf{\Theta}\mathbf{H}\mathbf{w})^H\mathbf{\Theta}\mathbf{H}\mathbf{w}+\kappa_{2,j}\mathbf{w}^H\mathbf{w}})\ln\varepsilon_k^{-1}\right), ~~\forall k,j.
\end{equation}
Putting $D^*_{k,j}$ into the objective function of $\mathcal{P}1$, we have 
\begin{equation}\label{eq:obj_temp1}
\frac{\min\limits_{1\leq k \leq K,1\leq j\leq J}\left[\log_2\left(\frac{1+{|(\sqrt{\alpha_{1}\alpha_{r,k}}\mathbf{h}_{r,k}^H\mathbf{\Theta}\mathbf{H}+\sqrt{\alpha_{d,k}}\mathbf{h}^H_{d,k})\mathbf{w}|^2}/{\sigma^2_k}}{1+(\kappa_{1,j}(\mathbf{\Theta}\mathbf{H}\mathbf{w})^H\mathbf{\Theta}\mathbf{H}\mathbf{w}+\kappa_{2,j}\mathbf{w}^H\mathbf{w})\ln\varepsilon_k^{-1}}\right)\right]^+}{\frac{1}{\eta}\mathrm{Tr}\left(\mathbf{w}\mathbf{w}^H\right)+ P_a+KP_c + MP_s}.
\end{equation}
It is observed that the denominator of~\eqref{eq:obj_temp1} does not depend on $k$ and $j$. As a result,~\eqref{eq:obj_temp1} is equivalent to
\begin{equation}\label{obj:final_EE_secure}
\min\limits_{1\leq k \leq K,1\leq j\leq J}\left \{\frac{\left[\log_2\left(\frac{1+{|(\sqrt{\alpha_{1}\alpha_{r,k}}\mathbf{h}_{r,k}^H\mathbf{\Theta}\mathbf{H}+\sqrt{\alpha_{d,k}}\mathbf{h}^H_{d,k})\mathbf{w}|^2}/{\sigma^2_k}}{1+(\kappa_{1,j}(\mathbf{\Theta}\mathbf{H}\mathbf{w})^H\mathbf{\Theta}\mathbf{H}\mathbf{w}+\kappa_{2,j}\mathbf{w}^H\mathbf{w})\ln\varepsilon_k^{-1}}\right)\right]^+}{\frac{1}{\eta}\mathrm{Tr}\left(\mathbf{w}\mathbf{w}^H\right)+ P_a+KP_c + MP_s}\right\}.
\end{equation}
Hence, $\mathcal{P}1$ is equivalently transformed into
\begin{subequations}
	\begin{align}
	\mathcal{P}2: \quad & \max_{\mathbf{w},\mathbf{\Theta}}~
	\min\limits_{1\leq k\leq K,1\leq j\leq J}\left \{~\frac{\left[\log_2\left(\frac{1+{|(\sqrt{\alpha_{1}\alpha_{r,k}}\mathbf{h}_{r,k}^H\mathbf{\Theta}\mathbf{H}+\sqrt{\alpha_{d,k}}\mathbf{h}^H_{d,k})\mathbf{w}|^2}/{\sigma^2_k}}{1+(\kappa_{1,j}(\mathbf{\Theta}\mathbf{H}\mathbf{w})^H\mathbf{\Theta}\mathbf{H}\mathbf{w}+\kappa_{2,j}\mathbf{w}^H\mathbf{w})\ln\varepsilon_k^{-1}}\right)\right]^+}{\frac{1}{\eta}\mathrm{Tr}\left(\mathbf{w}\mathbf{w}^H\right)+ P_a+KP_c + MP_s}\right\}, \label{obj:secure_EE_P2} \\
	\mathrm{s.t.}\quad
	& |\mathbf{\Theta}_{m,m}| = 1, \quad  \forall m, \label{cons: theta_modulus} \\
	& \mathrm{Tr}\left(\mathbf{w}\mathbf{w}^H\right)\leq P_\mathrm{max}. \label{cons:transm_power}
	\end{align}
\end{subequations}
Notice that the inner objective function of $\mathcal{P}2$ only depends on one $k$ or $j$, and $\{\kappa_{1,j},\kappa_{2,j}\}_{j=1}^J$ are independent of $\{\mathbf{w},\mathbf{\Theta}\}$.
Furthermore, parameters $\{\alpha_{r,k},\alpha_{d,k},\sigma_k,\varepsilon_k\}_{k=1}^K$ are independent of $\{\mathbf{w},\mathbf{\Theta}\}$. Therefore, the operations of maximization and minimization in $\mathcal{P}2$ can be interchanged, and we can solve $\mathcal{P}2$ by separately solving $KJ$ independent maximization problems and selecting the minimum value. Here, the subproblem of user $k$ with respect to Eve $j$ is expressed as 
\begin{subequations}
	\begin{align}
	\mathcal{P}2^{[k,j]}: \quad & \max_{\mathbf{w},\mathbf{\Theta}}~	
	~\frac{\left[\log_2\left(\frac{1+{|(\sqrt{\alpha_{1}\alpha_{r,k}}\mathbf{h}_{r,k}^H\mathbf{\Theta}\mathbf{H}+\sqrt{\alpha_{d,k}}\mathbf{h}^H_{d,k})\mathbf{w}|^2}/{\sigma^2_k}}{1+(\kappa_{1,j}(\mathbf{\Theta}\mathbf{H}\mathbf{w})^H\mathbf{\Theta}\mathbf{H}\mathbf{w}+\kappa_{2,j}\mathbf{w}^H\mathbf{w})\ln\varepsilon_k^{-1}}\right)\right]^+}{\frac{1}{\eta}\mathrm{Tr}\left(\mathbf{w}\mathbf{w}^H\right)+ P_a+KP_c + MP_s}, \label{obj:P1_j_parallel}\\
	\mathrm{s.t.}\quad
	& |\mathbf{\Theta}_{m,m}| = 1, \quad  \forall m, \label{cons:modul_thetaP1} \\
	& \mathrm{Tr}\left(\mathbf{w}\mathbf{w}^H\right)\leq P_\mathrm{max}. \label{cons:beamform_power}
	\end{align}
\end{subequations}
Compared to the spectral efficiency maximization problem~\cite{J_Yu20IRS_SE,J_ZhangR_SecureIRS19,J_20IRS_Secure}, the objective of EE maximization problem~\eqref{obj:P1_j_parallel} is defined as the ratio of the spectral efficiency to the total power consumption. Since the spectral efficiency maximization is already a non-concave problem, the EE maximization brings extra non-concave fractional form into the maximization problem, leading to a more challenging class of fractional programming.

In general, if the numerator of the fraction is concave and the denominator is convex (known as concave-convex form), one can employ the quadratic transform~\cite{J_Yuwei18FP} to convert the concave-convex fractional program into a sequence of concave subproblems. 
More specifically, consider a nonnegative concave function $X(\mathbf{w})$, a positive convex function $Y(\mathbf{w})$ and a maximization problem 
$\max\limits_{\mathbf{w}} \frac{X(\mathbf{w})}{Y(\mathbf{w})}$. Applying quadratic
transform leads to a sequence of iterative subproblems, with the $l^{th}$ subproblem expressed as
\begin{equation}\label{opt:general_FP_P3}
\max_{\mathbf{w}}~~
2z^{(l)}\sqrt{X(\mathbf{w})}-(z^{(l)})^2Y(\mathbf{w}), 
\end{equation}
where 
$z^{(l)}=\frac{\sqrt{X(\mathbf{w}^{(l-1)})}}{Y(\mathbf{w}^{(l-1)})}$ with $\mathbf{w}^{(l-1)}$ being the optimal solution of~\eqref{opt:general_FP_P3} at the $(l-1)^{th}$ iteration. 
Since~\eqref{opt:general_FP_P3} is a strongly concave problem, it can be readily solved by numerical convex tools, such as CVX.

Notice that the denominator of the objective function of $\mathcal{P}2^{[k,j]}$ is convex. So the challenge of solving $\mathcal{P}2^{[k,j]}$ comes from the non-concavity of the numerator in~\eqref{obj:P1_j_parallel}. A general framework for concavifying the non-concave term is the 
successive concave approximation (SCA). 
However, depending on the specific non-concave functions, different procedures are needed, and not all procedures can be easily executed. In this work, two procedures are adopted to handle two different forms of non-concave functions:
\begin{itemize}
\item If $X(\mathbf{w})$ can be equivalently rewritten as $X(\mathbf{w})=X_1(\mathbf{w})-X_2(\mathbf{w})$, where both $X_1(\mathbf{w})$ and $X_2(\mathbf{w})$ are concave functions, the concave-convex procedure (CCP)~\cite{J_Lipp16CCP} can be employed to concavify $X(\mathbf{w})$ at a feasible point $\mathbf{w}^{(t)}$ as
\begin{equation}\label{eq:CCP_prere}
\hat{X}(\mathbf{w};\mathbf{w}^{(t)})=X_1(\mathbf{w})-\left(X_2(\mathbf{w}^{(t)})+\nabla_{\mathbf{w}}X_2(\mathbf{w}^{(t)})^H(\mathbf{w}-\mathbf{w}^{(t)})\right).
\end{equation}
\item If either $X_1(\mathbf{w})$ or $X_2(\mathbf{w})$ is nonconcave, the CCP is not applicable and we need to construct a surrogate function to lower bound $X(\mathbf{w})$. To be specific, the path-following procedure (PFP)~\cite{J_Anstre01Anbf} is adopted to concavify $X_1(\mathbf{w})$ as $\hat{X}^{(t)}_1(\mathbf{w})$ and convexify $X_2(\mathbf{w})$ as $\hat{X}^{(t)}_2(\mathbf{w})$, giving a lower bound of $X(\mathbf{w})$ as:
\begin{equation}\label{eq:Path-foPre}
 X(\mathbf{w})\geq \hat{X}^{(t)}_1(\mathbf{w})-\hat{X}^{(t)}_2(\mathbf{w}),
\end{equation}
where the equality holds at $\mathbf{w}=\mathbf{w}^{(t)}$.
\end{itemize}

\subsection{Conventional Method to Solve $\mathcal{P}2^{[k,j]}$}
Noticing that the constraints in $\mathcal{P}2^{[k,j]}$ are decoupled when either $\mathbf{w}$ or $\mathbf{\Theta}$ is fixed, $\mathbf{w}$ and $\mathbf{\Theta}$ can be updated under the alternating maximization (AM) framework.  
When $\mathbf{\Theta}$ is fixed, the subproblem of $\mathcal{P}2^{[k,j]}$ for updating $\mathbf{w}$ is 
	\begin{align}\label{obj:Determini_opt}
	\mathcal{D}1: \quad  \max_{\mathbf{w}}~
~\frac{[f\left(\mathbf{w}\right)]^+}{\frac{1}{\eta}\mathrm{Tr}\left(\mathbf{w}\mathbf{w}^H\right)+ P_a+K P_c + MP_s}, 
~~~	\mathrm{s.t.}~
\mathrm{Tr}\left(\mathbf{w}\mathbf{w}^H\right)\leq P_\mathrm{max},
	\end{align}
where $f\left(\mathbf{w}\right)$ is given by
\begin{equation}\label{eq:f_objective_nu}
f\left(\mathbf{w}\right)=\log_2\left(\frac{1+{|(\sqrt{\alpha_{1}\alpha_{r,k}}\mathbf{h}_{r,k}^H\mathbf{\Theta}\mathbf{H}+\sqrt{\alpha_{d,k}}\mathbf{h}^H_{d,k})\mathbf{w}|^2}/{\sigma^2_k}}{1+(\kappa_{1,j}(\mathbf{\Theta}\mathbf{H}\mathbf{w})^H\mathbf{\Theta}\mathbf{H}\mathbf{w}+\kappa_{2,j}\mathbf{w}^H\mathbf{w})\ln\varepsilon_k^{-1}}\right).
\end{equation}
Our aim is to transform $[f\left(\mathbf{w}\right)]^+$ into a concave form. In particular, by introducing an auxiliary variable $\mathbf{W}=\mathbf{w}\mathbf{w}^H\in \mathbb{C}^{N\times N} $ with an additional rank constraint $\mathrm{rank}(\mathbf{W})=1$, $f\left(\mathbf{w}\right)$ can be rewritten as
\begin{align}\label{eq:DC_concavify}
F(\mathbf{W})=&\underbrace{\log_2\left(1+{(\sqrt{\alpha_{1}\alpha_{r,k}}\mathbf{h}_{r,k}^H\mathbf{\Theta}\mathbf{H}+\sqrt{\alpha_{d,k}}\mathbf{h}^H_{d,k})\mathbf{W}(\sqrt{\alpha_{1}\alpha_{r,k}}\mathbf{h}_{r,k}^H\mathbf{\Theta}\mathbf{H}+\sqrt{\alpha_{d,k}}\mathbf{h}^H_{d,k})^H}/{\sigma^2_k}\right)}_{:=F_1(\mathbf{W})} \nonumber \\
&-\underbrace{\log_2\left(1+(\kappa_{1,j}\mathrm{Tr}(\mathbf{\Theta}\mathbf{H}\mathbf{W}(\mathbf{\Theta}\mathbf{H})^H)+\kappa_{2,j}\mathrm{Tr}(\mathbf{W}))\ln\varepsilon_k^{-1}\right)}_{:=F_2(\mathbf{W})},
\end{align}
which is a difference of two concave functions. Based on CCP method~\eqref{eq:CCP_prere}, $F(\mathbf{W})$ can be locally concavified  at a feasible point $\mathbf{W}^{(n)}$: 
\begin{align}\label{eq: der_trace_CCP}
\hat{F}(\mathbf{W};\mathbf{W}^{(n)})
=&F_1(\mathbf{W})-F_2(\mathbf{W}^{(n)}) \nonumber \\
&-\frac{\left(\kappa_{1,j}\mathrm{Tr}\left(\mathbf{\Theta}\mathbf{H}
	\left(\mathbf{W}-\mathbf{W}^{(n)}\right)(\mathbf{\Theta}\mathbf{H})^H\right)+ \kappa_{2,j}\mathrm{Tr}(\mathbf{W}-\mathbf{W}^{(n)})
\right)\ln\varepsilon_k^{-1}}{\left(1+\left(\kappa_{1,j}\mathrm{Tr}(\mathbf{\Theta}\mathbf{H}\mathbf{W}^{(n)}(\mathbf{\Theta}\mathbf{H})^H)+\kappa_{2,j}\mathrm{Tr}(\mathbf{W}^{(n)})\right)\ln\varepsilon_k^{-1}\right)\ln 2}.
\end{align}
With the help of $\mathbf{W}$ and $\hat{F}(\mathbf{W};\mathbf{W}^{(n)})$, $\mathcal{D}1$ can be iteratively replaced by a sequence of subproblems with the $(n+1)^{th}$ subproblem being 
	\begin{align}\label{opt:D_W_Dinkel_nonCon}
\max_{\mathbf{W}\succeq \mathbf{0}}~
	~\frac{[\hat{F}(\mathbf{W};\mathbf{W}^{(n)})]^+}{\frac{1}{\eta}\mathrm{Tr}\left(\mathbf{W}\right)+ P_a+KP_c + MP_s},
~~\mathrm{s.t.}~
  \mathrm{Tr}\left(\mathbf{W}\right)\leq P_\mathrm{max},~ \mathrm{rank}(\mathbf{W})=1.
	\end{align}
Since $\hat{F}(\mathbf{W};\mathbf{W}^{(n)})$ is concave on $\mathbf{W}$ and pointwise maximum operation preserves concavity~\cite{Cov_Opt90}, $[\hat{F}(\mathbf{W};\mathbf{W}^{(n)})]^+$ is a concave function, leading to a concave-convex form of objective function in~\eqref{opt:D_W_Dinkel_nonCon}. Accordingly, for a fixed $n$, the 
quadratic transform method~\cite{J_Yuwei18FP} can be further applied to  convert~\eqref{opt:D_W_Dinkel_nonCon} into a sequence of subproblems, with the $l^{th}$ subproblem expressed as
\begin{subequations}\label{opt:D1_w_CCP}
	\begin{align}
	\max_{\mathbf{W}\succeq \mathbf{0}}~~
	&~2y^{(l)}\sqrt{[\hat{F}(\mathbf{W};\mathbf{W}^{(n)})]^+}-(y^{(l)})^2\left(\frac{1}{\eta}\mathrm{Tr}\left(\mathbf{W}\right)+ P_a+KP_c + MP_s\right), \\
	\mathrm{s.t.}\quad
	&  \mathrm{Tr}\left(\mathbf{W}\right)\leq P_\mathrm{max},~
	\mathrm{rank}(\mathbf{W})=1,\label{ineq:SDR_D4}
	\end{align}
\end{subequations}
where
$y^{(l)}$ is defined in Algorithm~\ref{alg:conventional_D1}. To handle the non-convex rank constraint,~\eqref{opt:D1_w_CCP} can be relaxed by dropping the rank constraint, and the relaxed problem is given by 
\begin{equation}\label{opt:final_IPM_W}
\max_{\mathbf{W}\succeq \mathbf{0}}~
2y^{(l)}\sqrt{[\hat{F}(\mathbf{W};\mathbf{W}^{(n)})]^+}\!-\!(y^{(l)})^2\left(\frac{\mathrm{Tr}\left(\mathbf{W}\right)}{\eta}\!+\! P_a\!+\!KP_c + MP_s\right)\!, ~\mathrm{s.t.}~
\mathrm{Tr}\left(\mathbf{W}\right)\leq P_\mathrm{max},
\end{equation}
which is a semidefinite programming and directly solved via the interior-point method with the complexity order of $\mathcal{O}(N^3)$~\cite{J_PA10IPM_Bok}.
The property of the solution to~\eqref{opt:final_IPM_W} is revealed by the following theorem, and the proof is delegated to Appendix~\ref{Rank:SDR_1}.
\begin{theorem}\label{tem:FP_CCP_def}
The optimal solution of~\eqref{opt:final_IPM_W} is always rank-one.
\end{theorem}

The entire procedure for solving $\mathcal{D}1$  is summarized in Algorithm~\ref{alg:conventional_D1} with outer iterations over $n$ and inner iterations over $l$. Based on Theorem~\ref{tem:FP_CCP_def}, it is known that there is no performance loss after rank relaxation, and the optimal solution to~\eqref{opt:final_IPM_W} is also the optimal solution to~\eqref{opt:D1_w_CCP}. Moreover, since~\eqref{opt:D1_w_CCP} is the quadratic transformation of~\eqref{opt:D_W_Dinkel_nonCon}, the iteration over $l$ is guaranteed to converge to a local optimal solution to~\eqref{opt:D_W_Dinkel_nonCon}~\cite{J_Yuwei18FP}.
Besides, the iteration over $n$ with CCP method is guaranteed to converge to a stationary solution of $\mathcal{D}1$~\cite{J_Lipp16CCP}.
\begin{algorithm}[H]
	\caption{Conventional method for solving $\mathcal{D}1$} 
	\begin{algorithmic}[1]\label{alg:conventional_D1}
		\STATE Initialize $\mathbf{W}^{(0)}$ and set $n:=0$.
		\REPEAT
		\STATE Initialize $y^{(0)}$ and set $l:=0$.
		\REPEAT
		\STATE  Solve problem~\eqref{opt:final_IPM_W} with interior-point method and output solution $\mathbf{W}^{\diamond}$.
		\STATE Set $l:=l+1$, and compute $y^{(l)}=\frac{([\hat{F}(\mathbf{W}^{\diamond},\mathbf{W}^{(n)})]^+)^{1/2}}{\frac{1}{\eta}\mathrm{Tr}\left(\mathbf{W}^{\diamond}\right)+ P_a+KP_c + MP_s}$.
		\UNTIL Stopping criterion is satisfied.
		\STATE Update $\mathbf{W}^{(n+1)}=\mathbf{W}^{\diamond}$ and iteration $n:=n+1$.
		\UNTIL Stopping criterion is satisfied.
	\end{algorithmic}
\end{algorithm}

On the other hand, when beamforming vector $\mathbf{w}$ is fixed, the subproblem of $\mathcal{P}2^{[k,j]}$ for updating $\mathbf{\Theta}$ becomes 
	\begin{align}
	\mathcal{Q}1: \max_{\mathbf{\Theta}}
	~\log_2\left(\frac{1+{|(\sqrt{\alpha_{1}\alpha_{r,k}}\mathbf{h}_{r,k}^H\mathbf{\Theta}\mathbf{H}+\sqrt{\alpha_{d,k}}\mathbf{h}^H_{d,k})\mathbf{w}|^2}/{\sigma^2_k}}{1+(\kappa_{1,j}(\mathbf{\Theta}\mathbf{H}\mathbf{w})^H\mathbf{\Theta}\mathbf{H}\mathbf{w}+\kappa_{2,j}\mathbf{w}^H\mathbf{w})\ln\varepsilon_k^{-1}}\right), \label{obj:Q1_theta} 
	~\mathrm{s.t.}~
	|\mathbf{\Theta}_{m,m}| = 1,~  \forall m, 
	\end{align}
where $[\cdot]^+$ is removed since the objective function value of $\mathcal{P}2^{[k,j]}$ must be non-negative at optimality. 
Denoting $\mathbf{A}_k=\left[\begin{matrix}
\sqrt{\alpha_{1}\alpha_{r,k}}\mathrm{diag}(\mathbf{h}_{r,k}^H)\mathbf{H} \\
\sqrt{\alpha_{d,k}}\mathbf{h}^H_{d,k}  
\end{matrix}\right]\in \mathbb{C}^{(M+1) \times N}$, $\hat{\mathbf{H}}=\left[
\mathbf{H}~
\mathbf{0}  \right]^T\in \mathbb{C}^{(M+1) \times N}$, and noticing that $\mathbf{\Theta}=\mathrm{diag}\{\mathbf{e}\}$, we have the following equalities 
\begin{equation}\label{eq:obje_demo_Q1}
|(\sqrt{\alpha_{1}\alpha_{r,k}}\mathbf{h}_{r,k}^H\mathbf{\Theta}\mathbf{H}+\sqrt{\alpha_{d,k}}\mathbf{h}^H_{d,k})\mathbf{w}|^2=\mathrm{Tr}(\mathbf{A}_k\mathbf{w}\mathbf{w}^H\mathbf{A}^H_k\mathbf{Q}),
\end{equation}
\begin{equation}\label{eq:obje_numer_Q1}
(\mathbf{\Theta}\mathbf{H}\mathbf{w})^H\mathbf{\Theta}\mathbf{H}\mathbf{w}=\mathrm{Tr}(\hat{\mathbf{H}}\mathbf{w}\mathbf{w}^H\hat{\mathbf{H}}^H\mathbf{Q}),
\end{equation}
where $\mathbf{Q}=[\mathbf{e}^T,1]^T [\mathbf{e}^T,1]$.  Putting~\eqref{eq:obje_demo_Q1} and~\eqref{eq:obje_numer_Q1} into $	\mathcal{Q}1$, $\mathcal{Q}1$ can be equivalently transformed into\footnote{ $\log_2$ is removed since logarithm function does not affect the optimization. }
\begin{subequations}\label{opt:Q1_equi_Qudra}
	\begin{align}
	\max_{{\mathbf{Q}\succeq \mathbf{0}}}~
	&\frac{1+\mathrm{Tr}(\mathbf{A}_k\mathbf{w}\mathbf{w}^H\mathbf{A}^H_k\mathbf{Q})/{\sigma^2_k}}{1+(\kappa_{1,j}\mathrm{Tr}(\hat{\mathbf{H}}\mathbf{w}\mathbf{w}^H\hat{\mathbf{H}}^H\mathbf{Q})+\kappa_{2,j}\mathbf{w}^H\mathbf{w})\ln\varepsilon_k^{-1}},  \label{obj:FP_Dinkel_theta} \\
	\mathrm{s.t.}\quad
	&\mathrm{Tr}\left(\mathbf{E}_m\mathbf{Q}\right)=1, ~\forall m,~\mathrm{rank}(\mathbf{Q})= 1, \label{cons:rank_1_Dinkel}
	\end{align}
\end{subequations}
where $\mathbf{E}_m$ is a matrix with the $(m,m)^{th}$ element being 1 and 0 in other positions such that~\eqref{cons:rank_1_Dinkel} is equivalent to constraint $|\mathbf{\Theta}_{m,m}| = 1, \forall m$.
Notice that~\eqref{opt:Q1_equi_Qudra} complies with concave-convex form such that 
the quadratic transformation~\eqref{opt:general_FP_P3} can be employed to solve this problem with the $l^{th}$ subproblem expressed as 
\begin{subequations}\label{opt:SDP_Q_SDR}
	\begin{align}
	\max_{\mathbf{Q}\succeq \mathbf{0}}~
	&2\tau^{(l)}\sqrt{1\!+\!\mathrm{Tr}(\mathbf{A}_k\mathbf{w}\mathbf{w}^H\mathbf{A}^H_k\mathbf{Q})/{\sigma^2_k}}\!-\!
	(\tau^{(l)})^2\left(1\!+\!(\kappa_{1,j}\mathrm{Tr}(\hat{\mathbf{H}}\mathbf{w}\mathbf{w}^H\hat{\mathbf{H}}^H\mathbf{Q})\!+\!\kappa_{2,j}\mathbf{w}^H\mathbf{w})\ln\varepsilon_k^{-1}\right), \\
	\mathrm{s.t.}\quad
	&  \mathrm{Tr}\left(\mathbf{E}_m\mathbf{Q}\right)=1, \forall m,
	 ~\mathrm{rank}(\mathbf{Q})= 1,
	\end{align}
\end{subequations}
where 
$\tau^{(l)}$ is defined in Algorithm~\ref{alg:conventional_Q1}. Subproblem~\eqref{opt:SDP_Q_SDR} can be directly solved by employing SDR, and Gaussian randomization method is further used to guarantee a feasible solution~\cite{J_LuoSDR}.

The entire procedure for solving $\mathcal{Q}1$ is summarized in Algorithm~\ref{alg:conventional_Q1} with iterations over $l$. Since $\mathcal{Q}1$ is equivalent to~\eqref{opt:Q1_equi_Qudra} and~\eqref{opt:Q1_equi_Qudra} is iteratively replaced by quadratic transformation without performance loss, the obtained solution to~\eqref{opt:SDP_Q_SDR} is a feasible solution to $\mathcal{Q}1$. Notice that~\eqref{opt:SDP_Q_SDR} has $M$ semidefinite programming constraints, and each involves an $M\times M$ positive semidefinite matrix. Therefore, the computational complexity order of~\eqref{opt:SDP_Q_SDR} is 
$\mathcal{O}\left(\sqrt{M}(2M^4+M^3)\right)$~\cite[Theorem 3.12]{J_PA10IPM_Bok}. 
\begin{algorithm}[H]
	\caption{Conventional method for solving $\mathcal{Q}1$} 
	\begin{algorithmic}[1]\label{alg:conventional_Q1}
		\STATE Initialize feasible values of $\mathbf{Q}$, $\tau^{(0)}$ and set $l:=0$.
		\REPEAT
		\STATE  Solve problem~\eqref{opt:SDP_Q_SDR} with SDR method and output solution $\mathbf{Q}^{\diamond}$.
		\STATE Update iteration $l:=l+1$ and $\tau^{(l)}=\frac{\sqrt{1+\mathrm{Tr}(\mathbf{A}_k\mathbf{w}\mathbf{w}^H\mathbf{A}^H_k\mathbf{Q}^{\diamond})/\sigma^2_k}}{1+(\kappa_{1,j}\mathrm{Tr}(\hat{\mathbf{H}}\mathbf{w}\mathbf{w}^H\hat{\mathbf{H}}^H\mathbf{Q}^{\diamond})+\kappa_{2,j}\mathbf{w}^H\mathbf{w})\ln\varepsilon_k^{-1}}$.
		\UNTIL Stopping criterion is satisfied.
	\end{algorithmic}
\end{algorithm}

To sum up, with the feasible initial points of $\mathbf{W}$ and $\mathbf{Q}$, $\mathcal{P}2^{[k,j]}$ can be solved through iteratively executing Algorithms~\ref{alg:conventional_D1} and~\ref{alg:conventional_Q1}.
However, since the interior-point method is involved in Algorithms~\ref{alg:conventional_D1} and~\ref{alg:conventional_Q1}, the computational complexity for solving $\mathcal{P}2^{[k,j]}$ is very high when the number of antennas $N$ or the reflecting elements $M$ is large.

\section{Large-Scale Optimization Algorithm}
To overcome the high complexity challenge brought by Algorithms~\ref{alg:conventional_D1} and~\ref{alg:conventional_Q1} when $N$ or $M$ is large, accelerated first-order algorithms are proposed in this section to significantly reduce the time complexity.
In particular, we propose an accelerated projected-gradient (PG) algorithm with a PFP for solving $\mathcal{D}1$. Besides, an accelerated Riemannian manifold optimization algorithm is proposed for solving $\mathcal{Q}1$. Finally, the convergence of the proposed iterative algorithms is proved mathematically.

\subsection{Accelerated PG Algorithm for Solving $\mathcal{D}1$}
Although $f\left(\mathbf{w}\right)$ in~\eqref{eq:f_objective_nu} can be concavified via~\eqref{eq: der_trace_CCP} based on CCP of~\eqref{eq:CCP_prere},
the introduced auxiliary variable $\mathbf{W}$ together with the rank constraint bring extra computational complexity. To maintain the optimization variable as $\mathbf{w}$, we employ the PFP of~\eqref{eq:Path-foPre} to derive a lower bound of $f\left(\mathbf{w}\right)$ as shown in the following property.
\begin{lemma}\label{lem:Approx_Obj_D1}
Denote $\hat{\mathbf{h}}=\sqrt{\alpha_{1}\alpha_{r,k}}\mathbf{h}_{r,k}^H\mathbf{\Theta}\mathbf{H}+\sqrt{\alpha_{d,k}}\mathbf{h}^H_{d,k}$. Given any fixed $\mathbf{w}^{(t)}$, we have: $f\left(\mathbf{w}\right) \geq R_1^{(t)}(\mathbf{w})-R_2^{(t)}(\mathbf{w})$ with the equality when $\mathbf{w}=\mathbf{w}^{(t)}$, where $R^{(t)}_1(\mathbf{w})$ and $R^{(t)}_2(\mathbf{w})$ are
\begin{align}\label{eq:R1_auxi}
R^{(t)}_1(\mathbf{w})\!=&\log_2\left(1+\frac{|\hat{\mathbf{h}}\mathbf{w}^{(t)}|^2}{\sigma^2_k}\right)\!+\!\frac{2\Re\{(\hat{\mathbf{h}}\mathbf{w}^{(t)})^H\hat{\mathbf{h}}\mathbf{w}\}}{\sigma^2_k\ln 2} 
\!-\!\frac{(\sigma^2_k+|\hat{\mathbf{h}}\mathbf{w}|^2)|\hat{\mathbf{h}}\mathbf{w}^{(t)}|^2}{(\sigma^2_k+|\hat{\mathbf{h}}\mathbf{w}^{(t)}|^2)\sigma^2_k\ln2}
-\frac{|\hat{\mathbf{h}}\mathbf{w}^{(t)}|^2}{\sigma^2_k\ln2},
\end{align}
\begin{align}\label{eq:R2_auxli}
R^{(t)}_2(\mathbf{w})\!=&\log_2\left(1+\left(\kappa_{1,j}\|\mathbf{\Theta}\mathbf{H}\mathbf{w}^{(t)}\|^2+\kappa_{2,j}\|\mathbf{w}^{(t)}\|^2\right)\ln\varepsilon_k^{-1}\right)\nonumber \\
&+\frac{\left(\kappa_{1,j}(\|\mathbf{\Theta}\mathbf{H}\mathbf{w}\|^2-\|\mathbf{\Theta}\mathbf{H}\mathbf{w}^{(t)}\|^2)+\kappa_{2,j}(\|\mathbf{w}\|^2-\|\mathbf{w}^{(t)}\|^2)\right)\ln\varepsilon_k^{-1}}{\left(1+\left(\kappa_{1,j}\|\mathbf{\Theta}\mathbf{H}\mathbf{w}^{(t)}\|^2+\kappa_{2,j}\|\mathbf{w}^{(t)}\|^2\right)\ln\varepsilon_k^{-1}\right)\ln2}.
\end{align}
\end{lemma}
\begin{proof}
Please see Appendix~\ref{lem:App_LB_obj}.
\end{proof}
\noindent Since $R_1^{(t)}(\mathbf{w})$ is a concave function and $R_2^{(t)}(\mathbf{w})$ is a convex function, the lower bound  $R_1^{(t)}(\mathbf{w})-R_2^{(t)}(\mathbf{w})$ is a concave function. 
With the sequence of concave functions $\{R_1^{(t)}(\mathbf{w})-R_2^{(t)}(\mathbf{w})\}_{t\in\mathbb{N}}$,
$\mathcal{D}1$ can be iteratively replaced by a sequence of subproblems with the $t^{th}$ subproblem given by
	\begin{align}\label{opt:CCR_D1_inner_Appx}
\mathbf{w}^{(t+1)}=\mathop{\arg\max}\limits_{\mathbf{w}}~
	&\frac{\left[R_1^{(t)}(\mathbf{w})-R_2^{(t)}(\mathbf{w})\right]^+}{\frac{1}{\eta}\mathrm{Tr}\left(\mathbf{w}\mathbf{w}^H\right)+ P_a+KP_c + MP_s}, 
	~~\mathrm{s.t.}
	~\mathrm{Tr}\left(\mathbf{w}\mathbf{w}^H\right)\leq P_\mathrm{max}.
	\end{align}
Since the objective function value of~\eqref{opt:CCR_D1_inner_Appx} must be non-negative at optimality, $[\cdot]^+$ in the numerator can be dropped. 
Recognizing the concave-convex form of~\eqref{opt:CCR_D1_inner_Appx}, quadratic transformation~\eqref{opt:general_FP_P3} can be applied to~\eqref{opt:CCR_D1_inner_Appx} with the $l^{th}$ subproblem written as
\begin{subequations}\label{eq:obj_approx_dual_problem}
	\begin{align}
\quad  \max_{\mathbf{w}}~~
	&\underbrace{2\gamma^{(l)}\sqrt{R_1^{(t)}(\mathbf{w})-R_2^{(t)}(\mathbf{w})}-(\gamma^{(l)})^2\left(\frac{1}{\eta}\mathrm{Tr}\left(\mathbf{w}\mathbf{w}^H\right)+ P_a+KP_c + MP_s\right)}_{\triangleq \Phi^{(t)}(\mathbf{w})}, \label{obj:Approx_concave_D1}  \\
	\mathrm{s.t.}\quad
	&  \mathrm{Tr}\left(\mathbf{w}\mathbf{w}^H\right)\leq P_\mathrm{max}, \label{ineq: power_budget_cons}
	\end{align}
\end{subequations}
where $\gamma^{(l)}$ is defined in Algorithm~\ref{alg:PG_Path}.

To solve~\eqref{eq:obj_approx_dual_problem}, the PG method with a warm-start initialization $\mathbf{x}^{(0)}:=\mathbf{w}^{(t)}$ can be employed.  In the $i^{th}$ iteration, the update of $\mathbf{x}$ at the gradient descent step is given by
\begin{equation}\label{eq:grad_PG_w}
\mathbf{x}^{(i+\frac{1}{2})}=\mathbf{x}^{(i)}+\mathcal{I}^{(i)}
\nabla_\mathbf{w}  \Phi^{(t)}|_{\mathbf{w}=\mathbf{x}^{(i)}},
\end{equation}
where $\mathcal{I}^{(i)}$ is a variable step-size chosen by Armijo condition to guarantee convergence~\cite{J_Abs09Accel_Amoji}, and the gradient of $\Phi^{(t)}(\mathbf{w})$ is given by
\begin{align}
&\nabla_\mathbf{w} \Phi^{(t)}
=\frac{2\gamma^{(l)}\left(\frac{(\hat{\mathbf{h}}\mathbf{w}^{(t)})^H\hat{\mathbf{h}}^H}{\sigma^2_k\ln 2}\!-\!\frac{(\hat{\mathbf{h}}\mathbf{w})^H\hat{\mathbf{h}}^H|\hat{\mathbf{h}}\mathbf{w}^{(t)}|^2}{(\sigma^2_k+|\hat{\mathbf{h}}\mathbf{w}^{(t)}|^2)\sigma^2_k\ln 2}
	\!-\!\frac{\left(\kappa_{1,j}(\mathbf{\Theta}\mathbf{H})^H(\mathbf{\Theta}\mathbf{H}\mathbf{w})+ \kappa_{2,j}\mathbf{w}\right)\ln\varepsilon_k^{-1} }{\left(1+(\kappa_{1,j}\|\mathbf{\Theta}\mathbf{H}\mathbf{w}^{(t)}\|^2+\kappa_{2,j}\|\mathbf{w}^{(t)}\|^2)\ln\varepsilon_k^{-1}\right)\ln2}\!-\!\frac{(\gamma^{(l)})^2}{\eta}\mathbf{w}\right)}{\sqrt{\left(R_1^{(t)}(\mathbf{w})-R_2^{(t)}(\mathbf{w})\right)}}.
\end{align}
On the other hand, to project $\mathbf{x}^{(i+\frac{1}{2})}$ onto the feasible set of~\eqref{eq:obj_approx_dual_problem}, we have an equivalent optimization problem expressed as
\begin{align}\label{opt:project_w}
&\mathbf{x}^{(i+1)}=\mathop{\arg\min}_{\mathbf{x}\in \mathcal{P}_{\mathbf{x}}} \|\mathbf{x}-\mathbf{x}^{(i+\frac{1}{2})}\|^2,
\end{align}
where $\mathcal{P}_{\mathbf{x}}=\{\mathbf{x}| \|\mathbf{x}\|^2\leq P_\mathrm{max}\}$ is a convex set. Since~\eqref{opt:project_w} is a strongly convex problem, a closed-form solution can be derived based on Karush-Kuhn-Tucker (KKT) condition and is given by the following property, which is proved in Appendix~\ref{proof:theorem_primal_dual}.
\begin{lemma}\label{them: optimal_w_Approxi}
The optimal solution to~\eqref{opt:project_w} is given by
$\mathbf{x}^{(i+1)}=	\min\left\{\mathbf{x}^{(i+\frac{1}{2})},\frac{\sqrt{P_\mathrm{max}}\mathbf{x}^{(i+\frac{1}{2})}}{\|\mathbf{x}^{(i+\frac{1}{2})}\|}\right\}$.
\end{lemma}
\noindent
Iteratively executing~\eqref{eq:grad_PG_w} and Lemma~\ref{them: optimal_w_Approxi} results in the standard projected gradient.  But since~\eqref{eq:obj_approx_dual_problem} is a smooth concave optimization problem, the momentum technique~\cite{J_BeckAmir09AFIS} can be used to accelerate the PG method. In particular, $\mathbf{x}^{(i)}$ in~\eqref{eq:grad_PG_w} can be augmented by a momentum term, giving an accelerated PG method~\cite{J_BeckAmir09AFIS}:
\begin{equation}\label{eq:acce_PG_moment}
\mathbf{x}^{(i+1)}=
\min\left\{\mathbf{f}^{(i)}+\mathcal{I}^{(i)}\nabla_\mathbf{w}  \Phi^{(t)}|_{\mathbf{w}=\mathbf{f}^{(i)}}, \frac{\sqrt{P_\mathrm{max}}\left(\mathbf{f}^{(i)}+\mathcal{I}^{(i)}\nabla_\mathbf{w}  \Phi^{(t)}|_{\mathbf{w}=\mathbf{f}^{(i)}}\right)}{\|\mathbf{f}^{(i)}+\mathcal{I}^{(i)}\nabla_\mathbf{w}  \Phi^{(t)}|_{\mathbf{w}=\mathbf{f}^{(i)}}\|} \right\},
\end{equation}
where $\mathbf{f}^{(i)}$ is given by 
\begin{equation}\label{eq:new_Accela_PG-para}
\mathbf{f}^{(i)}=\mathbf{x}^{(i)}+\frac{a^{(i-1)}-1}{a^{(i)}}\left(\mathbf{x}^{(i)}-\mathbf{x}^{(i-1)}\right)
\end{equation}
with a monotonically increasing sequence $\{a^{(i)}\}$ to adjust the momentum $\mathbf{x}^{(i)}-\mathbf{x}^{(i-1)}$. To achieve a fast convergence rate, $a^{(i)}$ is updated by~\cite{J_BeckAmir09AFIS}
\begin{equation}\label{eq:tunned_para_PG}
a^{(0)}=1,~a^{(i)}=\frac{1+\sqrt{1+4(a^{(i-1)})^2}}{2}.
\end{equation}
By updating $\mathbf{x}$ based on~\eqref{eq:acce_PG_moment}-\eqref{eq:tunned_para_PG}, the global optimal solution $\mathbf{w}^{*}$ to problem~\eqref{eq:obj_approx_dual_problem} can be obtained. 
Compared to the simple gradient descent, the proposed accelerated PG achieves a convergence rate of $\mathcal{O}\left(1/i^2\right)$~\cite{J_BeckAmir09AFIS}, and the differences are compared in Fig.~\ref{fig:algo_PG_accce}. 
This $\mathcal{O}\left(1/i^2\right)$ iteration complexity in fact touches the lower bound for any gradient based algorithm~\cite{B_YNesterov}, which indicates that the proposed algorithm is theoretically one of the fastest algorithms for this problem.
\begin{figure}[tb]
	\centering
	\includegraphics[scale=0.36]{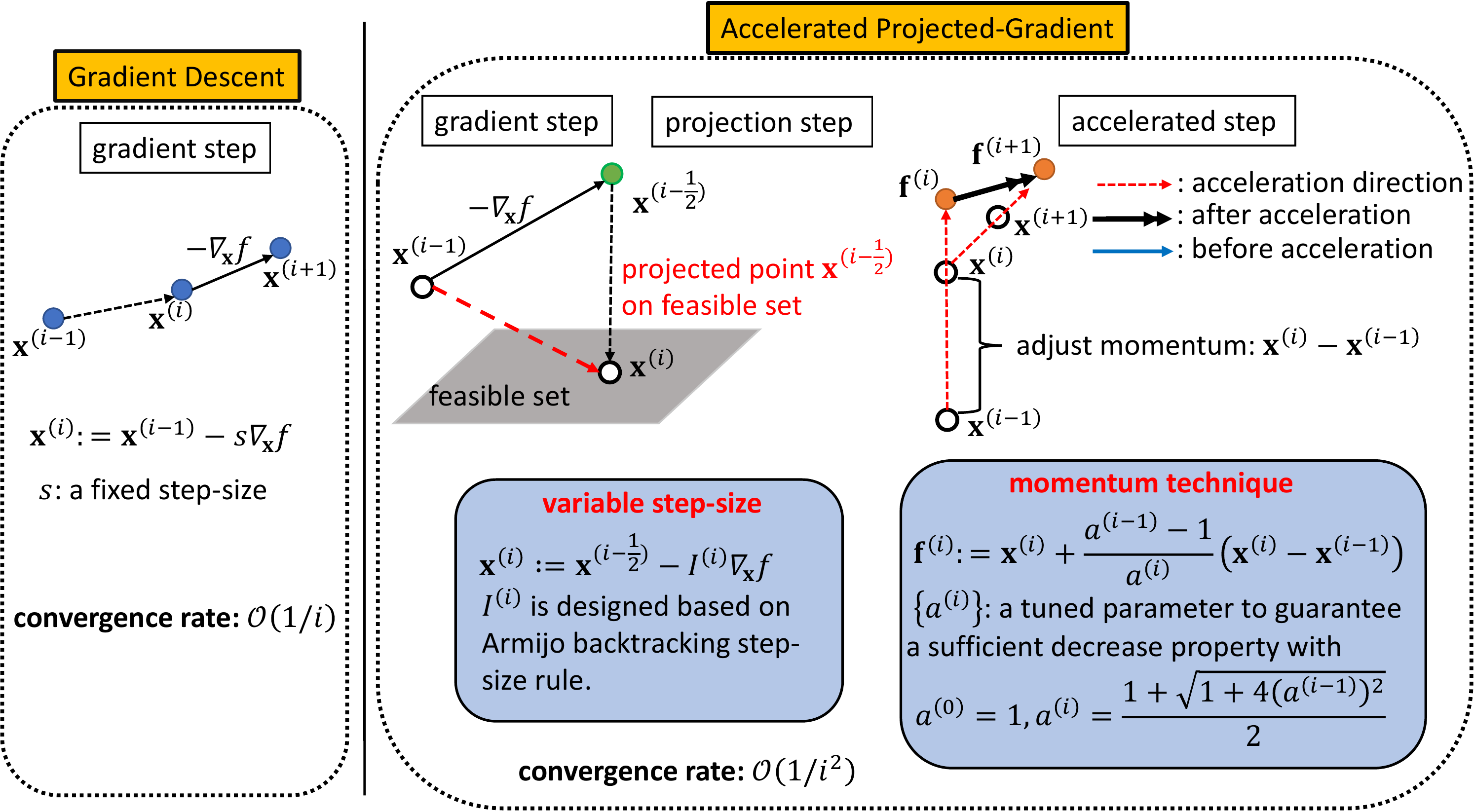}
	\caption{Comparison of gradient descent and accelerated projected-gradient methods.} \label{fig:algo_PG_accce} 
\end{figure}

The entire procedure for solving $\mathcal{D}1$ is summarized in Algorithm~\ref{alg:PG_Path}, where PFP is applied over $t$ and the quadratic transformation is applied over $l$ for the outer iterations, and each subproblem~\eqref{eq:obj_approx_dual_problem} is solved with accelerated PG method over $i$ for the inner iterations, where the convergence property is revealed by Theorem~\ref{theorm: convergence_property}, which is proved in Appendix~\ref{proof:theorem_KKT}.
\begin{theorem}\label{theorm: convergence_property}
	Starting from a feasible point of $\mathbf{w}$, the sequence of solutions generated from Algorithm~\ref{alg:PG_Path} converges to a local optimal point of $\mathcal{D}1$.
\end{theorem}

\subsection{Accelerated Riemannian Manifold Algorithm for Solving $\mathcal{Q}1$}
To reduce the complexity in solving $\mathcal{Q}1$, we notice that $\mathbf{e}=[e^{i\theta_1},\ldots, e^{i\theta_M}]^H$ is the diagonal elements of $\mathbf{\Theta}$, and $\mathcal{Q}1$ can be equivalently rewritten as
	\begin{align}\label{obj:Q2_theta_v} 
	\mathcal{Q}2: ~~\max_{\mathbf{v}}~
	~\underbrace{\log_2\left(\frac{1+{|\mathbf{v}^H\mathbf{A}_k\mathbf{w}|^2}/{\sigma^2_k}}{1+(\kappa_{1,j}|\mathbf{w}^H\hat{\mathbf{H}}^H\mathbf{v}|^2+\kappa_{2,j}\|\mathbf{w}\|^2)\ln\varepsilon_k^{-1}}\right)}_{:=\Upsilon(\mathbf{v})}, ~
	\mathrm{s.t.}~
	[\mathbf{v}\mathbf{v}^H]_{m,m} = 1,~ \forall m,
	\end{align}	
	where $\mathbf{v}=[\mathbf{e}^T,1]^T$.
	Since the constraint set of $\mathcal{Q}2$ forms an oblique manifold~\cite{B_AbsilP09}, the standard Riemannian conjugate gradient (CG) can be applied to solve $\mathcal{Q}2$ via alternatively computing Riemannian gradient, finding conjugate direction, and performing retraction mapping, as has been done in~\cite{C_Xu19Retraction,J_PanRCG_IRS20,J_ReminCG_IRS21}.

\begin{algorithm}[t]
	\caption{PFP-based accelerated PG method for solving $\mathcal{D}1$} 
	\begin{algorithmic}[1]\label{alg:PG_Path}
		\STATE Initialize $\mathbf{w}^{(0)}$ and set $t:=0$.
		\REPEAT 
		\STATE Initialize $\gamma^{(0)}$ and set $l:=0$.
		\REPEAT 
		\STATE Initialize $\mathbf{x}^{(0)}=\mathbf{w}^{(t)}$ and set $i:=0$.
		\REPEAT
		\STATE  Update $\mathbf{x}$ based on~\eqref{eq:acce_PG_moment}-\eqref{eq:tunned_para_PG} and iteration $i:=i+1$.
		\UNTIL Stopping criterion is satisfied.
		\STATE Output the converged point $\mathbf{x}^{\diamond}$.
		\STATE Update iteration $l:=l+1$ and  $\gamma^{(l)}=\frac{\sqrt{R_1^{(t)}(\mathbf{x}^{\diamond})-R_2^{(t)}(\mathbf{x}^{\diamond})}}{\frac{1}{\eta}\mathrm{Tr}\left(\mathbf{x}^{\diamond}(\mathbf{x}^{\diamond})^H\right)+ P_a+KP_c + MP_s}$.
		\UNTIL Stopping criterion is satisfied.
		\STATE Update $\mathbf{w}^{(t+1)}=\mathbf{x}^{\diamond} $ and iteration $t:=t+1$.
		\UNTIL Stopping criterion is satisfied.
	\end{algorithmic}
\end{algorithm}

However, notice that the feasible set of $\mathcal{Q}2$ is a geodesically convex set~\cite{J_vishnoi2018geodesic}, an accelerated Riemannian manifold algorithm~\cite{C_Accele_remani18Alg} can be exploited to further reduce the computation time of the Riemannian CG method. But since the accelerated Riemannian manifold method
requires a concave objective function, we need to transform $\mathcal{Q}2$ into a sequence of geodesically concave subproblems with the $t^{th}$ subproblem given by
\begin{align}\label{obj:opt_geo_concave}
\mathbf{v}^{(t+1)}=\mathop{\arg\max}\limits_{\mathbf{v}}~
\hat{\Upsilon}^{(t)}(\mathbf{v}), ~~~ 
\mathrm{s.t.} ~ [\mathbf{v}\mathbf{v}^H]_{m,m} = 1,~ \forall m,
\end{align}
where $\{\hat{\Upsilon}^{(t)}(\mathbf{v})\}_{t\in\mathbb{N}}$ is a sequence of concave lower bound functions of $\Upsilon(\mathbf{v})$ obtained by the PFP of~\eqref{eq:Path-foPre}, and its explicit expression is shown in~\eqref{eq:low_SCA_ups} of Appendix~\ref{lem:CF_lowerbound_AM}. For a given $t$,~\eqref{obj:opt_geo_concave} can be solved by manifold optimization method with a warm-start initialization $\mathbf{z}^{(0)}:=\mathbf{v}^{(t)}$. In particular, at the $l^{th}$ round update, we compute 
 $r^{(l)}\mathrm{grad} \hat{\Upsilon}^{(t)}(\mathbf{z}^{(l)})$, where $r^{(l)}$ is a step-size satisfying Armijo condition~\cite{J_Abs09Accel_Amoji}, and $\mathrm{grad} \hat{\Upsilon}^{(t)}(\mathbf{z}^{(l)})$ is the Riemannian gradient~\cite{B_AbsilP09}:
\begin{equation}\label{eq: Rem_grad_int}
\mathrm{grad} \hat{\Upsilon}^{(t)}(\mathbf{z}^{(l)})= \nabla \hat{\Upsilon}^{(t)}(\mathbf{v})|_{\mathbf{v}=\mathbf{z}^{(l)}}-\Re\{ \nabla \hat{\Upsilon}^{(t)}(\mathbf{v})|_{\mathbf{v}=\mathbf{z}^{(l)}}\circ ((\mathbf{z}^{(l)})^H)^T\}\circ \mathbf{z}^{(l)},
\end{equation}
with $\nabla \hat{\Upsilon}^{(t)}(\mathbf{v})$ being the Euclidean gradient 
\begin{align}\label{eq:gradient_Euc_e}
\nabla \hat{\Upsilon}^{(t)}(\mathbf{v})=& \frac{2((\mathbf{v}^{(t)})^H\mathbf{A}_k\mathbf{w})^H\mathbf{A}_k\mathbf{w}}{\sigma^2_k\ln 2}-\frac{2(\mathbf{v}^H\mathbf{A}_k\mathbf{w})^H\mathbf{A}_k\mathbf{w}|(\mathbf{v}^{(t)})^H\mathbf{A}_k\mathbf{w}|^2}{(\sigma^2_k+|(\mathbf{v}^{(t)})^H\mathbf{A}_k\mathbf{w}|^2)\sigma^2_k\ln 2}\nonumber \\
&-\frac{2\kappa_{1,j}(\mathbf{w}^H\hat{\mathbf{H}}^H\mathbf{v})^H\hat{\mathbf{H}}\mathbf{w}\ln\varepsilon_k^{-1}}{(1+(\kappa_{1,j}|\mathbf{w}^H\hat{\mathbf{H}}^H\mathbf{v}^{(t)}|^2+\kappa_{2,j}\|\mathbf{w}\|^2)\ln\varepsilon_k^{-1})\ln2}.
\end{align}
Then, the exponential mapping is employed to guarantee the updated point $\mathbf{z}^{(l+1)}$ stays in the oblique manifold~\cite[eq. 4.31]{B_AbsilP09}:
\begin{equation}\label{eq:update_v_stanRO}
\mathbf{z}^{(l+1)} :=\mathrm{exp}_{\mathbf{z}^{(l)}}\left(r^{(l)}\mathrm{grad} \hat{\Upsilon}^{(t)}(\mathbf{z}^{(l)})\right),
\end{equation}
where $\mathrm{exp}_{\mathbf{z}^{(l)}}(\mathbf{c}^{(l)})$ is defined as
$\mathrm{exp}_{\mathbf{z}^{(l)}}(\mathbf{c}^{(l)})= \mathbf{z}^{(l)}\cos(\|\mathbf{c}^{(l)}\|)+\frac{\mathbf{c}^{(l)}}{\|\mathbf{c}^{(l)}\|}\sin(\|\mathbf{c}^{(l)}\|)$.

Iteratively executing~\eqref{eq: Rem_grad_int} and~\eqref{eq:update_v_stanRO} results in the standard manifold optimization.
However, since~\eqref{obj:opt_geo_concave} is strongly geodesically concave and the inverse of $\mathrm{exp}_{\mathbf{z}^{(l)}}(\mathbf{c}^{(l)})$ is well-defined, momentum can be employed to speed up the convergence~\cite{C_Accele_remani18Alg}.
In particular, by introducing an additional variable $\mathbf{d}^{(l)}$ on the manifold, an interpolated point $\mathbf{q}^{(l)}$ is computed as
\begin{equation}\label{eq:Remin_accele_Q}
\mathbf{q}^{(l)}=\mathrm{exp}_{\mathbf{z}^{(l)}}\left(\frac{\left(\sqrt{\beta^2+4(1+\beta)u_g\mathcal{S}}-\beta\right)\mathrm{exp}^{-1}_{\mathbf{z}^{(l)}}(\mathbf{d}^{(l)})}{2+\sqrt{\beta^2+4(1+\beta)u_g\mathcal{S}}+\beta}\right),
\end{equation}
where $\mathrm{exp}^{-1}_{\mathbf{z}^{(l)}}(\mathbf{d}^{(l)})$ is the inverse of the exponential map and defined as~\cite{J_HuangWen2021Rpgm}
\begin{equation}\label{eq:invers_map_exp}
\mathrm{exp}^{-1}_{\mathbf{z}^{(l)}}(\mathbf{d}^{(l)})=\frac{\cos^{-1}((\mathbf{z}^{(l)})^T\mathbf{d}^{(l)})}{\sqrt{1-((\mathbf{z}^{(l)})^T\mathbf{d}^{(l)})^2}}(\mathbf{I}-\mathbf{z}^{(l)}(\mathbf{z}^{(l)})^T)\mathbf{d}^{(l)},
\end{equation}
which projects $\mathbf{d}^{(l)}$ from the manifold onto the tangent space. In~\eqref{eq:Remin_accele_Q}, $\mathcal{S}$ is a step-size satisfying $\mathcal{S}<1/L_g$ with
Lipschitz gradient constant $L_g$, $\beta>0$ is the shrinkage parameter, and $u_g$ is a strong concavity constant for $\hat{\Upsilon}^{(t)}(\mathbf{v})$ that can be easily obtained based on~\cite[Definition 2]{C_pmlrv49zhang16b}. 
From $\mathbf{q}^{(l)}$ in~\eqref{eq:Remin_accele_Q}, the updated points $\mathbf{d}^{(l+1)} $ and $\mathbf{z}^{(l+1)} $ on the manifold are calculated as~\cite{C_Accele_remani18Alg}
\begin{align}
\mathbf{d}^{(l+1)} 
=&~\mathrm{exp}_{\mathbf{q}^{(l)}}\Bigg(\frac{u_g(2-\sqrt{\beta^2+4(1+\beta)u_g\mathcal{S}}+\beta)\mathrm{exp}^{-1}_{\mathbf{q}^{(l)}}(\mathbf{d}^{(l)})}{2u_g(1+\beta)} \nonumber \\
&-\frac{(\sqrt{\beta^2+4(1+\beta)u_g\mathcal{S}}+\beta)\mathrm{grad}\hat{\Upsilon}^{(t)}(\mathbf{q}^{(l)})}{2u_g(1+\beta)}\Bigg),\label{eq:acce_Auxili_D_remina}\\
\mathbf{z}^{(l+1)}=&~\mathrm{exp}_{\mathbf{q}^{(l)}}\left(-\mathcal{S}\mathrm{grad} \hat{\Upsilon}^{(t)}(\mathbf{q}^{(l)}) \right). \label{eq:acce_RM_z}
\end{align}
It is now clear that $\mathbf{d}^{(l)}$ represents the momentum, which affects the momentum at the $(l+1)^{th}$ iteration via~\eqref{eq:acce_Auxili_D_remina}, and the optimization variable $\mathbf{z}^{(l+1)}$ via~\eqref{eq:acce_RM_z} together with~\eqref{eq:Remin_accele_Q}.

To sum up, the entire procedure for solving $\mathcal{Q}2$ is summarized in Algorithm~\ref{alg:CG_obligueManifold}, where PFP is applied over $t$ for the outer iterations, and each subproblem~\eqref{obj:opt_geo_concave} is solved with accelerated Riemannian manifold over $l$ for the inner iterations. The convergence property is revealed by Theorem~\ref{theorm: convergence_SCAmanifold}, which is proved in Appendix~\ref{Proof:Theorm_4}. 
\begin{theorem}\label{theorm: convergence_SCAmanifold}
	Starting from a feasible point of $\mathbf{v}$, the sequence of solutions generated from Algorithm~\ref{alg:CG_obligueManifold} converges to a stationary point of $\mathcal{Q}2$.
\end{theorem}
\begin{algorithm}[H]
	\caption{Accelerated manifold optimization for solving $\mathcal{Q}2$} 
	\begin{algorithmic}[1]\label{alg:CG_obligueManifold}
		\STATE Initialize $\mathbf{v}^{(0)}$, $\mathcal{S}<1/L_g$, $\beta > 0$ and set $t:=0$.
		\REPEAT 
		\STATE Initialize $\mathbf{z}^{(0)}=\mathbf{v}^{(t)}$, $\mathbf{d}^{(0)}=\mathbf{z}^{(0)}$, and set $l:=0$.
		\REPEAT
		\STATE  Update $\mathbf{z}$ based on~\eqref{eq:Remin_accele_Q}-\eqref{eq:acce_RM_z} and iteration $l:=l+1$.
		\UNTIL Stopping criterion is satisfied.
		\STATE Update $\mathbf{v}^{(t+1)}=\mathbf{z}^{\diamond} $ with the converged point $\mathbf{z}^{\diamond}$ and iteration $t:=t+1$.
		\UNTIL Stopping criterion is satisfied.
	\end{algorithmic}
\end{algorithm}

The framework of the accelerated manifold is shown in Fig.~\ref{fig:algo_alg4_mani}. 
Compared to standard Riemannian manifold optimization (only involves computing Riemannian gradient and performing exponential mapping in
 Fig.~\ref{fig:algo_alg4_mani}), which has a convergence rate $\mathcal{O}\left({L_g}/{u_g}\log(1/\varphi)\right)$, the proposed accelerated version achieves a faster convergence rate of $\mathcal{O}(\sqrt{{L_g}/{u_g}}\log(1/\varphi))$ per iteration, where $\varphi$ is the target solution accuracy~\cite{C_Accele_remani18Alg}. 
\begin{figure}[tb]
	\centering
	\includegraphics[scale=0.35]{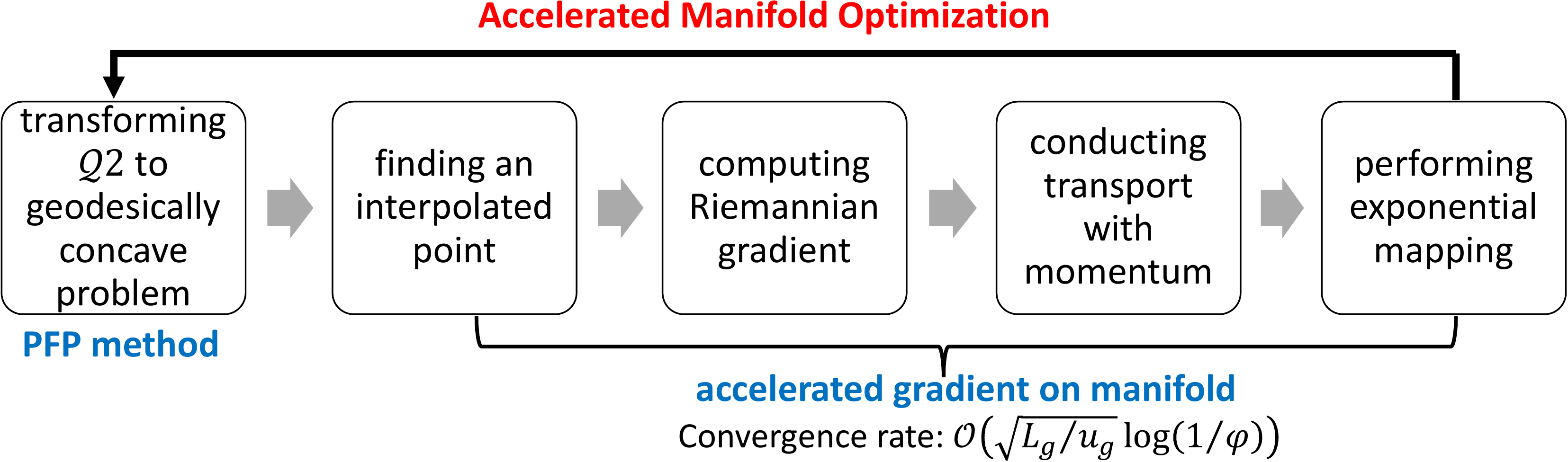}
	\caption{The framework of the accelerated Riemannian manifold optimization.} \label{fig:algo_alg4_mani} 
\end{figure}

\vspace{-0.5cm}
\subsection{Overall Algorithm for Solving $\mathcal{P}2$}
With the alternative update of $\mathbf{w}$ and $\mathbf{\Theta}$ given by Algorithms~\ref{alg:PG_Path} and~\ref{alg:CG_obligueManifold}  respectively, problem $\mathcal{P}2^{[k,j]}$ can be efficiently solved under the AM framework, and the convergence property is revealed by the following theorem, which is proved in Appendix~\ref{proof_tem_convergence}. 
\begin{theorem}\label{tem:conver_AM}
	Starting from a feasible solution of problem $ \mathcal{P}2^{[k,j]}$, the sequence of solutions generated by alternatively executing Algorithms~\ref{alg:PG_Path} and~\ref{alg:CG_obligueManifold} converges to a stationary point of problem $ \mathcal{P}2^{[k,j]}$.
\end{theorem} 
Noticing that $\mathcal{P}2$ consists of $KJ$ parallel subproblems in the form of $ \mathcal{P}2^{[k,j]}$, the overall algorithm for solving $\mathcal{P}2$ can be implemented in a parallel manner and is summarized in Algorithm~\ref{Overall_Algo:Sove_P1}, where the multi-core computing architecture can be leveraged for speeding up the computation.
\begin{algorithm}[H]
	\caption{Overall algorithm for solving $\mathcal{P}2$ } 
	\begin{algorithmic}[1]\label{Overall_Algo:Sove_P1}		
		\STATE Solve $\mathcal{P}2^{[k,j]}$ in a parallel manner for all $k,j$ by alternatively executing Algorithms~\ref{alg:PG_Path} and~\ref{alg:CG_obligueManifold}. 
		\STATE Compute the optimized objective function value of~\eqref{obj:P1_j_parallel} for all $k,j$.
		\STATE Select the minimum among $KJ$ objective function values of $\mathcal{P}2^{[k,j]}$.
	\end{algorithmic}
\end{algorithm}

For the computational complexity of Algorithm~\ref{Overall_Algo:Sove_P1}, it is dominated by the AM iteration in step 1. To be specific, Algorithm~\ref{alg:PG_Path} consists of the inner iterations over $i$, and the outer iterations over $l$ and $t$. 
Notice that the iteration with respect to $i$ is based on the accelerated PG method (i.e., step 7 in Algorithm~\ref{alg:PG_Path}), which only involves first-order differentiation. Therefore, it has $\mathcal{O}(N/\varrho)$ complexity order with an accuracy of $\varrho$~\cite{B_Bertseka97NP}. 
Together with the outer iterations, the complexity order of Algorithm~\ref{alg:PG_Path} is $\mathcal{O}(\mathcal{M}_3N/\varrho)$, where $\mathcal{M}_3$ is the outer iteration number for Algorithm~\ref{alg:PG_Path} to converge. 
On the other hand, Algorithm~\ref{alg:CG_obligueManifold} includes the inner iteration over $l$ and the outer iteration over $t$. The inner iteration is based on accelerated Riemannian manifold, which needs $\mathcal{O}(M\log(1/\varphi))$ operations for each outer iteration. Combined with the outer iteration, the complexity order of Algorithm~\ref{alg:CG_obligueManifold} is $\mathcal{O}(\mathcal{M}_4M\log(1/\varphi))$, where $\mathcal{M}_4$ is the outer iteration number.

Based on the above discussion, the complexities order for different algorithms are summarized in Table~\ref{tab:OptimizationMethodTable}, where $\mathcal{M}_1$ and $\mathcal{M}_2$ are the numbers of outer iterations for Algorithms~\ref{alg:conventional_D1} and~\ref{alg:conventional_Q1} to converge, respectively.\footnote{$\mathcal{M}_{i},i\in\{1,\ldots,4\}$ are within the same order of magnitude.}
\begin{table}[!htbp]
	\centering
	\caption{Summary of complexities of different algorithms}\label{tab:OptimizationMethodTable}
	\begin{tabular}{|c|c|c|}
\hline
		Type of Methods &Proposed Algorithms & Complexity Order \\
\hline
		\multirow{2}*{conventional method} &Algorithm~\ref{alg:conventional_D1}&   $\mathcal{O}\left(\mathcal{M}_1N^3/\varrho\right)$\\
	& Algorithm~\ref{alg:conventional_Q1} &$\mathcal{O}(\mathcal{M}_{2}\sqrt{M}(2M^4+M^3)\log(1/\varphi))$\\
	 \cline{1-3}
	\multirow{2}*{accelerated first-order method} 	&Algorithm~\ref{alg:PG_Path}&   $\mathcal{O}\left(\mathcal{M}_3N/\varrho\right)$\\
	& Algorithm~\ref{alg:CG_obligueManifold} &$\mathcal{O}(\mathcal{M}_4M\log(1/\varphi))$\\
\hline
	\end{tabular}
\end{table}
Compared with Algorithms~\ref{alg:conventional_D1} and~\ref{alg:conventional_Q1}, the complexity of accelerated first-order method is linear in $N$ or $M$. Hence, they are suitable for massive antennas and massive reflecting elements networks.

\noindent

\section{Simulation Results and Discussion}\label{Sec:VI}
In this section, we evaluate the secure transmission performance of the proposed algorithm through simulations. All problem instances are simulated using MATLAB R2017a on a Windows x64 desktop with 3.2 GHz CPU and 16 GB RAM, and the simulation results are obtained via averaging over 500 simulation trials, with independent users' and Eves' locations, channels, and noise realizations in each trial. Unless otherwise specified, the simulation set-up is as follows and kept throughout this section. 
There are 5 users and 10 Eves in the whole system, 
The BS and RIS are located at (0 m, 0 m), (50 m, 0 m).
All users are uniformly and randomly distributed in a circle centered at (50 m, 20 m) with radius 5 m. 
The distance from Eves to the RIS is randomly generated between 1 m and 10 m, and the path-loss exponent is set according to the 3GPP propagation environment~\cite{J_IRS_EE19Huang}.
As a result, the path-loss coefficients $\{\alpha_{1},\alpha_{r,k},\alpha_{d,k},\alpha_{r,j},\alpha_{d,j}\}$ can be respectively obtained based on the signal propagation model~\cite{J_Andes95_Prog}. Once the large-scale fading parameters are generated, they are assumed to be known and fixed throughout the simulations. For small-scale fading, the channel coefficients are generated from identically distributed and circularly complex Gaussian random variables, which are chosen from $\mathcal{CN}(0,1)$~\cite{J_Chu21IRSecure}.
The bandwidth is set to be $10$ MHz and the noise power spectral density for users and Eves are $\sigma^2_k=\sigma^2_j=-96$ dBm/Hz.
The power amplifier efficiency is $\eta=0.311$. The circuit power at the BS and the user are $P_a=39$ dBm and $P_c=20$ dBm, respectively~\cite{J_Archiyou20IRS_EE}. The static hardware-dissipated power at each reflecting element is $P_s = 10$ dBm~\cite{J_Archiyou20IRS_EE}. Besides, it is assumed that the upper bound of SOP for all users are set to be equal, i.e., $\varepsilon_k=\varepsilon$.
To avoid repeating figure descriptions, the settings for $(M,N,K,J,\varepsilon,P_\mathrm{max})$ are provided in the caption of each figure. 

\begin{figure*} 
	\centering
	\subfigure[The iterations of Algorithm~\ref{alg:PG_Path}.]{ 
		\label{fig:Converge_Alg1}
		\includegraphics[width=2in]{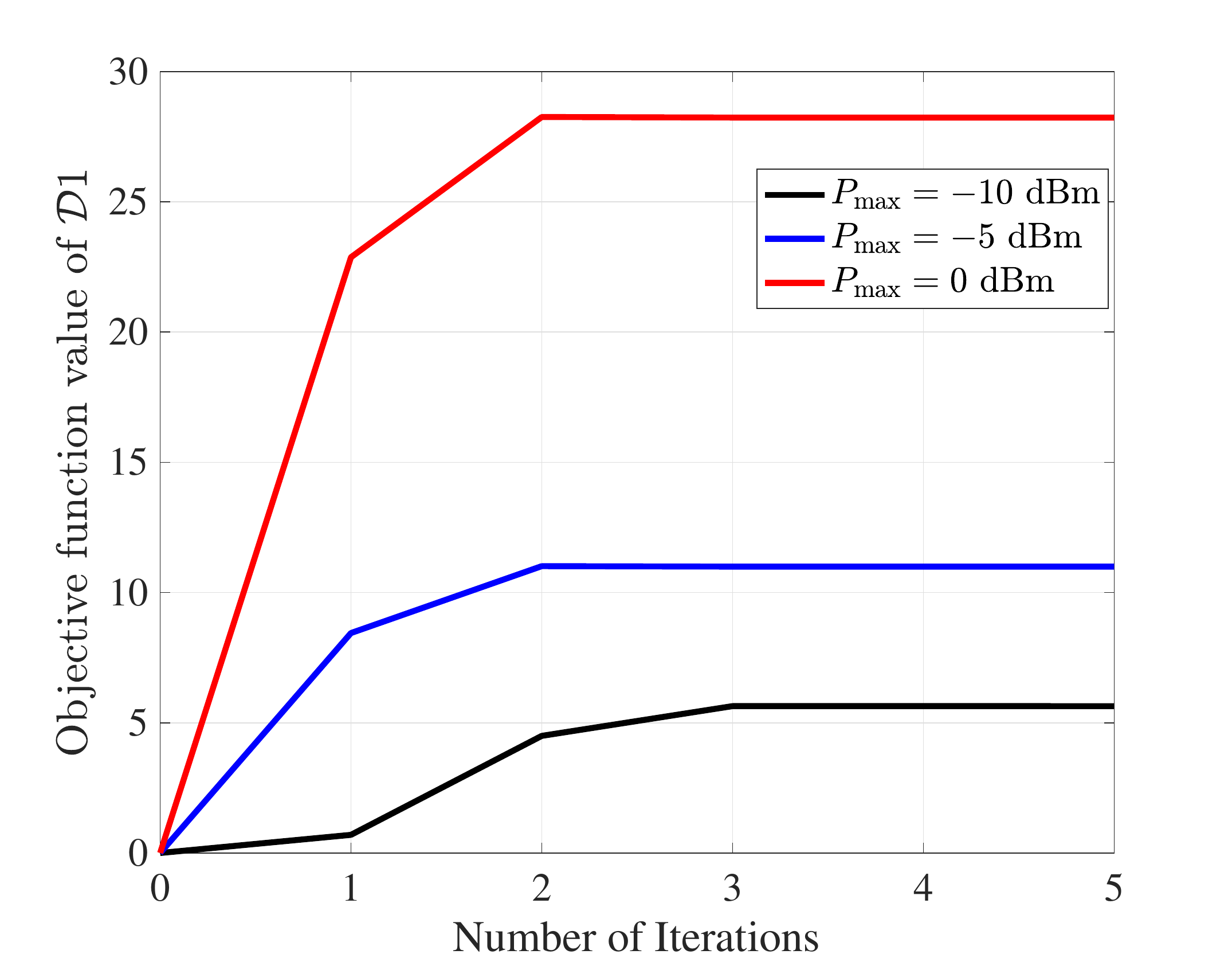}} \hspace{0in} 
	\subfigure[The iterations of Algorithm~\ref{alg:CG_obligueManifold}.]{
		\label{fig:Conver_alg2} 
		\includegraphics[width=2in]{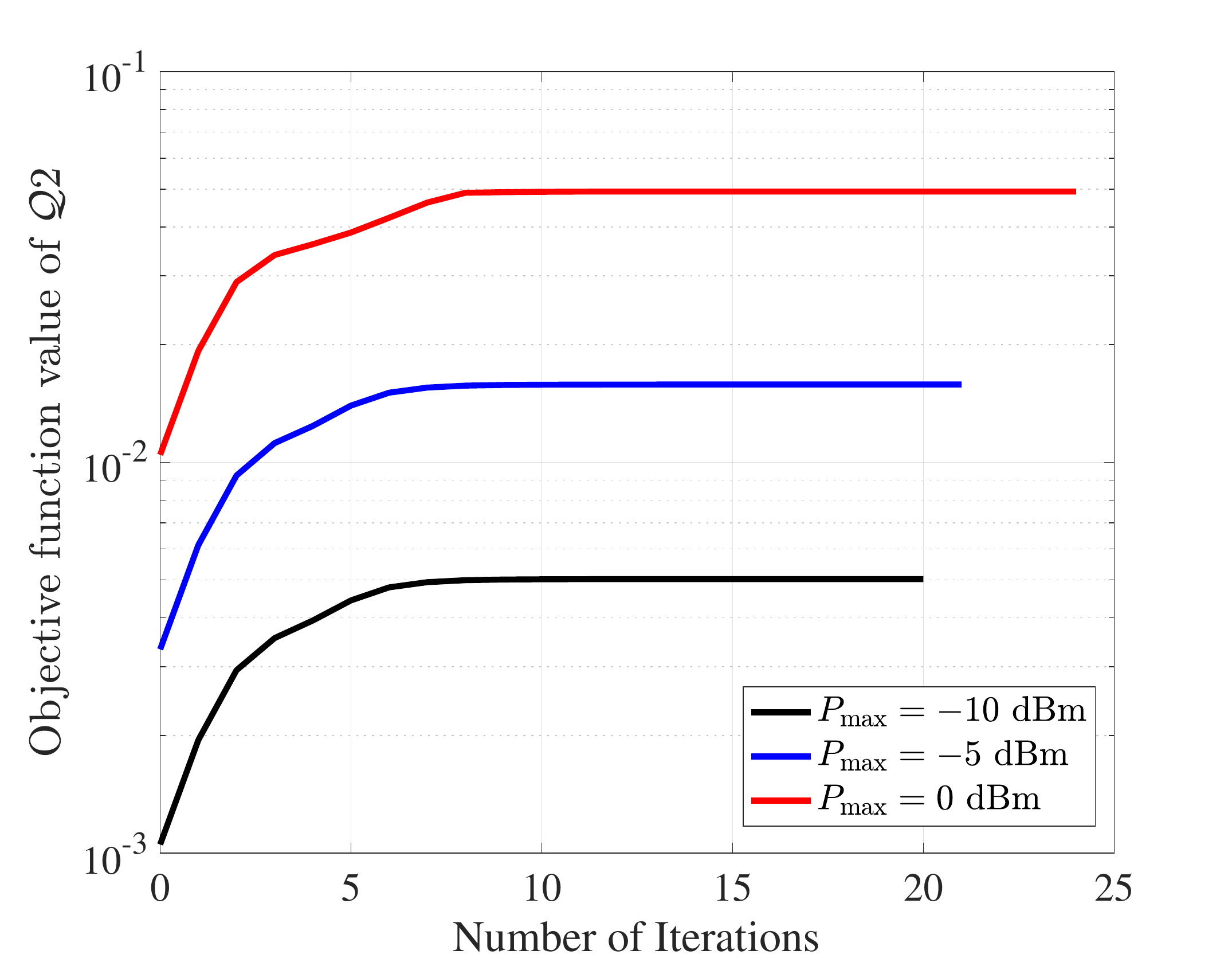}}  \hspace{0in} 
	\subfigure[The AM iterations in Algorithm~\ref{Overall_Algo:Sove_P1}.]{
		\label{fig:Conver_alg3} 
		\includegraphics[width=2in]{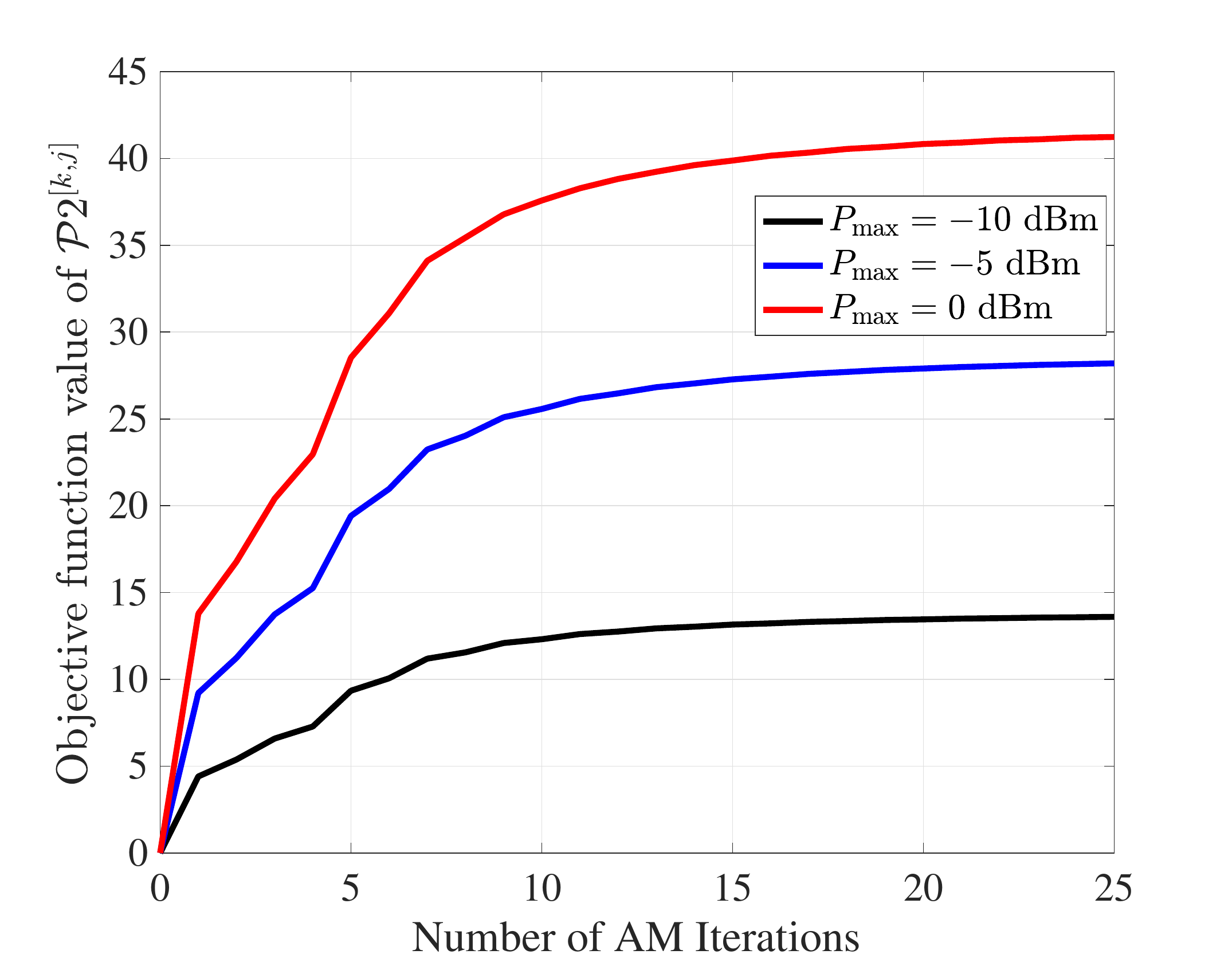}} 
	\caption{{Convergence behaviour of algorithms with $M=10$, $N=10$, $K=5$, $J=10$, $\varepsilon=0.1$.}}  
\end{figure*}
\subsection{Performance of the Proposed Algorithm}
First, we demonstrate the convergence behaviour of Algorithms~\ref{alg:PG_Path} and~\ref{alg:CG_obligueManifold}.  
Notice that Algorithm~\ref{alg:PG_Path} consists of multiple layers of iterations. The stopping criterion for each layer depends on the relative change of the corresponding objective function (e.g., less than $10^{-4}$), and the convergence results for updating $\mathbf{w}$ under fixed $\mathbf{\Theta}$ are shown in Fig.~\ref{fig:Converge_Alg1}. 
It is observed that Algorithm~\ref{alg:PG_Path} achieves fast convergence under different values of $P_\mathrm{max}$, which corroborates the results in Theorem~\ref{theorm: convergence_property}.
When fixing the number of reflecting elements $M$ and transmit antennas $N$, the objective function value of $\mathcal{D}1$ increases with $P_\mathrm{max}$.
On the other hand, the convergence of Algorithm~\ref{alg:CG_obligueManifold} is shown in Fig.~\ref{fig:Conver_alg2}. It can be seen that the accelerated Riemannian manifold algorithm for solving $\mathcal{Q}2$ converges within 20 iterations under different values of $P_\mathrm{max}$, which corroborates the results in Theorem~\ref{theorm: convergence_SCAmanifold}.
To verify the overall convergence of Algorithm~\ref{Overall_Algo:Sove_P1} by alternatively executing Algorithms~\ref{alg:PG_Path} and~\ref{alg:CG_obligueManifold}, Fig.~\ref{fig:Conver_alg3} shows the objective function value of $\mathcal{P}2^{[k,j]}$ versus the AM iteration.
It can be seen that the AM converges rapidly within 25 iterations under different values of $P_\mathrm{max}$, which corroborates the convergence result of Theorem~\ref{tem:conver_AM}.
\begin{figure*} 
	\centering
	\subfigure[Average computation time versus $N$ with $M=10$, $K=5$, $J=10$.]{ 
		\label{fig:time_compare}
		\includegraphics[width=2.8in]{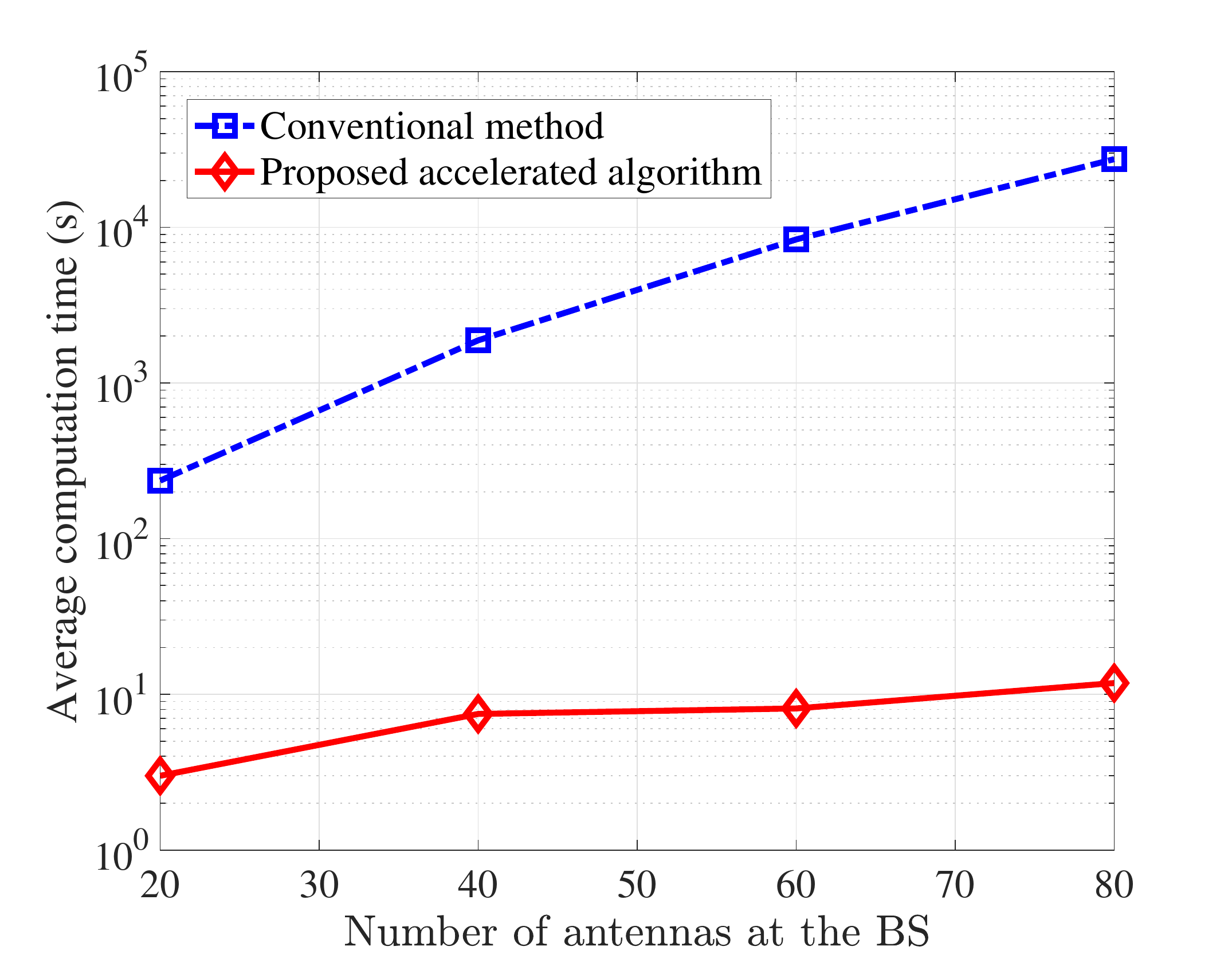}} \hspace{0.0in} 
		\subfigure[Average computation time versus $M$ with $N=10$, $K=5$, $J=10$.]{ 
		\label{fig:time_compare_M}
		\includegraphics[width=2.8in]{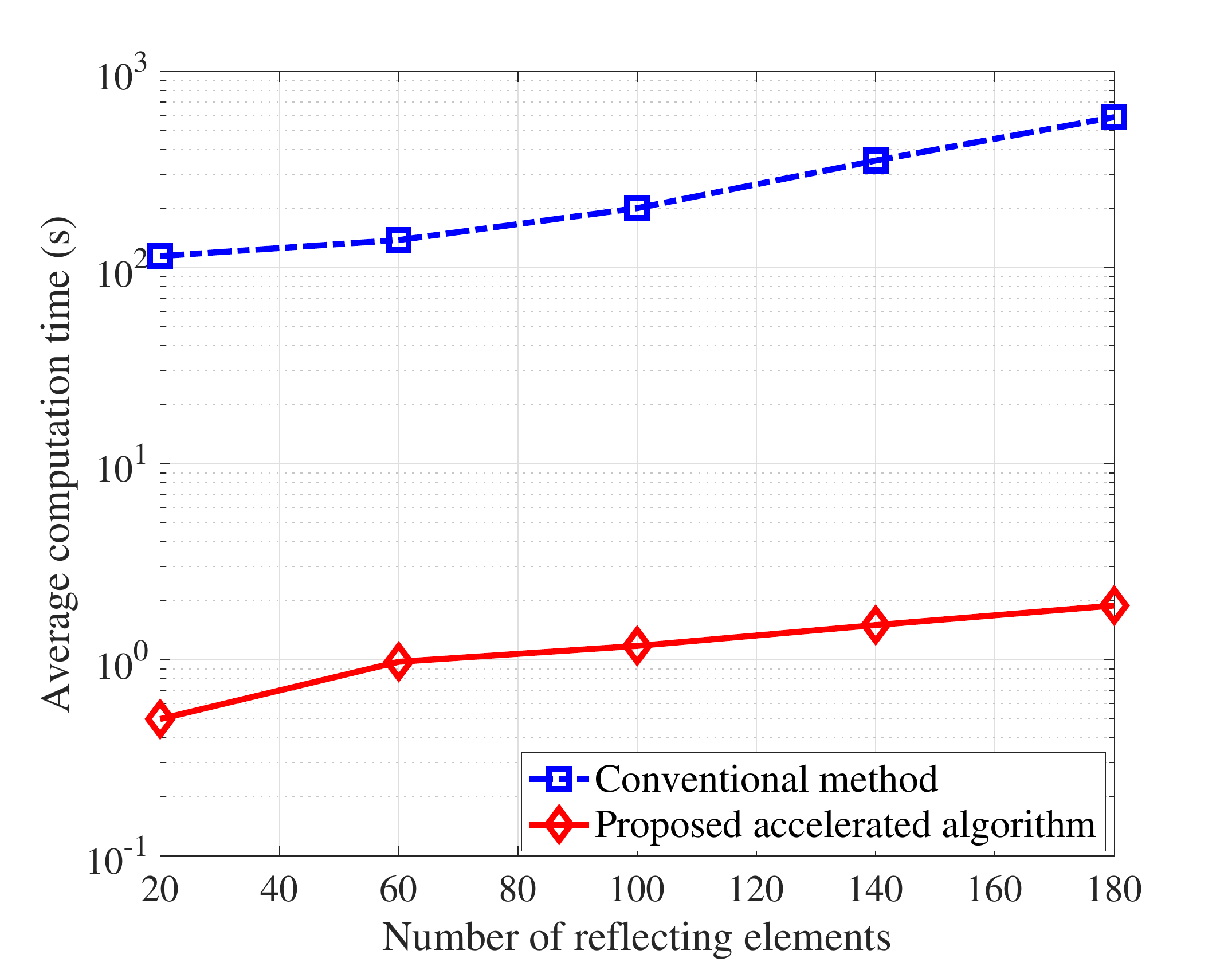}} \hspace{0.0in} 
	\subfigure[Average secure EE versus $N$ with $M=10$, $K=5$.]{
		\label{fig:Perf_alg_compare} 
		\includegraphics[width=2.8in]{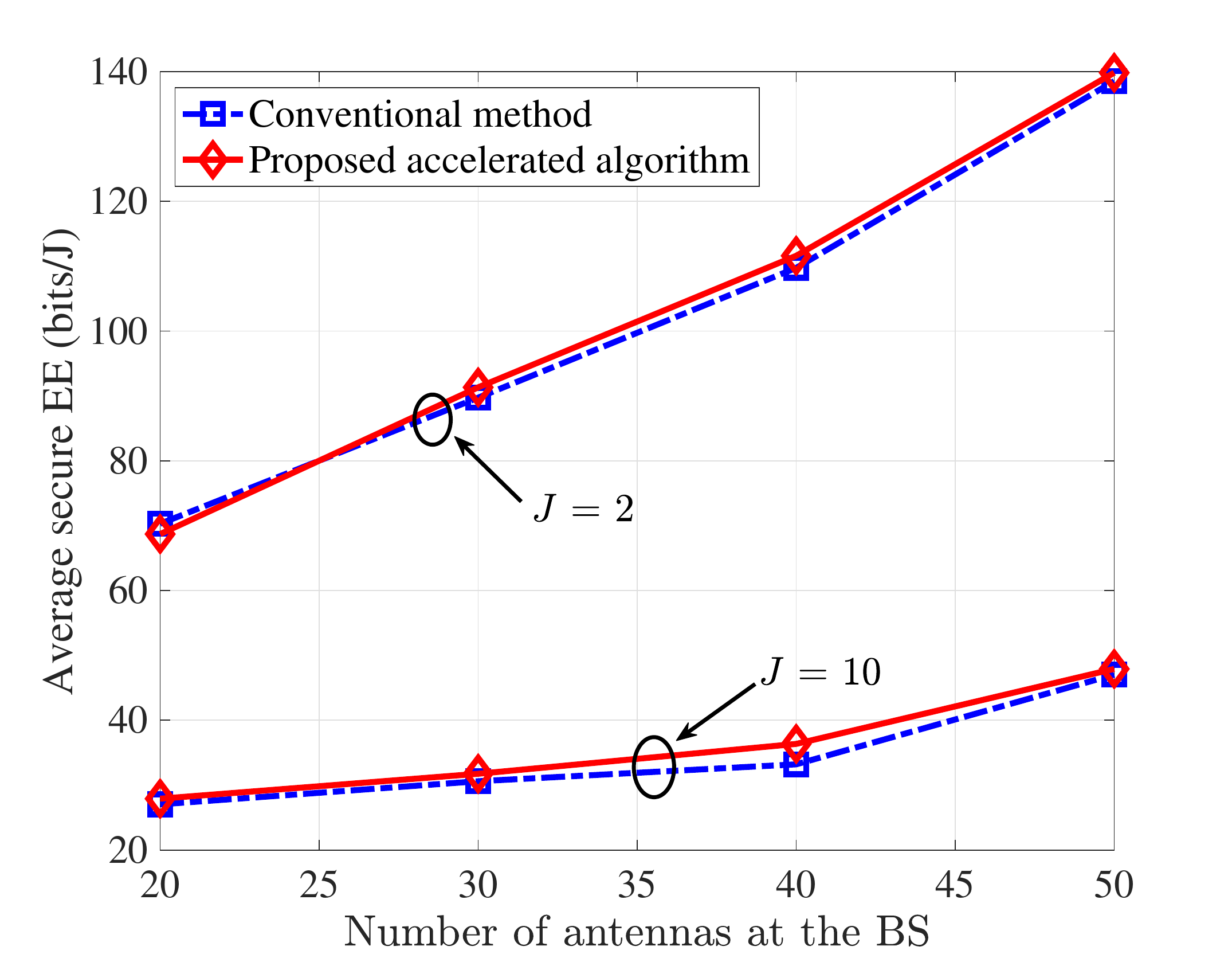}}
	\subfigure[Average secure EE versus $M$ with $N=10$, $J=10$.]{
		\label{fig:Perf_alg_compare_M} 
		\includegraphics[width=2.8in]{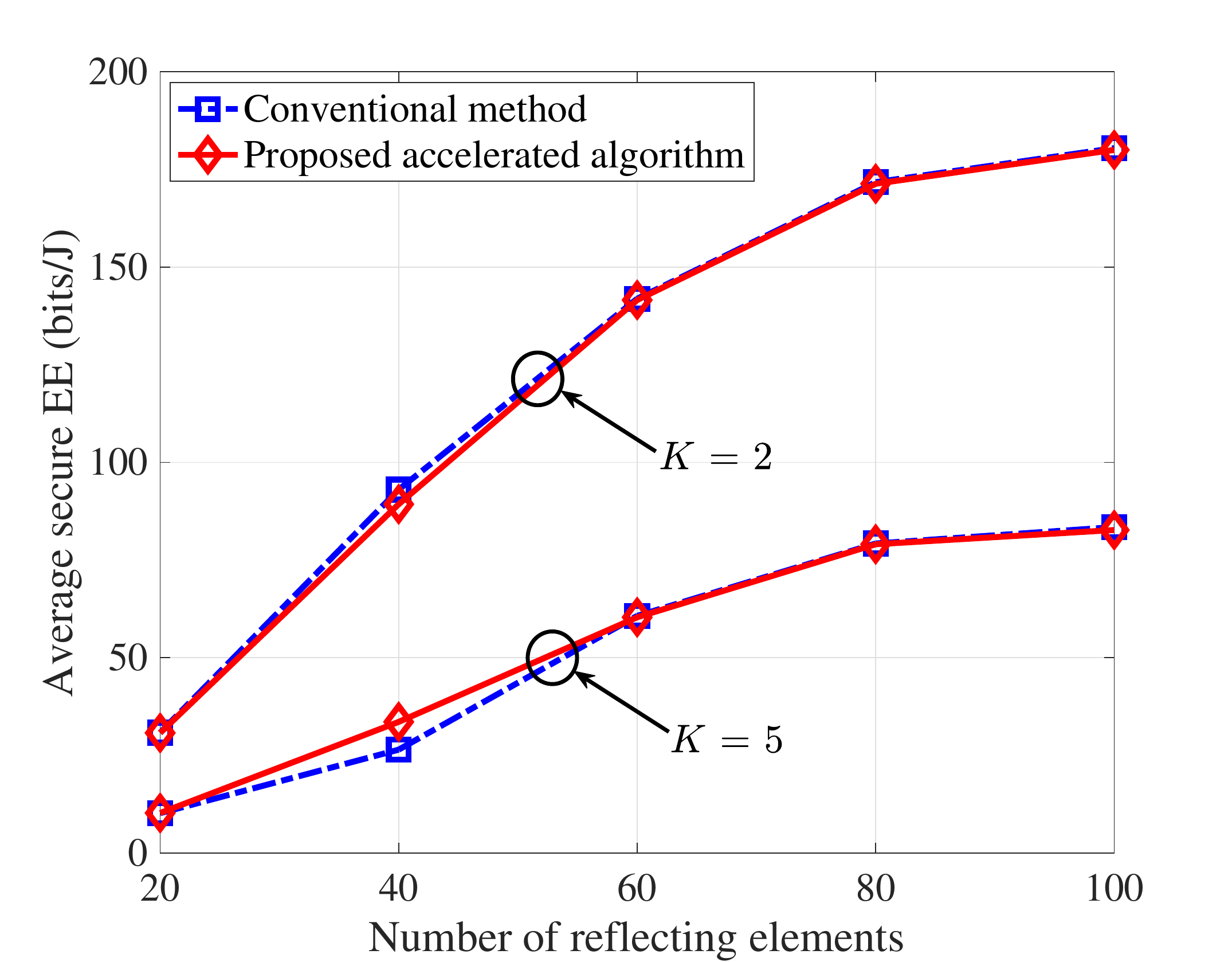}}  
	\caption{{Performance comparison with the conventional method with $\varepsilon=0.5$ and $P_\mathrm{max}=0$ dBm.}}  
\end{figure*}

Next, to show the complexity advantage of the proposed Algorithm~\ref{Overall_Algo:Sove_P1} for solving $\mathcal{P}2$, we compare it with the conventional method by alternatively executing Algorithms~\ref{alg:conventional_D1} and~\ref{alg:conventional_Q1} to solve each subproblem of $\mathcal{P}2$. The convergence tolerance and maximum number of iterations for the conventional method are set to $10^{-4}$ and 50, respectively. 
As shown in Fig.~\ref{fig:time_compare} and Fig.~\ref{fig:time_compare_M}, with the increase of transmit antennas at the BS and reflecting elements in the RIS, the proposed Algorithm~\ref{Overall_Algo:Sove_P1} reduces the computation time by at least two orders of magnitude compared with the conventional method and the advantage becomes more prominent as the number of transmit antenna or reflecting element increases. Since the proposed accelerated first-order algorithms replace the steps in conventional method that uses the interior-point method, the number of layers of iterative processes does not increase but the computation time is significantly reduced.
On the other hand, Fig.~\ref{fig:Perf_alg_compare} and Fig.~\ref{fig:Perf_alg_compare_M} show that the proposed Algorithm~\ref{Overall_Algo:Sove_P1} achieves almost the same average secure EE as the conventional method. 

Furthermore, with an increase of $J$, the average secure EE is degraded due to more Eves in the system as shown in Fig.~\ref{fig:Perf_alg_compare}.
With the increase of $M$, the average secure EE is increasing under different values of $K$ as shown in Fig.~\ref{fig:Perf_alg_compare_M}.
However, the average secure EE decreases when the number of users increases, since the overall performance is determined by the worse case user in network.
Since the degrees of freedom increase with the increase of antennas or reflecting elements, the average secure EE is increasing in $N$ or $M$. 
Due to the computational complexity advantage of the accelerated first-order algorithm, we only provide the solution to $\mathcal{P}2$ obtained via Algorithm~\ref{Overall_Algo:Sove_P1} in the following discussion.

\subsection{Performance Comparison with Other RIS Schemes}
		\begin{figure}[tb]
		\centering
		\includegraphics[scale=0.35]{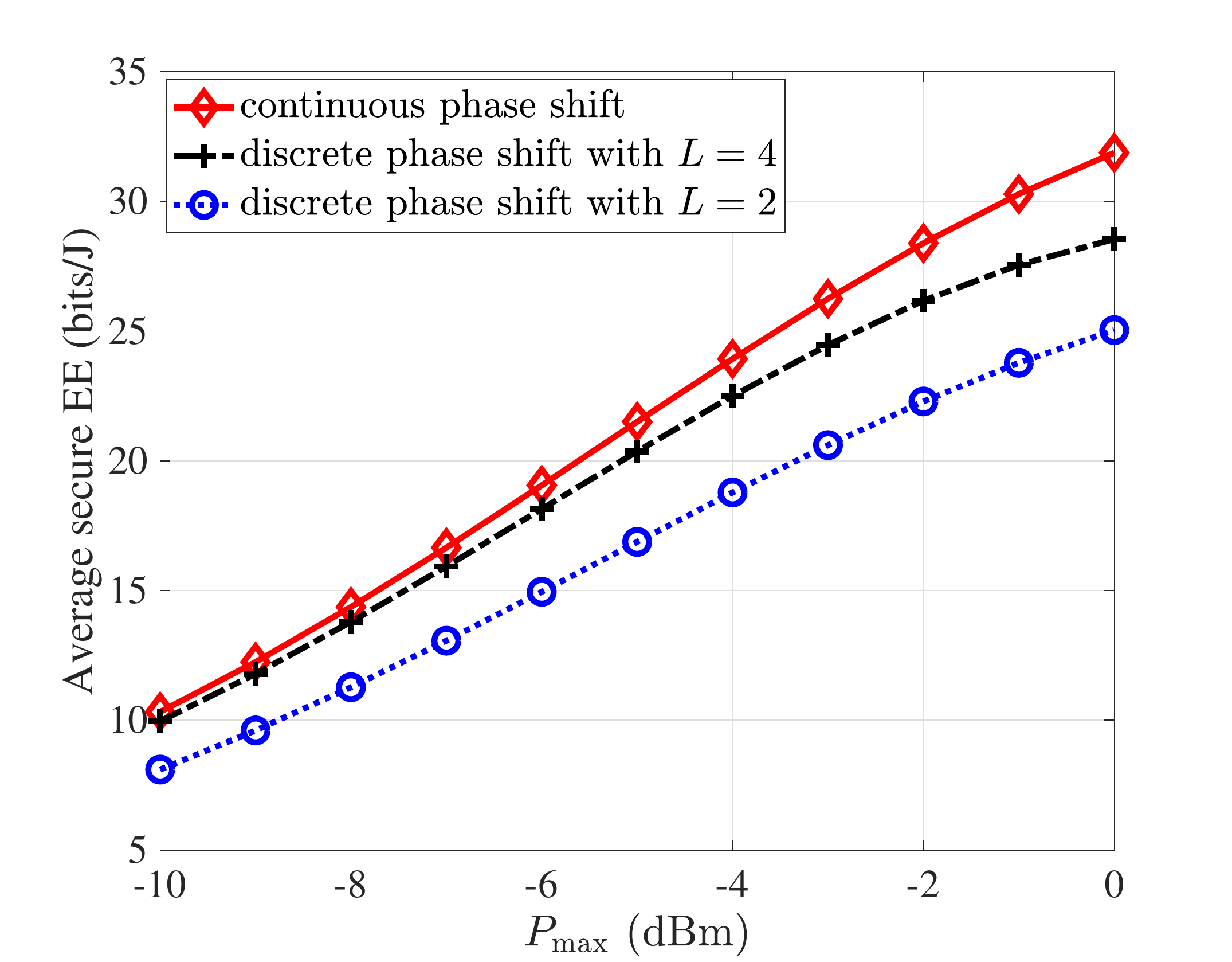}
		\caption{{Comparison with different quantization levels $L$ with $M = 8$, $N = 10$, $K=5$, $J=10$, $\varepsilon=0.1$.}} \label{fig:ObjValue_Q2} 
	\end{figure}
First, we demonstrate the performance of RIS with discrete phase shifts, where the average secure EE versus transmit power under different numbers of allowable phase shifts $L$ is shown in Fig.~\ref{fig:ObjValue_Q2}. For the discrete case, we simply take the optimized RIS phases from Algorithm~\ref{Overall_Algo:Sove_P1} and then project them to the nearest allowable values. It is observed that with the increase of $L$, the average secure EEs approach that of the continuous phase shift case due to a finer resolution of reflecting coefficients of RIS. Furthermore, we can observe from Fig.~\ref{fig:ObjValue_Q2} that the case of $L=4$ is already close to the continuous phase case. It is expected that the quantization loss will be insignificant when $L>4$.

Next, we compare the proposed RIS scheme with a random phase-shift (RPS) baseline scheme, where the phase-shift coefficients are randomly generated with equal probability and Algorithm~\ref{Overall_Algo:Sove_P1} is applied without updating the phase-shift of RIS. To make a fair comparison with the RPS scheme~\cite{J_Chen19IRS_Secure}, we simulate all schemes under the same security requirement and select $K=J=1$.
Furthermore, we compare the proposed scheme with the fixed phase-shift (FPS) baseline scheme~\cite{J_Archiyou20IRS_EE} that only optimizes the beamforming vector, i.e., $\mathbf{\Theta}=\mathbf{I}_M$.
To show the importance of considering the imperfect CSI of Eves' channels, we also compare the proposed scheme with the RIS aided secure transmission that ignoring CSI uncertainty. 
Without loss of generality, only the results of continuous phase shift is presented here, as it can be regarded as an upper bound of the discrete phase case.

To begin with, we illustrate the impact of SOP constraint on the system performance, where all average secure EEs are increasing in $\varepsilon$ as shown in Fig.~\ref{fig:SeEE_epsilon}, and the proposed scheme always achieves a significantly higher average secure EE than other baseline schemes and the RIS ignoring CSI uncertainty. 
Besides, the FPS scheme has a better performance than the RPS scheme, since some optimized phase shifts are quite close to 1. 
When comparing RIS schemes under different values of $P_\mathrm{max}$, Fig.~\ref{fig:Se_EE_Power} shows that the proposed scheme significantly improves the average secure EE. Moreover, the performance gaps between the proposed solution and other baseline schemes increase with an increase of $\varepsilon$ or $P_\mathrm{max}$. The results from Fig.~\ref{fig:perform_otherBaLiScheme} validate the strength of the proposed RIS aided transmission scheme. 
\begin{figure*}	\label{fig:SeEE_baselines}
	\centering
	\subfigure[Average secure EE versus $\varepsilon$: $P_\mathrm{max}=0$ dBm.]{ 
		\label{fig:SeEE_epsilon}
		\includegraphics[width=3in]{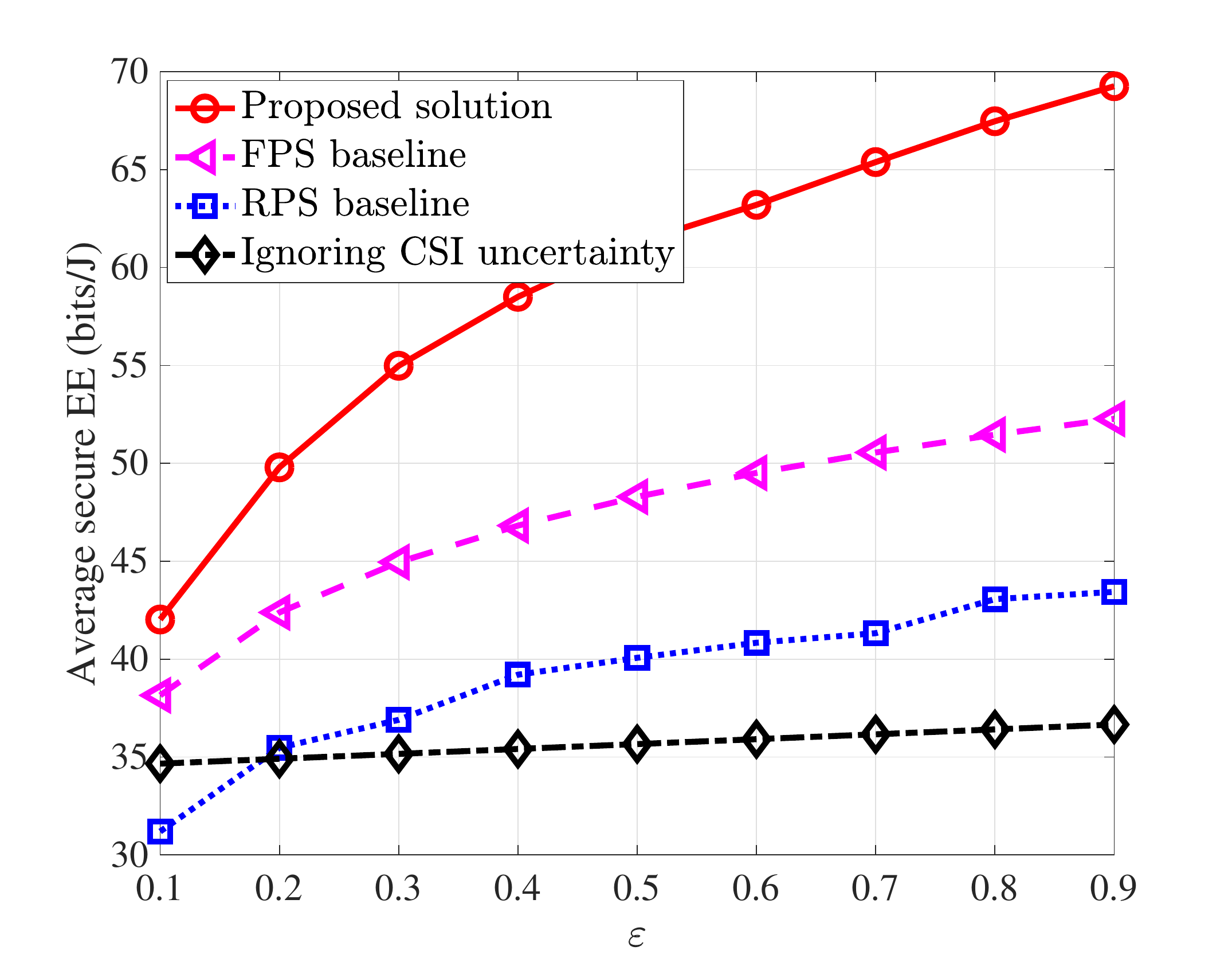}} \hspace{0.0in} 
	\subfigure[Average secure EE versus $P_\mathrm{max}$: $\varepsilon=0.5$.]{
		\label{fig:Se_EE_Power} 
		\includegraphics[width=3in]{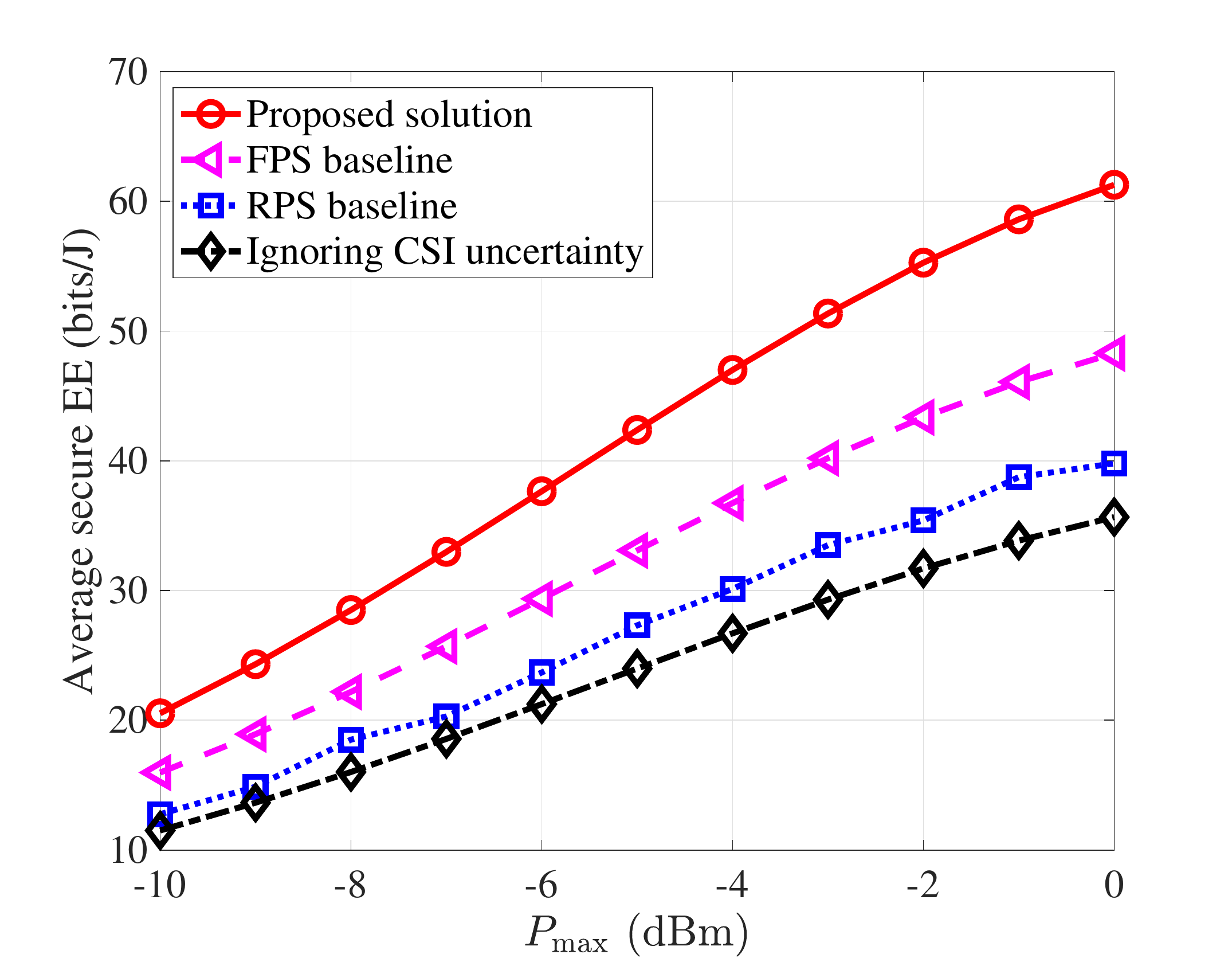}} \hspace{0.0in} 
	\caption{Comparison with other RIS schemes with $M = 10$, $N = 10$, $K=J=1$.} \label{fig:perform_otherBaLiScheme}
\end{figure*}

\section{Conclusion}\label{Sec:VII}
This paper studied the massive RIS aided secure multicast transmission with channel uncertainty, and proposed accelerated first-order algorithms to maximize the secure EE by optimizing the transmit beamformer, RIS phase-shift, and transmission rate. The proposed accelerated first-order algorithm has a linear complexity order with respect to transmit antennas at the BS and reflecting elements in the RIS, making it ideal for massive antennas and massive RIS applications. Numerical results demonstrated that the proposed accelerated first-order algorithm achieves identical performance to the conventional method but saves at least two orders of magnitude in computation time, and it significantly improves the secure EE compared to baseline schemes of random phase-shift, fixed phase-shift, and RIS ignoring CSI uncertainty.

\appendices

\section{Proof of Theorem~\ref{the:CF_SOP_prob}}
\label{appe:Derive_SOP}
Denote $\mathbf{\Omega}=\sqrt{\alpha_{1}\alpha_{r,j}}\mathbf{{g}}_{r,j}^{H}\mathbf{\Theta}\mathbf{H}+\sqrt{\alpha_{d,j}}\mathbf{{g}}_{d,j}^{H}\in \mathbb{C}^{1\times N}$, the SOP of~\eqref{eq:SOP_Cons} can be rewritten as
\begin{align}
p^{k,j}_\mathrm{so}&=\mathrm{Pr}\left\{ \log_2\left(1+\frac{|\mathbf{\Omega}\mathbf{w}|^2}{\sigma^2_j}\right)>D_{k,j}\right\}\nonumber \\
&=\mathrm{Pr}\left\{\mathbf{w}^H\mathbf{\Omega}^H\mathbf{\Omega}\mathbf{w}>{\sigma^2_j(2^{D_{k,j}}-1)}\right\}.\label{eq: SOP_deriva} 
\end{align}
Since $\mathbf{g}_{r,j}\sim \mathcal{CN}(\mathbf{0},\mu_{r,j}^2\mathbf{I}_M)$, we have $\mathbf{\Theta}\mathbf{H}\mathbf{{g}}_{r,j}\sim \mathcal{CN}(\mathbf{0},\mu_{r,j}^2\mathbf{\Theta}\mathbf{H}(\mathbf{\Theta}\mathbf{H})^H)$. Therefore, the random vector $\mathcal{X}:=\sqrt{\alpha_{1}\alpha_{r,j}}\mathbf{{g}}^H_{r,j}\mathbf{\Theta}\mathbf{H}$ satisfies $\mathcal{X}\sim \mathcal{CN}(\mathbf{0},\alpha_{1}\alpha_{r,j}\mu_{r,j}^2(\mathbf{\Theta}\mathbf{H})^H\mathbf{\Theta}\mathbf{H})$. On the other hand, since $\mathbf{g}_{d,j}\sim \mathcal{CN}(\mathbf{0},\mu_{d,j}^2\mathbf{I}_N)$, the random vector $\mathcal{Y}:=\sqrt{\alpha_{d,j}}\mathbf{{g}}^H_{d,j}$ satisfies
$\mathcal{Y}\sim \mathcal{CN}(\mathbf{0},\alpha_{d,j}\mu_{d,j}^2\mathbf{I}_N)$.
Furthermore, since $\mathcal{X}$ and $\mathcal{Y}$ are independent random vectors, we have 
\begin{equation}
\mathbf{\Omega}=\mathcal{X}+\mathcal{Y}\sim \mathcal{CN}(\mathbf{0},\alpha_{1}\alpha_{r,j}\mu_{r,j}^2(\mathbf{\Theta}\mathbf{H})^H\mathbf{\Theta}\mathbf{H}+\alpha_{d,j}\mu_{d,j}^2\mathbf{I}_N),
\end{equation}
and the expectation of $\mathbf{w}^H\mathbf{\Omega}^H\mathbf{\Omega}\mathbf{w}$ is given by~\cite{J_Kanduk02SOP}
\begin{equation}\label{eq:Expectation_auteo}
\mathbb{E}\left\{\mathbf{w}^H\mathbf{\Omega}^H\mathbf{\Omega}\mathbf{w}\right\}= \mathbf{w}^H(\alpha_{1}\alpha_{r,j}\mu_{r,j}^2(\mathbf{\Theta}\mathbf{H})^H\mathbf{\Theta}\mathbf{H}+\alpha_{d,j}\mu_{d,j}^2\mathbf{I}_N)\mathbf{w}.
\end{equation} 
Moreover, it is known that the received signal power follows the exponential distribution~\cite{B_S_Prin02}. Together with~\eqref{eq:Expectation_auteo}, we have $\mathbf{w}^H\mathbf{\Omega}^H\mathbf{\Omega}\mathbf{w}\sim\mathrm{Exp}(\alpha_{1}\alpha_{r,j}\mu_{r,j}^2\mathbf{w}^H(\mathbf{\Theta}\mathbf{H})^H\mathbf{\Theta}\mathbf{H}\mathbf{w}+\alpha_{d,j}\mu_{d,j}^2\mathbf{w}^H\mathbf{w})$, and a closed-form expression of~\eqref{eq: SOP_deriva} can be obtained as shown in~\eqref{eq:them1_SOP}.

\section{Proof of Theorem~\ref{tem:FP_CCP_def}}\label{Rank:SDR_1}
Since the function value of $\hat{F}(\mathbf{W};\mathbf{W}^{(n)})$ must be non-negative at optimality of problem~\eqref{opt:D_W_Dinkel_nonCon}, the pointwise maximum operation $[\cdot]^+$ in~\eqref{opt:final_IPM_W} can be dropped, and problem~\eqref{opt:final_IPM_W} is rewritten as
\begin{equation}\label{opt:obj_cons_W_CCP}
\max_{\mathbf{W}\succeq \mathbf{0}}~
2y^{(l)}\sqrt{\hat{F}(\mathbf{W};\mathbf{W}^{(n)})}-(y^{(l)})^2\left(\frac{1}{\eta}\mathrm{Tr}\left(\mathbf{W}\right)\!+\! P_a\!+\!KP_c \!+\! MP_s\right), ~\mathrm{s.t.}~
\mathrm{Tr}\left(\mathbf{W}\right)\leq P_\mathrm{max}.
\end{equation}
Then the Lagrangian function of~\eqref{opt:obj_cons_W_CCP} is given by
\begin{equation}
\begin{split}
\mathcal{L}\left(\mathbf{W},\mathbf{\Xi},\mu\right)=&-2y^{(l)}\sqrt{\hat{F}(\mathbf{W};\mathbf{W}^{(n)})}+(y^{(l)})^2\left(\frac{1}{\eta}\mathrm{Tr}\left(\mathbf{W}\right)+ P_a+KP_c + MP_s\right)\\
&+\mu\left(\mathrm{Tr}\left(\mathbf{W}\right)- P_\mathrm{max}\right)
- \mathrm{Tr}\left(\mathbf{\Xi}\mathbf{W}\right),
\end{split}
\end{equation}
where $\mathbf{\Xi} \in \mathbb{C}_+^{N\times N}$ and $\mu\geq 0$ are the dual variables corresponding to the constraints in~\eqref{opt:obj_cons_W_CCP}. Therefore, the optimal solution of~\eqref{opt:obj_cons_W_CCP} must satisfy the following KKT conditions:
$\frac{\partial \mathcal{L}\left(\mathbf{W},\mathbf{\Xi},\mu\right)}{\partial \mathbf{W}}=-\frac{y^{(l)}(\hat{F}(\mathbf{W};\mathbf{W}^{(n)}))^{-\frac{1}{2}}\bf{B}^H\bf{B}}{\left(\sigma^2_k+\bf{B}\mathbf{W}\bf{B}^H\right)\ln2}+\mathbf{\Lambda}-\mathbf{\Xi}=\mathbf{0}$,
	$\mathbf{\Xi}\mathbf{W}=\mathbf{0}$, and $\mathbf{\Xi}\succeq \mathbf{0}$,
	where $\bf{B}=(\sqrt{\alpha_{1}\alpha_{r,k}}\mathbf{h}_{r,k}^H\mathbf{\Theta}\mathbf{H}+\sqrt{\alpha_{d,k}}\mathbf{h}^H_{d,k})\in\mathbb{C}^{1\times N}$ and $\mathbf{\Lambda}$ is given by
	\begin{equation}\label{eq:matrix_lamda}
	\mathbf{\Lambda}= \left(\frac{(y^{(l)})^2}{\eta}+\mu\right)\mathbf{I}_N+
\frac{y^{(l)}(\hat{F}(\mathbf{W};\mathbf{W}^{(n)}))^{-\frac{1}{2}}(\kappa_{1,j}(\mathbf{\Theta}\mathbf{H})^H\mathbf{\Theta}\mathbf{H}+\kappa_{2,j}\mathbf{I}_N)\ln\varepsilon_k^{-1}
}{\left(1+\left(\kappa_{1,j}\mathrm{Tr}(\mathbf{\Theta}\mathbf{H}\mathbf{W}^{(n)}(\mathbf{\Theta}\mathbf{H})^H)+\kappa_{2,j}\mathrm{Tr}(\mathbf{W}^{(n)})\right)\ln\varepsilon_k^{-1}\right)\ln 2}.
\end{equation}
Hence, the optimal primal variable $\mathbf{W}^*$ and dual variable $\mathbf{\Xi}^*$ should satisfy
\begin{equation}
\begin{split}
\mathbf{\Xi}^*=-\frac{y^{(l)}(\hat{F}(\mathbf{W};\mathbf{W}^{(n)}))^{-\frac{1}{2}}\bf{B}^H\bf{B}}{\left(\sigma^2_k+\bf{B}\mathbf{W}\bf{B}^H\right)\ln2}+
\mathbf{\Lambda}.
\end{split}
\end{equation}
By putting $\mathbf{\Xi}^*$ into condition $\mathbf{\Xi}\mathbf{W}=\mathbf{0}$, the optimal $\mathbf{W}^*$ must satisfy
\begin{equation}\label{eq:tem_Inverse_W}
\begin{split}
&\mathbf{\Lambda}\mathbf{W}^*=\frac{y^{(l)}(\hat{F}(\mathbf{W};\mathbf{W}^{(n)}))^{-\frac{1}{2}}\bf{B}^H\bf{B}}{\left(\sigma^2_k+\bf{B}\mathbf{W}\bf{B}^H\right)\ln2}\mathbf{W}^*.
\end{split}
\end{equation}
Notice that matrix $\mathbf{\Lambda}$ of~\eqref{eq:matrix_lamda} is invertible. As a result,~\eqref{eq:tem_Inverse_W} can be rewritten as
\begin{equation}\label{eq:W_rank-1}
\begin{split}
\mathbf{W}^*=
\mathbf{\Lambda}^{-1}
\frac{y^{(l)}(\hat{F}(\mathbf{W};\mathbf{W}^{(n)}))^{-\frac{1}{2}}\bf{B}^H\bf{B}}{\left(\sigma^2_k+\bf{B}\mathbf{W}\bf{B}^H\right)\ln2}\mathbf{W}^*.
\end{split}
\end{equation}
Then taking the rank on both sides of~\eqref{eq:W_rank-1}, we have the following rank relation
\begin{align}
\mathrm{rank}(\mathbf{W}^*)
&=\mathrm{rank}\left(
\mathbf{\Lambda}^{-1}
\frac{y^{(l)}(\hat{F}(\mathbf{W};\mathbf{W}^{(n)}))^{-\frac{1}{2}}\bf{B}^H\bf{B}}{\left(\sigma^2_k+\bf{B}\mathbf{W}\bf{B}^H\right)\ln2}\mathbf{W}^*\right) \nonumber\\
&\leq \mathrm{rank}\left(\frac{y^{(l)}(\hat{F}(\mathbf{W};\mathbf{W}^{(n)}))^{-\frac{1}{2}}\bf{B}^H\bf{B}}{\left(\sigma^2_k+\bf{B}\mathbf{W}\bf{B}^H\right)\ln2}\right) \label{ineq: rank_ineq_ref}\\
&= \mathrm{rank}\left(\bf{B}^H\bf{B}\right), \label{ineq: rank_ineq_W}
\end{align}
where~\eqref{ineq: rank_ineq_ref} follows from~\cite[Lemma 4]{J_Tian16Equality}. 

On the other hand, notice that $\mathrm{rank}\left(\mathbf{q}\right)=\mathrm{rank}\left(\mathbf{q}^H\mathbf{q}\right)\leq 1$ always holds for $\mathbf{q}\in \mathbb{C}^{1\times N}$~\cite{B_LAaM14Banj}. By substituting $\bf{B}=\mathbf{q}$ into~\eqref{ineq: rank_ineq_W}, we can obtain $\mathrm{rank}(\mathbf{W}^*)\leq 1 $. By eliminating the trivial solution of $\mathbf{W}^*=\mathbf{0}$, we can conclude that $\mathrm{rank}(\mathbf{W}^*)=1$ holds.

\section{Proof of Lemma~\ref{lem:Approx_Obj_D1}}\label{lem:App_LB_obj}
From~\eqref{eq:f_objective_nu}, $f\left(\mathbf{w}\right)$ can be rewritten as
\begin{align}\label{eq: first_order_alg_fw}
f\left(\mathbf{w}\right)=&\underbrace{\log_2\left(1+{|\hat{\mathbf{h}}\mathbf{w}|^2}/{\sigma^2_k}\right)}_{:=f_1(\mathbf{w})}
-\underbrace{\log_2\left(1+\left(\kappa_{1,j}\|\mathbf{\Theta}\mathbf{H}\mathbf{w}\|^2+\kappa_{2,j}\|\mathbf{w}\|^2\right)\ln\varepsilon_k^{-1}\right) }_{:=f_2(\mathbf{w})}.
\end{align}
Based on~\eqref{eq: first_order_alg_fw}, we first derive a lower bound of $f_1(\mathbf{w})$ by applying the linear interpolation with the help of an auxiliary function $g_1(x,y):=-\log_2(1-|x|^2/y)$, which is a convex function based on~\cite{J_Ali17ConvexAux}. 
Then by using the first-order Taylor series expansion of $g_1(x,y)$ at any fixed point $(\hat{x},\hat{y})\in\mathcal{S}$, we have
\begin{equation}\label{ineq:taylor_expa}
\begin{split}
g_1(x,y)&\geq g_1(\hat{x},\hat{y})+(x-\hat{x})\nabla_{\hat{x}} g_1(x,\hat{y})+
(y-\hat{y})\nabla_{\hat{y}} g_1(\hat{x},y)
\\
&=g_1(\hat{x},\hat{y})+\frac{2\Re\{\hat{x}(x-\hat{x})\}}{(\hat{y}-|\hat{x}|^2)\ln 2}-\frac{(y-\hat{y})|\hat{x}|^2}{\hat{y}(\hat{y}-|\hat{x}|^2)\ln2}.
\end{split}
\end{equation}
By substituting $b:=y-|{x}|^2$ and $\hat{b}:=\hat{y}-|\hat{x}|^2$ into~\eqref{ineq:taylor_expa}, we obtain
\begin{equation}\label{ineq: Taylor_f1_LB}
\log_2\left(1+\frac{|{x}|^2}{b}\right)\geq \log_2\left( 1+\frac{|\hat{x}|^2}{\hat{b}}\right)+\frac{2\Re\{\hat{x}x\}}{\hat{b}\ln2}
-\frac{(b+|{x}|^2)|\hat{x}|^2}{(\hat{b}+|\hat{x}|^2)\hat{b}\ln 2}-\frac{|\hat{x}|^2}{\hat{b}\ln2}.
\end{equation}
Finally, by putting $x=\hat{\mathbf{h}}\mathbf{w}$, $\hat{x}=\hat{\mathbf{h}}\mathbf{w}^{(t)}$, $b=\sigma^2_k$ and $\hat{b}=\sigma^2_k$ into~\eqref{ineq: Taylor_f1_LB}, a lower bound of $f_1(\mathbf{w})$ is obtained as $f_1(\mathbf{w})\geq R^{(t)}_1(\mathbf{w})$ with the equality holds at $\mathbf{w}=\mathbf{w}^{(t)}$, where $R^{(t)}_1(\mathbf{w})$ is shown in~\eqref{eq:R1_auxi}.

On the other hand, we provide an upper bound of $f_2(\mathbf{w})$ with the help of an auxiliary function $g_2(z):=\log_2(1+z)$.
Since $g_2(z)$ is a concave function on $z\geq 0$, by using the first-order Taylor series expansion of $g_2(z)$ at any fixed point $\hat{z}\geq 0$, we have
\begin{equation}\label{ineq:Taylor_f2_ub}
\log_2(1+z)\leq \log_2(1+\hat{z})+\frac{(z-\hat{z})}{(1+\hat{z})\ln2}.
\end{equation}
By putting $z=\left(\kappa_{1,j}\|\mathbf{\Theta}\mathbf{H}\mathbf{w}\|^2+\kappa_{2,j}\|\mathbf{w}\|^2\right)\ln\varepsilon_k^{-1}$ and $\hat{z}=\left(\kappa_{1,j}\|\mathbf{\Theta}\mathbf{H}\mathbf{w}^{(t)}\|^2+\kappa_{2,j}\|\mathbf{w}^{(t)}\|^2\right)\ln\varepsilon_k^{-1}$ into~\eqref{ineq:Taylor_f2_ub}, we obtain an upper bound of 
$f_2(\mathbf{w})$ as $f_2(\mathbf{w})\leq R^{(t)}_2(\mathbf{w})$
with the equality holds at $\mathbf{w}=\mathbf{w}^{(t)}$, where 
$R^{(t)}_2(\mathbf{w})$ is shown in~\eqref{eq:R2_auxli}.

Based on above discussion, a lower bound for $f\left(\mathbf{w}\right)$ can be obtained as $f\left(\mathbf{w}\right) \geq R_1^{(t)}(\mathbf{w})-R_2^{(t)}(\mathbf{w})$ with the equality holds at $\mathbf{w}=\mathbf{w}^{(t)}$.

\section{Proof of Lemma~\ref{them: optimal_w_Approxi}}\label{proof:theorem_primal_dual}
Based on feasible set $\mathcal{P}_{\mathbf{x}}$, the Lagrangian function of problem~\eqref{opt:project_w} is given by
\begin{align}\label{eq:lagar_dualFunc}
\mathcal{L}\left(\mathbf{x},\zeta\right)=   \|\mathbf{x}-\mathbf{x}^{(i+\frac{1}{2})}\|^2
+\zeta \left(\|\mathbf{x}\|^2-P_\mathrm{max}\right),
\end{align}
where $\zeta$ is the corresponding dual variable. 
The optimal solutions $\mathbf{x}^{\diamond}$ and $\zeta^{\diamond}$ must satisfy the following KKT conditions:
$(1+\zeta^{\diamond})\mathbf{x}^{\diamond} = \mathbf{x}^{(i+\frac{1}{2})}$, 
$\zeta^{\diamond} \left(\|\mathbf{x}^{\diamond} \|^2-P_\mathrm{max}\right)=0$,  
$\|\mathbf{x}^{\diamond}\|^2-P_\mathrm{max}\leq 0$, and $\zeta^{\diamond}\geq 0$.
Therefore, $\mathbf{x}^{\diamond}$ is derived as
$\mathbf{x}^{\diamond} =\mathbf{x}^{(i+\frac{1}{2})}/(1+\zeta^{\diamond})$. Together with dual feasibility $\zeta\geq 0$, $\zeta^{\diamond}$ is given by
$\zeta^{\diamond}=({\|\mathbf{x}^{(i+\frac{1}{2})}\|}/{\sqrt{P_\mathrm{max}}}-1)\mathbb{I}(\zeta>0)$, 
where $\mathbb{I}(H)$ is the indicator function with $\mathbb{I}(H)=1$ if the event \textit{H} occurs and $\mathbb{I}(H)=0$ otherwise.
By putting $\zeta^{\diamond}$ into $\mathbf{x}^{\diamond}$, the optimal $\mathbf{x}^{\diamond}$ is given by
\begin{equation}
\begin{split}
\mathbf{x}^{\diamond}=&\left\{\begin{array}{ll}
\mathbf{x}^{(i+\frac{1}{2})}
, & \!\!\mathrm{if} ~\|\mathbf{x}^{(i+\frac{1}{2})}\|^2\leq P_\mathrm{max}, \\
\frac{\sqrt{P_\mathrm{max}}}{\left\|\mathbf{x}^{(i+\frac{1}{2})}\right\|}\mathbf{x}^{(i+\frac{1}{2})}, &\!\!\mathrm{if} ~\|\mathbf{x}^{(i+\frac{1}{2})}\|^2> P_\mathrm{max}.
\end{array}\right.
\end{split}
\end{equation}
Therefore, the optimal solution to~\eqref{opt:project_w} is obtained as shown in Lemma~\ref{them: optimal_w_Approxi}.

\section{Proof of Theorem~\ref{theorm: convergence_property}}
\label{proof:theorem_KKT}
Firstly, it is noticed that problem~\eqref{eq:obj_approx_dual_problem} is solved via accelerated PG method. Since the objective function is smooth concave and the constraint set is closed and convex, the iteration with respect to $i$ according to~\eqref{eq:acce_PG_moment}-\eqref{eq:tunned_para_PG} is guaranteed to converge to the global optimal point of~\eqref{eq:obj_approx_dual_problem}~\cite[Theorem 4.4]{J_BeckAmir09AFIS}. 
Since~\eqref{eq:obj_approx_dual_problem} is strongly concave over $\mathbf{w}$, the obtained global optimal point must satisfy KKT conditions of~\eqref{eq:obj_approx_dual_problem} and the optimal $\mathbf{w}^{*}$ is a KKT solution to problem~\eqref{eq:obj_approx_dual_problem}.
Moreover, it is known that problem~\eqref{eq:obj_approx_dual_problem} is reformulated from~\eqref{opt:CCR_D1_inner_Appx} based on the quadratic  transformation and the objective function of~\eqref{opt:CCR_D1_inner_Appx} is in concave-convex form.
As a result, as iteration number $l$ increases, the solution is guaranteed to converge to a local optimal solution to~\eqref{opt:CCR_D1_inner_Appx}~\cite{J_Yuwei18FP}. Finally, notice that~\eqref{opt:CCR_D1_inner_Appx} is the $t^{th}$ subproblem with PFP, to prove the convergence of the iteration with respect to $t$, we establish the following property:
\begin{lemma}\label{lem: gradient_inner_approx}
Denote the gradient of $f\left(\mathbf{w}\right)$, $R_1^{(t)}(\mathbf{w})$ and $R_2^{(t)}(\mathbf{w})$ as $\nabla f\left(\mathbf{w}\right) $, $\nabla R_1^{(t)}(\mathbf{w})$ and $\nabla R_2^{(t)}(\mathbf{w})$, respectively, we have 
$\nabla f\left(\mathbf{w}^{(t)}\right) = \nabla R_1^{(t)}(\mathbf{w}^{(t)})-\nabla R_2^{(t)}(\mathbf{w}^{(t)})$.
\end{lemma}
\begin{proof}
From~\eqref{eq: first_order_alg_fw}, $\nabla f\left(\mathbf{w}\right) $ is given by
\begin{equation}\label{eq:Grad_f_proof}
\nabla f\left(\mathbf{w}\right) =\frac{2}{\ln 2}\left(\frac{(\hat{\mathbf{h}}\mathbf{w})^H\hat{\mathbf{h}}^H}{\sigma^2_k+|\hat{\mathbf{h}}\mathbf{w}|^2}-\frac{(\kappa_{1,j}(\mathbf{\Theta}\mathbf{H})^H\mathbf{\Theta}\mathbf{H}\mathbf{w}+\kappa_{2,j}\mathbf{w})\ln\varepsilon_k^{-1}}{1+\left(\kappa_{1,j}\|\mathbf{\Theta}\mathbf{H}\mathbf{w}\|^2+\kappa_{2,j}\|\mathbf{w}\|^2\right)\ln\varepsilon_k^{-1}}\right).
\end{equation}
Based on Lemma~\ref{lem:Approx_Obj_D1}, $\nabla R_1^{(t)}(\mathbf{w})$ and $\nabla R_2^{(t)}(\mathbf{w})$ are respectively given by
\begin{equation}
\nabla R_1^{(t)}(\mathbf{w})=\frac{2(\hat{\mathbf{h}}\mathbf{w}^{(t)})^H\hat{\mathbf{h}}^H}{\sigma^2_k\ln 2}-\frac{2(\hat{\mathbf{h}}\mathbf{w})^H\hat{\mathbf{h}}^H|\hat{\mathbf{h}}\mathbf{w}^{(t)}|^2}{(\sigma^2_k+|\hat{\mathbf{h}}\mathbf{w}^{(t)}|^2)\sigma^2_k\ln 2},
\end{equation}
\begin{equation}\label{eq:gradient_R2}
\nabla R_2^{(t)}(\mathbf{w})=\frac{2(\kappa_{1,j}(\mathbf{\Theta}\mathbf{H})^H\mathbf{\Theta}\mathbf{H}\mathbf{w}+\kappa_{2,j}\mathbf{w})\ln\varepsilon_k^{-1}}{\left(1+\left(\kappa_{1,j}\|\mathbf{\Theta}\mathbf{H}\mathbf{w}^{(t)}\|^2+\kappa_{2,j}\|\mathbf{w}^{(t)}\|^2\right)\ln\varepsilon_k^{-1}\right)\ln2}.
\end{equation}
By substituting $\mathbf{w}=\mathbf{w}^{(t)}$ into~\eqref{eq:Grad_f_proof}-\eqref{eq:gradient_R2}, we can verify that the equality $\nabla f\left(\mathbf{w}^{(t)}\right)=\nabla R_1^{(t)}(\mathbf{w}^{(t)})-\nabla R_2^{(t)}(\mathbf{w}^{(t)})$ holds, which completes the proof.
\end{proof}
\noindent Based on Lemma~\ref{lem:Approx_Obj_D1} and Lemma~\ref{lem: gradient_inner_approx}, it is known that the constructed lower bound of $f(\mathbf{w})$ satisfies the gradient consistency. As a result, the convergence with respect to $t$ is a KKT point of $\mathcal{D}1$~\cite[Theorem 1]{J_Marks78AGIA}.

To sum up, there is no performance loss from~\eqref{eq:obj_approx_dual_problem} to~\eqref{opt:CCR_D1_inner_Appx}, and the sequence of solutions obtained based on~\eqref{opt:CCR_D1_inner_Appx} with the PFP iterations converges to a KKT solution to $\mathcal{D}1$. Therefore, the sequence of solutions generated from Algorithm~\ref{alg:PG_Path} converges to a KKT point, which is at least a local optimal point of $\mathcal{D}1$.

\section{Derivation of a lower bound of $\Upsilon(\mathbf{v})$
}\label{lem:CF_lowerbound_AM}
From~\eqref{obj:Q2_theta_v}, $\Upsilon(\mathbf{v})$ can be rewritten as 
\begin{equation}
\Upsilon(\mathbf{v})=\underbrace{\log_2\left(1+{|\mathbf{v}^H\mathbf{A}_k\mathbf{w}|^2}/{\sigma^2_k}\right)}_{:=\Upsilon_1(\mathbf{v})}-\underbrace{\log_2\left(1+(\kappa_{1,j}|\mathbf{w}^H\hat{\mathbf{H}}^H\mathbf{v}|^2+\kappa_{2,j}\|\mathbf{w}\|^2)\ln\varepsilon_k^{-1}\right)}_{:=\Upsilon_2(\mathbf{v})}.
\end{equation}
Using the procedures similar to Appendix~\ref{lem:App_LB_obj}, given any fixed $\mathbf{v}^{(t)}$, a lower bound of $\Upsilon_1(\mathbf{v})$ and an upper bound of $\Upsilon_2(\mathbf{v})$ can be derived as
\begin{align}\label{eq:lower_ups1}
\Upsilon_1(\mathbf{v})\geq& \log_2\left(1+{|(\mathbf{v}^{(t)})^H\mathbf{A}_k\mathbf{w}|^2/\sigma^2_k}\right)+\frac{2\Re\{\left((\mathbf{v}^{(t)})^H\mathbf{A}_k\mathbf{w}\right)^H(\mathbf{v}^H\mathbf{A}_k\mathbf{w})\}}{\sigma^2_k\ln 2}\nonumber \\
&-\frac{\left(\sigma^2_k+|\mathbf{v}^H\mathbf{A}_k\mathbf{w}|^2\right)|(\mathbf{v}^{(t)})^H\mathbf{A}_k\mathbf{w}|^2}{(\sigma^2_k+|(\mathbf{v}^{(t)})^H\mathbf{A}_k\mathbf{w}|^2)\sigma^2_k\ln2}-\frac{|(\mathbf{v}^{(t)})^H\mathbf{A}_k\mathbf{w}|^2}{\sigma^2_k\ln2},
\end{align}
\begin{align}\label{eq:upper_ups2}
\Upsilon_2(\mathbf{v})\leq& \log_2\left(1+(\kappa_{1,j}|\mathbf{w}^H\hat{\mathbf{H}}^H\mathbf{v}^{(t)}|^2+\kappa_{2,j}\|\mathbf{w}\|^2)\ln\varepsilon_k^{-1}\right)\nonumber\\
&+\frac{\kappa_{1,j}(|\mathbf{w}^H\hat{\mathbf{H}}^H\mathbf{v}|^2-|\mathbf{w}^H\hat{\mathbf{H}}^H\mathbf{v}^{(t)}|^2)\ln\varepsilon_k^{-1}}{(1+(\kappa_{1,j}|\mathbf{w}^H\hat{\mathbf{H}}^H\mathbf{v}^{(t)}|^2+\kappa_{2,j}\|\mathbf{w}\|^2)\ln\varepsilon_k^{-1})\ln2},
\end{align}
with the equality holds at $\mathbf{v}=\mathbf{v}^{(t)}$.
Based on~\eqref{eq:lower_ups1} and~\eqref{eq:upper_ups2}, a lower bound of $\Upsilon(\mathbf{v})$, denoted by $\hat{\Upsilon}^{(t)}(\mathbf{v})$, can be obtained as 
\begin{align}\label{eq:low_SCA_ups}
\hat{\Upsilon}^{(t)}(\mathbf{v})=&\log_2\left(1+{|(\mathbf{v}^{(t)})^H\mathbf{A}_k\mathbf{w}|^2/\sigma^2_k}\right)+\frac{2\Re\{\left((\mathbf{v}^{(t)})^H\mathbf{A}_k\mathbf{w}\right)^H(\mathbf{v}^H\mathbf{A}_k\mathbf{w})\}}{\sigma^2_k\ln 2}-\frac{|(\mathbf{v}^{(t)})^H\mathbf{A}_k\mathbf{w}|^2}{\sigma^2_k\ln2}\nonumber\\ &-\!\frac{\left(\sigma^2_k+|\mathbf{v}^H\mathbf{A}_k\mathbf{w}|^2\right)|(\mathbf{v}^{(t)})^H\mathbf{A}_k\mathbf{w}|^2}{(\sigma^2_k+|(\mathbf{v}^{(t)})^H\mathbf{A}_k\mathbf{w}|^2)\sigma^2_k\ln2}\!-\!\log_2\left(1+(\kappa_{1,j}|\mathbf{w}^H\hat{\mathbf{H}}^H\mathbf{v}^{(t)}|^2+\kappa_{2,j}\|\mathbf{w}\|^2)\ln\varepsilon_k^{-1}\right)\nonumber\\
&-\!\frac{\kappa_{1,j}(|\mathbf{w}^H\hat{\mathbf{H}}^H\mathbf{v}|^2-|\mathbf{w}^H\hat{\mathbf{H}}^H\mathbf{v}^{(t)}|^2)\ln\varepsilon_k^{-1}}{(1+(\kappa_{1,j}|\mathbf{w}^H\hat{\mathbf{H}}^H\mathbf{v}^{(t)}|^2+\kappa_{2,j}\|\mathbf{w}\|^2)\ln\varepsilon_k^{-1})\ln2}.
\end{align}

\section{Proof of Theorem~\ref{theorm: convergence_SCAmanifold}}\label{Proof:Theorm_4}
Notice that the inner iteration is to solve a geodesically concave optimization problem, and the iteration with respect to $l$ is guaranteed to converge to a local optimal solution to~\eqref{obj:opt_geo_concave}~\cite{C_Accele_remani18Alg}. For outer iteration with the PFP, since the constructed concave functions $\{\hat{\Upsilon}^{(t)}(\mathbf{v})\}_{t\in\mathbb{N}}$ satisfy $ \nabla \hat{\Upsilon}^{(t)}(\mathbf{v}^{(t)})=\nabla \Upsilon\left(\mathbf{v}^{(t)}\right) $ and $\Upsilon(\mathbf{v})\geq \hat{\Upsilon}^{(t)}(\mathbf{v})$ with equality holds at $\mathbf{v}=\mathbf{v}^{(t)}$, $\hat{\Upsilon}^{(t)}(\mathbf{v})$ satisfies the gradient consistency. Therefore, the iteration over $t$ converges to a stationary point of $\mathcal{Q}2$~\cite{J_Marks78AGIA}. 	
Notice that there is no performance loss for solving~\eqref{obj:opt_geo_concave}, and the sequence of solutions to~\eqref{obj:opt_geo_concave} with the PFP converges to a stationary solution to $\mathcal{Q}2$. Therefore, the sequence of solutions generated from Algorithm~\ref{alg:CG_obligueManifold} converges to a stationary point.

\section{Proof of Theorem~\ref{tem:conver_AM}} \label{proof_tem_convergence}
Define $\mathbf{\Theta}_r$ and $\mathbf{w}_r$ as the solution generated from Algorithm~\ref{alg:PG_Path} and~\ref{alg:CG_obligueManifold}, respectively, under AM framework at the $r^{th}$ iteration with the corresponding objective function of $\mathcal{P}2^{[k,j]}$ denoted by $Y(\mathbf{\Theta}_r,\mathbf{w}_r)$.
To prove Theorem~\ref{tem:conver_AM}, we provide the following property. 
\begin{lemma}\label{lem:limitPoint_convergence}
The sequence of solutions $\{\mathbf{\Theta}_{r},\mathbf{w}_{r}\}_{r\in\mathbb{N}}$ is bounded and must have a limit point $(\mathbf{\Theta}^*,\mathbf{w}^*)$. 
\end{lemma}
\begin{proof}
First, we prove the objective function value $Y(\mathbf{\Theta}_r,\mathbf{w}_r)$ is monotonically increasing as $r$ increases.
According to Theorem~\ref{theorm: convergence_property}, it is known that the obtained solution in Algorithm~\ref{alg:PG_Path} is a local optimal point for $\mathcal{D}1$, and it must be a saddle-point for $\mathcal{P}2^{[k,j]}$. Together with the property of saddle-points, we have the following inequality~\cite{J_Liu09Saddle}
\begin{equation}\label{ineq:AM_conv_B2}
Y(\mathbf{\Theta}_{r},\mathbf{w}_{r+1})\geq 
Y(\mathbf{\Theta}_r,\mathbf{w}_r), r\in\mathbb{N}.
\end{equation}
On the other hand, the optimization over $\mathbf{\Theta}$ is independent on $\mathbf{w}$. Furthermore, according to Theorem~\ref{theorm: convergence_SCAmanifold}, the obtained 
point $\mathbf{\Theta}_{r+1}$ is a stationary point. Hence, we have the following inequality
\begin{equation}\label{ineq:Converge_AM_condi}
Y(\mathbf{\Theta}_{r+1},\mathbf{w}_{r+1})\geq Y(\mathbf{\Theta}_{r},\mathbf{w}_{r+1}), \forall r\in\mathbb{N}.
\end{equation}
Combining~\eqref{ineq:AM_conv_B2} and~\eqref{ineq:Converge_AM_condi}, we conclude that 
\begin{equation}\label{eq:mon_incre_pro1}
Y(\mathbf{\Theta}_{r+1},\mathbf{w}_{r+1})\geq Y(\mathbf{\Theta}_{r},\mathbf{w}_{r})\geq \cdots\geq  Y(\mathbf{\Theta}_{0},\mathbf{w}_{0}), r\in\mathbb{N},
\end{equation}
where $Y(\mathbf{\Theta}_{0},\mathbf{w}_{0})$ is any finite initial value of the objective function. 

Now we prove the boundedness of $\{\mathbf{\Theta}_{r},\mathbf{w}_{r}\}_{r\in\mathbb{N}}$. It is known that 
the constraint~\eqref{cons:beamform_power} can be rewritten as a norm constraint $\|\mathbf{w}\|^2\leq P_\mathrm{max}$. As a result, $\mathbf{w}$ is located in a closed set and the sequence of solutions $\{\mathbf{w}_{r}\}_{r\in \mathbb{N}}$ is bounded by constraint~\eqref{cons:beamform_power}. 
On the other hand, since the obtained solution of $\mathcal{Q}2$ is guaranteed to converge to a maximizer, every solution of $\mathcal{Q}2$ converges to a compact and connected set~\cite{J_Trendafil99OM}. As a result, the set of stationary points is a compact set. Therefore, the sequence of solutions $\{\mathbf{v}_{r}\}_{r\in \mathbb{N}}$ is bounded.
After linear transformation of a vector $\mathbf{v}$ to a matrix $\mathbf{\Theta}$, $\{\mathbf{\Theta}_{r}\}_{r\in \mathbb{N}}$ is also bounded.
Together with non-decreasing property of~\eqref{eq:mon_incre_pro1}, there must exist a limit point $(\mathbf{\Theta}^*,\mathbf{w}^*)$ based on Bolzano-Weierstrass theorem~\cite{B_BartleRobert11}. 
\end{proof}

To further investigate the property of limit point $(\mathbf{\Theta}^*,\mathbf{w}^*)$, we transform the constrained optimization problem $\mathcal{P}2^{[k,j]}$ into an unconstrained problem based on proximal alternating maximization method~\cite{J_Bolte14ConverAnalysis}. 
Specifically, to deal with constraint~\eqref{cons:modul_thetaP1}, an indicator function $\mathbb{I}_1\left(\mathbf{\Theta}\right)$ is introduced with $\mathbb{I}_1( \mathbf{\Theta})=0$ if $\mathbf{\Theta}$ satisfies~\eqref{cons:modul_thetaP1} and $\mathbb{I}_1(\mathbf{\Theta})=-\infty$ otherwise. Similarly, constraint~\eqref{cons:beamform_power} is characterized by an indicator function $\mathbb{I}_2\left(\mathbf{w}\right)$.
With the help of  $\mathbb{I}_1(\mathbf{\Theta})$ and $\mathbb{I}_2\left(\mathbf{w}\right)$, the constrained problem $\mathcal{P}2^{[k,j]}$ can be equivalently written as an unconstrained form: 
\begin{equation}\label{eq:non_constraint_OPt}
  \max_{\mathbf{\Theta},\mathbf{w}}~Y(\mathbf{\Theta},\mathbf{w})+ \mathbb{I}_1(\mathbf{\Theta})+   \mathbb{I}_2\left(\mathbf{w}\right). 
\end{equation}
Then, we can establish the following property based on~\eqref{eq:non_constraint_OPt}.
\begin{lemma}\label{lem:maximizer_prove_limitpoint}
The limit point $(\mathbf{\Theta}^*,\mathbf{w}^*)$ is the maximizer of problem~\eqref{eq:non_constraint_OPt} and satisfies the first-order optimality condition.
\end{lemma}
\begin{proof}
Notice that the converged point $\mathbf{w}_{r+1}$ for the $r^{th}$ iteration is a local optimal point of $\mathcal{D}1$ according to Theorem~\ref{theorm: convergence_property}. On the other hand, since $\mathcal{Q}2$ is equivalent to $\mathcal{Q}1$ and the obtained solution of $\mathbf{v}$ is a stationary point of $\mathcal{Q}2$ according to Theorem~\ref{theorm: convergence_SCAmanifold}, the converged point $\mathbf{\Theta}_{r+1}$ for the $r^{th}$ iteration is a stationary point of $\mathcal{Q}1$.
Furthermore, since constraint sets of $\mathbf{w}$ and $\mathbf{\Theta}$ are independent and separable, functions $\mathbb{I}_1(\mathbf{\Theta})$ and $\mathbb{I}_2\left(\mathbf{w}\right)$ are independent. Hence, $(\mathbf{\Theta}_{r+1},\mathbf{w}_{r+1})$ is the maximizer of~\eqref{eq:non_constraint_OPt} for the $r^{th}$ iteration. As a result, we have 
\begin{equation}
Y(\mathbf{\Theta},\mathbf{w})+ \mathbb{I}_1(\mathbf{\Theta})+   \mathbb{I}_2\left(\mathbf{w}\right)\leq Y(\mathbf{\Theta}_{r+1},\mathbf{w}_{r+1})+ \mathbb{I}_1(\mathbf{\Theta}_{r+1})+   \mathbb{I}_2\left(\mathbf{w}_{r+1}\right).
\end{equation}
Accordingly, the following inequalities hold:
\begin{equation}\label{ineq:conver_w}
Y(\mathbf{\Theta},\mathbf{w}_{r+1})+ \mathbb{I}_1(\mathbf{\Theta})\leq Y(\mathbf{\Theta}_{r+1},\mathbf{w}_{r+1})+ \mathbb{I}_1(\mathbf{\Theta}_{r+1}),~\forall \mathbf{\Theta},
\end{equation}
\begin{equation}\label{ineq:conver_theta}
Y(\mathbf{\Theta}_{r+1},\mathbf{w})+   \mathbb{I}_2\left(\mathbf{w}\right)\leq 
Y(\mathbf{\Theta}_{r+1},\mathbf{w}_{r+1})+   \mathbb{I}_2\left(\mathbf{w}_{r+1}\right),~\forall \mathbf{w}.
\end{equation}

On the other hand, since constraints~\eqref{cons:modul_thetaP1} and~\eqref{cons:beamform_power} respectively represent separable compact constraint sets, $\mathbb{I}_1(\mathbf{\Theta})$ and $\mathbb{I}_2\left(\mathbf{w}\right)$ are upper semi-continuous functions~\cite{J_Semi_con01Kiwi}. As a result, we have 
\begin{equation}\label{ineq:sub_diff_theta}
\limsup\limits_{r\rightarrow \infty} ~\mathbb{I}_1(\mathbf{\Theta}_{r+1}) \leq \mathbb{I}(\mathbf{\Theta}^*),
\end{equation}
\begin{equation}\label{ineq:sub_diff_w}
\limsup\limits_{r\rightarrow \infty} ~\mathbb{I}_2(\mathbf{w}_{r+1}) \leq \mathbb{I}(\mathbf{w}^*).
\end{equation}
Furthermore, due to the continuity of $Y(\mathbf{\Theta},\mathbf{w})$ over both $\mathbf{\Theta}$ and $\mathbf{w}$, together with the non-decreasing property of~\eqref{eq:mon_incre_pro1}, we have 
\begin{equation}\label{ineq:continu_Objec_AM}
\lim\limits_{r\rightarrow \infty} Y(\mathbf{\Theta}_{r+1},\mathbf{w}_{r+1}) \leq Y(\mathbf{\Theta}^*,\mathbf{w}^*).
\end{equation}
Taking $r\rightarrow \infty$ on both sides of~\eqref{ineq:conver_w}-\eqref{ineq:conver_theta} and applying~\eqref{ineq:sub_diff_theta}-\eqref{ineq:continu_Objec_AM}, we have 
\begin{equation}\label{ineq:convergence_w}
Y(\mathbf{\Theta},\mathbf{w}^*)+ \mathbb{I}_1(\mathbf{\Theta})\leq Y(\mathbf{\Theta}^*,\mathbf{w}^*)+ \mathbb{I}_1(\mathbf{\Theta}^*),~ \forall \mathbf{\Theta},
\end{equation}
\begin{equation}\label{ineq:convergence_theta}
Y(\mathbf{\Theta}^*,\mathbf{w})+   \mathbb{I}_2\left(\mathbf{w}\right)\leq Y(\mathbf{\Theta}^*,\mathbf{w}^*)+ \mathbb{I}_2\left(\mathbf{w}^*\right),~ \forall \mathbf{w},
\end{equation}
which indicate that $\mathbf{\Theta}^*$ is the maximizer of $Y(\mathbf{\Theta},\mathbf{w}^*)+ \mathbb{I}_1(\mathbf{\Theta})$ and $\mathbf{w}^*$ is the maximizer of $Y(\mathbf{\Theta}^*,\mathbf{w})+  \mathbb{I}_2\left(\mathbf{w}\right)$. Hence, we obtain 
 \begin{equation}\label{eq:first_order_theta}
 \mathbf{0}\in \nabla_ {\mathbf{\Theta}}Y(\mathbf{\Theta},\mathbf{w}^*)|_{\mathbf{\Theta}=\mathbf{\Theta}^*}+\partial_\mathbf{\Theta} \mathbb{I}_1(\mathbf{\Theta})|_{\mathbf{\Theta}=\mathbf{\Theta}^*},
 \end{equation}
 \begin{equation}\label{eq:first_order_w}
 \mathbf{0}\in \nabla_ {\mathbf{w}}Y(\mathbf{\Theta}^*,\mathbf{w})|_{\mathbf{w}=\mathbf{w}^*}+\partial_\mathbf{w} \mathbb{I}_2(\mathbf{w})|_{\mathbf{w}=\mathbf{w}^*},
 \end{equation}
where $\partial_\mathbf{\Theta} \mathbb{I}_1(\mathbf{\Theta})$ and $\partial_\mathbf{w} \mathbb{I}_2(\mathbf{w})$ are the limiting subdifferential of the non-smooth functions $\mathbb{I}_1(\mathbf{\Theta})$ and $\mathbb{I}_2(\mathbf{w})$, respectively~\cite{B_MordukhovichBorisS12}. 
Since $\mathbb{I}_1(\mathbf{\Theta})$ does
not depend on $\mathbf{w}$, $\mathbf{0}\in \partial_\mathbf{w} \mathbb{I}_1(\mathbf{\Theta}^*)|_{\mathbf{w}=\mathbf{w}^*}$ holds.
Similarly, we have  $\mathbf{0}\in \partial_\mathbf{\Theta} \mathbb{I}_2(\mathbf{w}^*)|_{\mathbf{\Theta}=\mathbf{\Theta}^*}$.
Together with~\eqref{eq:first_order_theta}-\eqref{eq:first_order_w}, we have 
\begin{equation}
\mathbf{0}\in \nabla_ {\mathbf{\Theta}}Y(\mathbf{\Theta},\mathbf{w}^*)|_{\mathbf{\Theta}=\mathbf{\Theta}^*}+\partial_\mathbf{\Theta} \mathbb{I}_1(\mathbf{\Theta})|_{\mathbf{\Theta}=\mathbf{\Theta}^*}+\partial_\mathbf{\Theta} \mathbb{I}_2(\mathbf{w}^*)|_{\mathbf{\Theta}=\mathbf{\Theta}^*},
\end{equation}
\begin{equation}
\mathbf{0}\in \nabla_ {\mathbf{w}}Y(\mathbf{\Theta}^*,\mathbf{w})|_{\mathbf{w}=\mathbf{w}^*}+\partial_\mathbf{w} \mathbb{I}_1(\mathbf{\Theta}^*)|_{\mathbf{w}=\mathbf{w}^*}+\partial_\mathbf{w} \mathbb{I}_2(\mathbf{w})|_{\mathbf{w}=\mathbf{w}^*}.
\end{equation}
Hence, the limit point $(\mathbf{\Theta}^*,\mathbf{w}^*)$ is the maximizer of~\eqref{eq:non_constraint_OPt} and $(\mathbf{\Theta}^*,\mathbf{w}^*)$ satisfies the first-order optimality condition of~\eqref{eq:non_constraint_OPt} over $\mathbf{\Theta}$ and $\mathbf{w}$. 
\end{proof}
Based on Lemma~\ref{lem:maximizer_prove_limitpoint}, the limit point $(\mathbf{\Theta}^*,\mathbf{w}^*)$ is a critical point of~\eqref{eq:non_constraint_OPt}~\cite{J_Bolte14ConverAnalysis}. Since $\mathcal{P}2^{[k,j]}$ is equivalently written as~\eqref{eq:non_constraint_OPt}, $(\mathbf{\Theta}^*,\mathbf{w}^*)$ is a stationary point of $\mathcal{P}2^{[k,j]}$.

\bibliographystyle{IEEEtran}

\bibliography{paper}

\begin{thebibliography}{10}
\providecommand{\url}[1]{#1}
\csname url@samestyle\endcsname
\providecommand{\newblock}{\relax}
\providecommand{\bibinfo}[2]{#2}
\providecommand{\BIBentrySTDinterwordspacing}{\spaceskip=0pt\relax}
\providecommand{\BIBentryALTinterwordstretchfactor}{4}
\providecommand{\BIBentryALTinterwordspacing}{\spaceskip=\fontdimen2\font plus
\BIBentryALTinterwordstretchfactor\fontdimen3\font minus
  \fontdimen4\font\relax}
\providecommand{\BIBforeignlanguage}[2]{{%
\expandafter\ifx\csname l@#1\endcsname\relax
\typeout{** WARNING: IEEEtran.bst: No hyphenation pattern has been}%
\typeout{** loaded for the language `#1'. Using the pattern for}%
\typeout{** the default language instead.}%
\else
\language=\csname l@#1\endcsname
\fi
#2}}
\providecommand{\BIBdecl}{\relax}
\BIBdecl

\bibitem{J_IRS_EE19Huang}
C.~{Huang}, A.~{Zappone}, G.~C. {Alexandropoulos}, M.~{Debbah}, and C.~{Yuen},
  ``Reconfigurable intelligent surfaces for energy efficiency in wireless
  communication,'' \emph{IEEE Trans. Wireless Commun.}, vol.~18, no.~8, pp.
  4157--4170, Aug. 2019.

\bibitem{J_alendos20reconfigurable}
G.~C. Alexandropoulos, G.~Lerosey, M.~Debbah, and M.~Fink, ``Reconfigurable
  intelligent surfaces and metamaterials: The potential of wave propagation
  control for {6G} wireless communications,'' \emph{arXiv: 2006.11136}, 2020.

\bibitem{J_Alexa21DMA}
N.~{Shlezinger}, G.~C. {Alexandropoulos}, M.~F. {Imani}, Y.~C. {Eldar}, and
  D.~R. {Smith}, ``Dynamic metasurface antennas for {6G} extreme massive {MIMO}
  communications,'' \emph{IEEE Wireless Commun.}, pp. 1--8, 2021.

\bibitem{J_Wu17Review5G}
Q.~{Wu}, G.~Y. {Li}, W.~{Chen}, D.~W.~K. {Ng}, and R.~{Schober}, ``An overview
  of sustainable green {5G} networks,'' \emph{IEEE Wireless Commun.}, vol.~24,
  no.~4, pp. 72--80, Aug. 2017.

\bibitem{J_alexandropoulos20safeguarding}
G.~C. Alexandropoulos, K.~Katsanos, M.~Wen, and D.~B. da~Costa, ``Safeguarding
  {MIMO} communications with reconfigurable metasurfaces and artificial
  noise,'' \emph{arXiv: 2005.10062}, 2020.

\bibitem{J_Chen19IRS_Secure}
J.~{Chen}, Y.~{Liang}, Y.~{Pei}, and H.~{Guo}, ``Intelligent reflecting
  surface: A programmable wireless environment for physical layer security,''
  \emph{IEEE Access}, vol.~7, pp. 82\,599--82\,612, 2019.

\bibitem{J_Chu21IRSecure}
Z.~{Chu}, W.~{Hao}, P.~{Xiao}, and J.~{Shi}, ``Intelligent reflecting surface
  aided multi-antenna secure transmission,'' \emph{IEEE Wireless Commun.
  Lett.}, vol.~9, no.~1, pp. 108--112, 2020.

\bibitem{J_Chu21IRS_perfectCSI}
Z.~Chu, W.~Hao, P.~Xiao, D.~Mi, Z.~Liu, M.~Khalily, J.~R. Kelly, and A.~P.
  Feresidis, ``Secrecy rate optimization for intelligent reflecting surface
  assisted {MIMO} system,'' \emph{IEEE Trans. Inf. Forensics Security},
  vol.~16, pp. 1655--1669, 2021.

\bibitem{J_He20CascadeIFR}
Z.~{He} and X.~{Yuan}, ``Cascaded channel estimation for large intelligent
  metasurface assisted massive {MIMO},'' \emph{IEEE Wireless Commun. Lett.},
  vol.~9, no.~2, pp. 210--214, Feb. 2020.

\bibitem{J_YouZRui20_IRS}
C.~{You}, B.~{Zheng}, and R.~{Zhang}, ``Channel estimation and passive
  beamforming for intelligent reflecting surface: Discrete phase shift and
  progressive refinement,'' \emph{IEEE J. Sel. Areas Commun.}, vol.~38, no.~11,
  pp. 2604--2620, Nov. 2020.

\bibitem{J_Zheng18SecureEE}
T.~{Zheng}, H.~{Wang}, and J.~{Yuan}, ``Secure and energy-efficient
  transmissions in cache-enabled heterogeneous cellular networks: Performance
  analysis and optimization,'' \emph{IEEE Trans. Commun.}, vol.~66, no.~11, pp.
  5554--5567, Nov. 2018.

\bibitem{B_S_Prin02}
G.~Stuber, \emph{\BIBforeignlanguage{eng}{Principles of Mobile Communication}},
  2nd~ed.\hskip 1em plus 0.5em minus 0.4em\relax Boston, MA: Springer US, 2002.

\bibitem{J_Lipp16CCP}
T.~Lipp and S.~Boyd, ``Variations and extension of the convex--concave
  procedure,'' \emph{Optim. Eng.}, vol.~17, no.~2, pp. 263--287, Jun. 2016.

\bibitem{J_LuoSDR}
N.~D. {Sidiropoulos}, T.~N. {Davidson}, and Z.-Q. {Luo}, ``Transmit beamforming
  for physical-layer multicasting,'' \emph{IEEE Trans. Signal Process.},
  vol.~54, no.~6, pp. 2239--2251, Jun. 2006.

\bibitem{Ben-TalA01}
A.~Ben-Tal and A.~Nemirovski, \emph{\BIBforeignlanguage{eng}{Lectures on Modern
  Convex Optimization: Analysis, Algorithms, and Engineering
  Applications}}.\hskip 1em plus 0.5em minus 0.4em\relax Philadelphia, PA, USA:
  SIAM, 2001.

\bibitem{J_Anstre01Anbf}
K.~M. Anstreicher and N.~W. Brixius, ``\BIBforeignlanguage{eng}{A new bound for
  the quadratic assignment problem based on convex quadratic programming},''
  \emph{\BIBforeignlanguage{eng}{Math. Programming}}, vol.~89, no.~3, pp.
  341--357, 2001.

\bibitem{B_AbsilP09}
P.-A. Absil, R.~Mahony, and R.~Sepulchre,
  \emph{\BIBforeignlanguage{eng}{Optimization Algorithms on Matrix
  Manifolds}}.\hskip 1em plus 0.5em minus 0.4em\relax Princeton University
  Press, 2009.

\bibitem{J_Archiyou20IRS_EE}
L.~You, J.~Xiong, Y.~Huang, D.~W.~K. Ng, C.~Pan, W.~Wang, and X.~Gao,
  ``Reconfigurable intelligent surfaces-assisted multiuser {MIMO} uplink
  transmission with partial {CSI},'' \emph{arXiv: 2003.13014}, Mar. 2020.

\bibitem{J_ZZ20SecureProb}
Z.~{Li}, S.~{Wang}, P.~{Mu}, and Y.~{Wu}, ``Probabilistic constrained secure
  transmissions: Variable-rate design and performance analysis,'' \emph{IEEE
  Trans. Wireless Commun.}, vol.~19, no.~4, pp. 2543--2557, Apr. 2020.

\bibitem{J_Liu14Gao_outage}
X.~{Liu}, F.~{Gao}, G.~{Wang}, and X.~{Wang}, ``Joint beamforming and user
  selection in multicast downlink channel under secrecy-outage constraint,''
  \emph{IEEE Commun. Lett.}, vol.~18, no.~1, pp. 82--85, Jan. 2014.

\bibitem{J_Huang18IRSDis}
C.~{Huang}, G.~C. {Alexandropoulos}, A.~{Zappone}, M.~{Debbah}, and C.~{Yuen},
  ``Energy efficient multi-user {MISO} communication using low resolution large
  intelligent surfaces,'' in \emph{IEEE Globecom Workshops (GC Wkshps)}, 2018,
  pp. 1--6.

\bibitem{J_DongHuiM20Secure}
L.~{Dong} and H.~{Wang}, ``Secure {MIMO} transmission via intelligent
  reflecting surface,'' \emph{IEEE Wireless Commun. Lett.}, vol.~9, no.~6, pp.
  787--790, Jun. 2020.

\bibitem{J_Yu20IRS_SE}
X.~{Yu}, D.~{Xu}, Y.~{Sun}, D.~W.~K. {Ng}, and R.~{Schober}, ``Robust and
  secure wireless communications via intelligent reflecting surfaces,''
  \emph{IEEE J. Sel. Areas Commun.}, vol.~38, no.~11, pp. 2637--2652, 2020.

\bibitem{J_ZhangR_SecureIRS19}
M.~{Cui}, G.~{Zhang}, and R.~{Zhang}, ``Secure wireless communication via
  intelligent reflecting surface,'' \emph{IEEE Wireless Commun. Lett.}, vol.~8,
  no.~5, pp. 1410--1414, Oct. 2019.

\bibitem{J_20IRS_Secure}
X.~{Lu}, W.~{Yang}, X.~{Guan}, Q.~{Wu}, and Y.~{Cai}, ``Robust and secure
  beamforming for intelligent reflecting surface aided mmwave {MISO} systems,''
  \emph{IEEE Wireless Commun. Lett.}, vol.~9, no.~12, pp. 2068--2072, 2020.

\bibitem{J_Yuwei18FP}
K.~{Shen} and W.~{Yu}, ``Fractional programming for communication systems-part
  {I}: Power control and beamforming,'' \emph{IEEE Trans. Signal Process.},
  vol.~66, no.~10, pp. 2616--2630, May 2018.

\bibitem{Cov_Opt90}
S.~P. Boyd and L.~Vandenberghe, \emph{Convex Optimization}.\hskip 1em plus
  0.5em minus 0.4em\relax Cambridge, U.K.: Cambridge Univ. Press, 2004.

\bibitem{J_PA10IPM_Bok}
I.~Polik and T.~Terlaky, \emph{\BIBforeignlanguage{English}{Interior Point
  Methods for Nonlinear Optimization}}.\hskip 1em plus 0.5em minus 0.4em\relax
  Springer, 2010.

\bibitem{J_Abs09Accel_Amoji}
P.-A. Absil and K.~A. Gallivan, ``Accelerated line-search and trust-region
  methods,'' \emph{SIAM Journal on Numerical Analysis}, vol.~47, no.~2, pp.
  997--1018, 2009.

\bibitem{J_BeckAmir09AFIS}
A.~Beck and M.~Teboulle, ``\BIBforeignlanguage{eng}{A fast iterative
  shrinkage-thresholding algorithm for linear inverse problems},''
  \emph{\BIBforeignlanguage{eng}{SIAM J. Imag. Sci.}}, vol.~2, no.~1, pp.
  183--202, 2009.

\bibitem{B_YNesterov}
Y.~Nesterov, \emph{Introductory Lectures on Convex Optimization: A Basic Course
  (Applied Optimization)}.\hskip 1em plus 0.5em minus 0.4em\relax Springer,
  2004.

\bibitem{C_Xu19Retraction}
D.~{Xu}, X.~{Yu}, Y.~{Sun}, D.~W.~K. {Ng}, and R.~{Schober}, ``Resource
  allocation for secure {IRS}-assisted multiuser {MISO} systems,'' in
  \emph{Proc. IEEE Globecom Workshops}, Dec. 2019, pp. 1--6.

\bibitem{J_PanRCG_IRS20}
C.~{Pan}, H.~{Ren}, K.~{Wang}, W.~{Xu}, M.~{Elkashlan}, A.~{Nallanathan}, and
  L.~{Hanzo}, ``Multicell {MIMO} communications relying on intelligent
  reflecting surfaces,'' \emph{IEEE Trans. Wireless Commun.}, vol.~19, no.~8,
  pp. 5218--5233, 2020.

\bibitem{J_ReminCG_IRS21}
K.~{Feng}, X.~{Li}, Y.~{Han}, S.~{Jin}, and Y.~{Chen}, ``Physical layer
  security enhancement exploiting intelligent reflecting surface,'' \emph{IEEE
  Commun. Lett.}, vol.~25, no.~3, pp. 734--738, 2021.

\bibitem{J_vishnoi2018geodesic}
N.~K. Vishnoi, ``Geodesic convex optimization: Differentiation on manifolds,
  geodesics, and convexity,'' \emph{arXiv: 1806.06373}, Jun. 2018.

\bibitem{C_Accele_remani18Alg}
H.~Zhang and S.~Sra, ``An estimate sequence for geodesically convex
  optimization,'' in \emph{Proceedings of the Conference On Learning Theory},
  ser. Proceedings of Machine Learning Research, vol.~75.\hskip 1em plus 0.5em
  minus 0.4em\relax PMLR, 06--09 Jul 2018, pp. 1703--1723.

\bibitem{J_HuangWen2021Rpgm}
W.~Huang and K.~Wei, ``\BIBforeignlanguage{eng}{Riemannian proximal gradient
  methods},'' \emph{\BIBforeignlanguage{eng}{Mathematical programming}}, 2021.

\bibitem{C_pmlrv49zhang16b}
H.~Zhang and S.~Sra, ``First-order methods for geodesically convex
  optimization,'' in \emph{29th Annual Conference on Learning Theory}, ser.
  Proceedings of Machine Learning Research, vol.~49, Jun. 2016, pp. 1617--1638.

\bibitem{B_Bertseka97NP}
D.~P. Bertsekas, \emph{\BIBforeignlanguage{eng}{Nonlinear Programming}},
  2nd~ed.\hskip 1em plus 0.5em minus 0.4em\relax Belmont, MA, USA: Athena
  Scientific, 1999.

\bibitem{J_Andes95_Prog}
J.~B. {Andersen}, T.~S. {Rappaport}, and S.~{Yoshida}, ``Propagation
  measurements and models for wireless communications channels,'' \emph{IEEE
  Commun. Mag.}, vol.~33, no.~1, pp. 42--49, 1995.

\bibitem{J_Kanduk02SOP}
S.~{Kandukuri} and S.~{Boyd}, ``Optimal power control in interference-limited
  fading wireless channels with outage-probability specifications,'' \emph{IEEE
  Trans. Wireless Commun.}, vol.~1, no.~1, pp. 46--55, Jan. 2002.

\bibitem{J_Tian16Equality}
Y.~Tian, ``\BIBforeignlanguage{eng}{Equalities and inequalities for ranks of
  products of generalized inverses of two matrices and their applications},''
  \emph{\BIBforeignlanguage{eng}{Journal of Inequalities and Applications}},
  vol. 2016, no.~1, pp. 1--51, 2016.

\bibitem{B_LAaM14Banj}
S.~Banerjee and A.~Roy, \emph{\BIBforeignlanguage{eng}{Linear Algebra and
  Matrix Analysis for Statistics}}.\hskip 1em plus 0.5em minus 0.4em\relax New
  York: Chapman and Hall/CRC, 2014.

\bibitem{J_Ali17ConvexAux}
A.~A. Nasir, H.~D. Tuan, T.~Q. Duong, and H.~V. Poor, ``Secrecy rate
  beamforming for multicell networks with information and energy harvesting,''
  \emph{IEEE Trans. Signal Process.}, vol.~65, no.~3, pp. 677--689, 2017.

\bibitem{J_Marks78AGIA}
B.~Marks and G.~Wright, ``\BIBforeignlanguage{eng}{A general inner
  approximation algorithm for nonconvex mathematical programs},''
  \emph{\BIBforeignlanguage{eng}{Operations Research}}, vol.~26, no.~4, p. 681,
  1978.

\bibitem{J_Liu09Saddle}
Q.~Liu, W.~M. Tang, and X.~M. Yang, ``\BIBforeignlanguage{English}{Properties
  of saddle points for generalized augmented lagrangian},''
  \emph{\BIBforeignlanguage{English}{Mathematical Methods of Operations
  Research}}, vol.~69, no.~1, pp. 111--124, 03 2009.

\bibitem{J_Trendafil99OM}
N.~T. Trendafilov, ``A continuous-time approach to the oblique procrustes
  problem,'' \emph{Behaviormetrika}, vol.~26, no.~2, pp. 167--181, 1999.

\bibitem{B_BartleRobert11}
R.~G. Bartle, \emph{\BIBforeignlanguage{eng}{Introduction to Real Analysis}},
  4th~ed.\hskip 1em plus 0.5em minus 0.4em\relax Hoboken, NJ: Wiley, 2011.

\bibitem{J_Bolte14ConverAnalysis}
J.~Bolte, S.~Sabach, and M.~Teboulle, ``Proximal alternating linearized
  minimization for nonconvex and nonsmooth problems,'' \emph{Math. Program.},
  vol. 146, no. 1-2, pp. 459--494, Aug. 2014.

\bibitem{J_Semi_con01Kiwi}
K.~C. Kiwiel, ``Convergence and efficiency of subgradient methods for
  quasiconvex minimization.'' \emph{Math. Programming}, vol.~90, no.~1, p.~1,
  2001.

\bibitem{B_MordukhovichBorisS12}
B.~S. Mordukhovich, \emph{\BIBforeignlanguage{eng}{Variational Analysis and
  Generalized Differentiation I: Basic Theory}}.\hskip 1em plus 0.5em minus
  0.4em\relax Berlin, Heidelberg: Springer Berlin, 2012.

\end{thebibliography}

\end{document}